\documentclass[a4paper, 10pt]{article}
\usepackage[T1]{fontenc}			
\usepackage[english]{babel}
\usepackage{array}						
\usepackage{graphics}					
\usepackage{indentfirst} 			
\usepackage{vmargin}					
\usepackage{amsmath,amsthm,amssymb,stmaryrd,wasysym}
\usepackage{color,colortbl}
\usepackage[table]{xcolor}
\usepackage{stackrel}
\usepackage{mathrsfs}
\usepackage{fancyref}
\usepackage{graphicx}
\usepackage{enumerate}
\usepackage{mdwlist}
\usepackage{caption}
\usepackage{fancybox}
\usepackage{pifont}
\usepackage{makecell}
\usepackage{epigraph}
\usepackage{pifont}
\usepackage{dsfont}
\usepackage{relsize}
\usepackage{rotating}
\usepackage{multirow}
\newcolumntype{?}{!{\vrule width 1pt}}
\usepackage{csquotes}
\usepackage{chngcntr}
\usepackage{tikz-cd}
\usepackage{changepage}
\usepackage{todonotes}
\usepackage{mathtools}
\usepackage{stmaryrd}

\numberwithin{equation}{section}
\counterwithin*{equation}{section}

\input amssym.def

\usepackage{fancyhdr}
\makeindex

\def\bqn{\begin{eqnarray}}
\def\eqn{\end{eqnarray}}
\newcommand*\xbar[1]{
  \hbox{
    \vbox{
      \hrule height 0.5pt 
      \kern0.5ex         
      \hbox{
        \kern-0.25em      
        \ensuremath{#1}
        \kern-0.25em
      }
    }
  }
}

\newtheorem{Theorem}{Theorem}[section]
\newtheorem*{Theorem*}{Theorem}
\newtheorem{Corollary}[Theorem]{Corollary}
\newtheorem{Lemma}[Theorem]{Lemma}
\newtheorem{Proposition}[Theorem]{Proposition}

 { \theoremstyle{definition}
\newtheorem{Definition}[Theorem]{Definition}

\newtheorem{Example}[Theorem]{Example}
\newtheorem{Remark}[Theorem]{Remark}
}

\usepackage{amsmath}
\usepackage{empheq}
\usepackage[most]{tcolorbox}

\newtcbox{\mymath}[1][]{    nobeforeafter, math upper, tcbox raise base,
    enhanced, colframe=blue!30!black,
       colback=blue!30!red!30!white, boxrule=1pt,   
    #1}

\usepackage[	
			colorlinks=true,	
			pdfstartview=FitV,
			linkcolor= color3,
			citecolor= lightblue,
			urlcolor= color3,
			hyperindex=true,
			hyperfigures=false,unicode]
                      			{hyperref}
			
\hypersetup{
           breaklinks=true,   
           colorlinks=true,   
           pdfusetitle=true,  
        }			

\newcommand{\bea}{\begin{eqnarray}}
\newcommand{\eea}{\end{eqnarray}}
\newcommand{\bi}{\begin{itemize}}
\newcommand{\ei}{\end{itemize}}
\newcommand{\be}{\begin{enumerate}}
\newcommand{\ee}{\end{enumerate}}

\newcommand{\nn}{\nonumber}

\newcommand\qq{{\quad,\quad}}

\newcommand{\From}{\colon}

\newcommand{\p}{\partial}

\newcommand{\Lag}{\mathcal{L}}

\newcommand{\om}{\omega}
\newcommand\form[1]{\Omega^1\pl #1\pr}

\newcommand{\ff}{\form{\M}}

\newcommand{\vfh}{\Gamma\big(T\M\big)}

\newcommand{\un}{^{-1}}
\newcommand{\Id}{I}

\newcommand{\tr}{\text{tr}}

\newcommand{\Ker}{\text{\rm Ker }}

\newcommand{\End}{\text{\rm End}}

\newcommand{\Aut}{\text{\rm Aut}}

\newcommand\Span[1]{\text{\rm Span}\hspace{1.2mm} #1}

\newcommand\Spannn[1]{\text{\rm Span}\left\lbrace #1\right\rbrace}

\newcommand{\M}{\mathscr{M}}

\newcommand{\pl}{\left(}
\newcommand{\pr}{\right)}

\newcommand{\br}[2]{\big[ #1,#2 \big]}

\newcommand{\tK}{\textbf{K}}
\newcommand{\tP}{\textbf{P}}
\newcommand{\tJ}{\textbf{J}}
\newcommand{\tD}{\textbf{D}}
\newcommand{\tS}{{S}}
\newcommand{\tV}{\textbf{V}}
\newcommand{\tX}{\textbf{X}}
\def\tSs#1{\tS^{(#1)}}
\def\tVs#1{\tV^{(#1)}}

\newcommand{\so}{\mathfrak{so}}

\newcommand{\ad}{\text{\rm ad}}

\newcommand{\alg}{\mathfrak g}
\newcommand{\alh}{\mathfrak h}
\newcommand{\aln}{\mathfrak n}
\newcommand{\alk}{\mathfrak k}
\newcommand{\alp}{\mathfrak p}
\newcommand{\ali}{\mathfrak i}
\newcommand{\alj}{\mathfrak j}
\newcommand{\alt}{\mathfrak t}

\newcommand{\IWun}{\alk}
\newcommand{\IWdeux}{\ali}

\newcommand{\gh}{{\alg/\alh}}

\newcommand{\bs}{\boldsymbol}
\newcommand{\GL}{\text{\rm GL}}
\newcommand\GLR[1]{\GL(#1|\mathbb R)}
\newcommand{\SL}{\text{\rm SL}}
\newcommand\SLR[1]{\SL(#1|\mathbb R)}
\newcommand{\PSL}{\text{\rm PSL}}
\newcommand\PSLR[1]{\PSL(#1|\mathbb R)}
\newcommand{\PGL}{\text{\rm PGL}}
\newcommand\PGLR[1]{\PGL(#1|\mathbb R)}
\newcommand{\SO}{\text{\rm SO}}

\newcommand{\GroupR}{{\mathbb R^\times}}
\newcommand{\T}{\text{\rm T}}
\newcommand\TR[1]{\T(#1|\mathbb R)}
\newcommand{\mirroredinplus}{\reflectbox{$\inplus$}}

\newcommand{\ie}{\textit{i.e.}~}

\newcommand{\cf}{\textit{cf.}~}

\newcommand{\eg}{\textit{e.g.}~}

\newcommand{\etal}{\textit{et al.}~}
\newcommand{\vs}{\textit{vs.}~}
\newcommand{\ala}{\textit{\`a la}~}
\newcommand{\etc}{\textit{etc.}~}

\newcommand\pset[1]{\left \lbrace #1\right \rbrace}

\newcommand\Prop[1]{Proposition \ref{#1}}
\newcommand\Defi[1]{Definition \ref{#1}}

\newcommand{\Milnegh}{\fhpsi}

\newcommand\fhpsi{\Gamma(\Ker\psi)}
\newcommand\IW{\.{I}n\"{o}n\"{u}--Wigner }

\definecolor{rougef}{rgb}{0.56,0,0}	
	
\definecolor{LightCyan}{rgb}{0.88,1,1}
\definecolor{Gray}{gray}{0.9}
\definecolor{alizarin}{RGB}{234, 182, 118}
\definecolor{Blue}{RGB}{63, 118, 135}
\definecolor{WarmGray}{RGB}{188, 186, 190}
\definecolor{Overcast}{RGB}{241, 241, 242}
\definecolor{GlacierBlue}{RGB}{25, 149, 173}
\definecolor{Ice}{RGB}{161, 214, 226}
\definecolor{LightIce}{RGB}{217, 255, 255}
\definecolor{lightblue}{rgb}{0.22,0.45,0.70}
\definecolor{mygreen}{rgb}{0.04, 0.76, 0.53}
\definecolor{color1}{rgb}{0.255,0.53, 0.53} \definecolor{linkcolor}{RGB}{ 6, 57, 112}
  \definecolor{linkcolor}{RGB}{ 6, 57, 112}
  \definecolor{color2}{RGB}{  112, 6, 57}
  \definecolor{color3}{RGB}{  6, 99, 112}

\newcommand\btitle[1]{\newblock ``\textit{{#1}}''}

\newcommand\barxiv[1]{\newblock \href{http://arXiv.org/abs/#1}{\texttt{arXiv:#1}}}

\newcommand{\UnskipRef}[1]{\unskip~\textnormal{[\ref{#1}]}}
\newcommand{\UnskipRefs}[2]{\unskip~\textnormal{[\ref{#1} and \ref{#2}]}}
\newcommand{\SectionRef}[1]{\unskip~[\hyperref[#1]{\S\thinspace\ref{#1}}]}
\newcounter{rowcntr}[table]
\renewcommand{\therowcntr}{\thetable.\arabic{rowcntr}}

\newcommand{\mybox}[3]{\begin{tcolorbox}[ams align, colback=#1!25!white,colframe=#1,title={#2}]
  {#3} \nn
\end{tcolorbox}}

\newcolumntype{N}{>{\refstepcounter{rowcntr}\therowcntr}c}

\AtBeginEnvironment{tabular}{\setcounter{rowcntr}{0}}

\begin{document}
\tcbset{highlight math style={colback=blue!30!red!30!white}}
\thispagestyle{empty}

\vspace{1cm}

 \begin{centering}

{\large {\bfseries 
Possible ambient kinematics
}
}
\\\vspace{0.5cm}
Kevin Morand
\\
\vspace{0.5cm}
Department of Physics, Sogang University, Seoul  04107, South Korea\\
Center for Quantum Spacetime, Sogang University, Seoul  04107, South Korea\\
\vspace{0.5cm}
{\tt morand@sogang.ac.kr}

\vspace{0.5cm}

\end{centering}
\begin{abstract}
\medskip

\centering\begin{minipage}{\dimexpr\paperwidth-6.5cm}
\noindent
In a seminal paper, Bacry and L\'evy--Leblond classified kinematical algebras, a class of Lie algebras encoding the symmetries of spacetime.
Homogeneous spacetimes (infinitesimally, Klein pairs) associated with these possible kinematics can be partitioned into four families---riemannian, lorentzian, galilean and carrollian---based on
the type of invariant metric structure they admit. 
In this work, we classify possible ambient kinematics---defined as extensions of kinematical algebras by a scalar ideal---as well as their associated Klein pairs. Kinematical Klein pairs arising as 
 quotient space along the extra scalar ideal are said to admit a lift into the corresponding ambient Klein pair.
While all non-galilean Klein pairs admit a unique---trivial and torsionfree---higher-dimensional lift, galilean Klein pairs are constructively shown to admit lifts into two distinct families of ambient Klein pairs. The first family includes the bargmann algebra as well as its curved/torsional avatars while the second family is novel and generically allows lifts into torsional ambient spaces. 
We further comment on the relation between these two families and the maximally symmetric family of leibnizian Klein pairs.
\end{minipage}
\end{abstract}
\pagenumbering{gobble}
\tableofcontents

\newpage
\setcounter{secnumdepth}{0}
\section{Introduction}
\label{sectionintro}
\setcounter{secnumdepth}{3}

\pagenumbering{arabic}

On November $3^\text{rd}$ 1823, the Hungarian mathematician J\'anos Bolyai wrote---in a letter addressed to his father---about his latest discovery in the following terms \cite{Halsted}: 
`\emph{From nothing, I have created another entirely new world}'. As is known today, this `new world' coincides with the one of hyperbolic geometry---discovered separately and independently by the Russian mathematician Nikolai Lobachevsky---first occurrence of a \emph{non-euclidean} geometry in a universe hitherto dominated by Euclid's magisteria. The paradigm shift brought about by the advent of hyperbolic geometry hence challenged the millennial monopoly held by euclidean space as sole candidate to represent physical space.\footnote{Inasmuch that the philosopher Immanuel Kant, distancing himself from his previously held leibnizian views from the pre-critical period, endowed euclidean space---in his \emph{Critique of Pure Reason} (1781)---with the status of paradigmatic example of {\it a priori} synthetic knowledge, deeming it to be an unavoidable necessity of thought. Note that this view was famously challenged in \cite{Borges} where it was contradistinctly argued that `$[\ldots]$\emph{ hexagonal rooms are the necessary shape of absolute space, or at least of our perception of space}'. } A natural question that emerges from these considerations is then the leibnizian interrogation \cite{Leibniz}:
\vspace{-3mm}
\begin{center}
\textit{
What are the possible worlds?
}
\end{center}
The few decades following the wake of hyperbolic geometry\footnote{The emergence of these new geometries caused a stir in European intellectual life beyond mathematical circles, raising in particular important debates in the Victorian era regarding the pedagogy of geometry that was heretofore predicated on Euclid's \emph{Elements}. Among the protagonists of this controversy features Charles Lutwidge Dodgson---\emph{alias} Lewis Carroll---who penned a defence of the \emph{Elements} in the form of a play wherein Euclid's ghost argues against his `modern rivals' \cite{Carroll1879}.} witnessed the flourishing of other non-euclidean geometries, such as elliptic, affine and projective geometries. Faced with this newly found profusion of geometries, Felix Klein formulated in 1872 a vast synthesis \cite{Klein1872} based on a renewed conception of geometry envisaged as the study of invariant properties of space under a group of transformations. The proposed emphasis on the role played by the underlying symmetry groups in the classification and relation between different geometries constituted the initial impetus for H.~Bacry and J.~M.~L\'evy--Leblond \cite{Bacry1968} to rephrase the above question in mathematical terms. Explicitly, this was done by defining the notion of \emph{kinematical algebras}---\ie Lie algebras encoding the infinitesimal kinematical symmetries of any free particle---such that each of the associated generators finds its origin in a physical property of spacetime (in $d+1$ dimensions), according to the following table:
\begin{table}[ht]
\centering
\resizebox{12cm}{!}{
\begin{tabular}{c|c|c|c}
\rule[-0.2cm]{0cm}{0.6cm}{\bf Physical requirement}&{\bf Generator}&{\bf Transformation}&{\bf Dimension}\\
\hline
Time homogeneity&$H$&Time translation&1\\
\rowcolor{Gray}
Space homogeneity&$\tP$&Spatial translations&$d$\\
Relativity principle&$\tK$&Inertial boosts&$d$\\
\rowcolor{Gray}
Space isotropy&$\tJ$&Spatial rotations&$\frac{d\pl d-1\pr}{2}$
  \end{tabular}
  }
\end{table}

Remarkably, imposing the Jacobi identity highly constrains the space of possible kinematical algebras, allowing for their explicit classification \cite{Bacry1968,Bacry1986}. Such classification was recently revisited in arbitrary dimension in a series of works \cite{Figueroa-OFarrill2017c,Figueroa-OFarrill2017a,Andrzejewski2018} due to J.~Figueroa--O'Farrill and collaborators [\cf \SectionRef{section:Beginning at the beginning: kinematical algebras} below as well as \cite{FigueroaOFarrill2017} for a summary]. These authors also addressed the classification of the corresponding Klein geometries \cite{Figueroa-OFarrill2018,Figueroa-OFarrill2019}, thereby offering an exhaustive and panoramic view on the landscape of possible kinematics. In the spirit of Klein's Erlangen program, the kinematical Klein geometries delineated in these works can be partitioned according to the type of invariant structures living on the associated homogeneous spacetimes, thus enabling a refined partition into four distinct categories, namely, \emph{riemannian}, \emph{lorentzian}, \emph{galilean} and \emph{carrollian}. Whereas the former two families exhibit an invariant non-degenerate bilinear form [of signature $(0,d+1)$ and $(1,d)$, respectively], the galilean and carrollian homogeneous geometries are characterised by the existence of a \emph{degenerate} structure. The paradigmatic example of such degenerate Klein geometries is given by the flat galilei spacetime, the latter being foundational to Newtonian mechanics, wherein time and space are independent and absolute. The underlying galilei algebra also plays a crucial role in nonrelativistic quantum mechanics, whereby the symmetries of a physical system are projective unitary representations of the galilei group, or equivalently ordinary representations of the central extension of the galilei group, known as the bargmann group  \cite{Bargmann1952}. 

The significance of the bargmann group and its corresponding algebra in relation to nonrelativistic physics and the covariant formulation thereof has been emphasised in \cite{Duval1977,Duval1982a}, and more recently revisited in \cite{Andringa:2010it}. In addition to the generators listed in the above table, the bargmann algebra is characterised by an extra, central, generator $M$, conventionally interpreted `internally' in relation to the inertial mass of the system. However, an alternative interpretation---put forward in the works of L.~P.~Eisenhart \cite{Eisenhart1928} and independently elaborated on by C.~Duval and collaborators \cite{Duval1985,Duval1991}---construes this extra generator as encoding the translational symmetry for an additional spacetime direction. In this interpretation, the homogeneous space modelled on the bargmann Klein pair is of dimension $d+2$ and characterised by an invariant lorentzian metric structure---isomorphic to the minkowski metric---together with an invariant lightlike vector field, thereby forming a \emph{bargmannian} metric structure \cite{Duval1985,Duval1991}. Upon quotienting along the lightlike direction, one recovers the flat galilei spacetime in $d+1$ dimensions, the latter being said to admit a \emph{lift} into the \emph{ambient} bargmannian Klein geometry. Following these pioneering works, other lifts of homogeneous galilean spacetimes within the bargmannian framework have been proposed: in \cite{Gibbons:2003rv}, the curved homogeneous spaces modelled on the \emph{euclidean/lorentzian newton} algebras\footnote{
The euclidean/lorentzian newton algebras and their corresponding spacetimes are referred to by various names in the literature, most notably as \emph{Newton--Hooke} spacetimes, \emph{non-relativistic cosmological} spacetimes, or \emph{galilean (anti) de Sitter} spacetimes. The bivalent terminology has the advantage to reflect the existence of two possible signs for the (non-vanishing) associated `cosmological constant'.
} have been shown to admit a lift into Hpp-waves modelled on their respective central extensions; more recently, bargmannian Klein pairs allowing to lift the torsional galilean Klein pairs classified in \cite{Figueroa-OFarrill2018} \big[see also \cite{FigueroaOFarrill2022d}\big] have been exhibited and studied in \cite{FigueroaOFarrill2022,FigueroaOFarrill2022e}. As an alternative to the bargmannian framework, a new class of possible ambient geometries---dubbed \emph{leibnizian}---was introduced in \cite{Bekaert2015b} \big[\cf also the recent work \cite{FigueroaOFarrill2022f}\big] and argued to provide a minimal ambient unification of galilean and carrollian geometries. The paradigmatic example thereof---called \emph{leibniz} algebra\footnote{Note that the leibniz algebra introduced in \cite{Bekaert2015b} is a Lie algebra and as such should not be confused with the unrelated notion of Leibniz algebras \cite{Loday1993} which refers to vector spaces endowed with a  bilinear but non-necessarily skewsymmetric bracket satisfying the Leibniz identity. }---is maximal\footnote{That is, the leibniz algebra is of dimension $\frac{(d+2)(d+3)}{2}$ in $(d+2)$-spacetime dimension.} in spacetime dimension $d+2$ and contains the bargmann algebra as subalgebra. 
This recent flurry of ambient spaces naturally prompts the following variation:
\begin{center}
\textit{
What are the possible ambient worlds?
}
\end{center}
The object of the present work\footnote{Together with its companion paper \cite{Morand2023}.} will be to cast this problem in mathematical terms by presenting a notion of possible ambient kinematics, which we shall henceforth classify both at the level of algebras and associated Klein pairs. Such classification will see the advent of new families of ambient Klein pairs, thus disputing bargmann's long-held magisteria.

\paragraph{A first example}
\renewcommand{\theequation}{\roman{equation}}
\hfill

\medskip

Let us consider by way of illustration the \emph{euclidean/lorentzian newton spacetime} defined as the $(d+1)$-dimensional spacetime $\bar\M\simeq \mathbb R^{d+1}$ endowed with the following structure of galilean manifold\footnote{Recall that a \emph{galilean structure} \cite{Kuenzle:1972zw} is a triplet $(\bar\M,\bar\psi,\bar\gamma)$ where:
\bi
\item $\bar{\M}$ is a manifold of dimension $d+1$.
\item $\bar{\psi}$ is a nowhere vanishing 1-form.
\item $\bar\gamma$ is a {positive definite} covariant metric of rank $d$ acting on ${\Gamma(\Ker\bar\psi)}$.
\ei
A galilean structure endowed with a compatible connection $\bar\nabla$ is called a \emph{galilean manifold} [\cf \eg \cite{Morand2018} for a review].
}:
\bea
\label{equation:euclidean/lorentzian newton structure}
 \begin{tikzcd}[row sep = small,column sep = small,ampersand replacement=\&]
\bar\psi=dt\& \bar\gamma=\delta_{ij}\, dx^i\otimes dx^j\&\bar\Gamma^i_{tt}=\epsilon\, x^i
\end{tikzcd}
\eea
where $\epsilon=1$ (resp. $\epsilon=-1$) in the euclidean (resp. lorentzian) case. The above spacetime is maximally symmetric in dimension $d+1$, \ie the associated isometry algebra\footnote{Here and in the following, isometry algebras are assumed to preserve every ingredient making up the manifold [including the connection thus ensuring the algebra to be finite-dimensional].}
$\alg_0$ is of maximal dimension $\frac{(d+1)(d+2)}{2}$ and isomorphic to the kinematical algebra known as the \emph{euclidean/lorentzian newton algebra}, with non-trivial commutators\footnote{The commutators involving the $\so(d)$-generators $\tJ$ are fixed by definition of kinematical algebras (and generalisations) and are therefore omitted in this introductory section. }:
\bea
\label{equation:euclidean/lorentzian newton algebra}
 \begin{tikzcd}[row sep = small,column sep = small,ampersand replacement=\&]
\br{\tK}{H}=\tP\&\br{H}{\tP}=\epsilon\, \tK\, .
\end{tikzcd}
\eea
The Duval--Eisenhart lift of the euclidean/lorentzian newtonian spacetime \eqref{equation:euclidean/lorentzian newton structure} has been investigated in \cite{Gibbons:2003rv} where it was found to arise as the quotient manifold of the following $(d+2)$-dimensional \emph{bargmannian manifold}\footnote{A \emph{bargmannian structure}  \cite{Duval1985,Duval1991}  is a triplet $(\M,\xi,g)$ where:
\bi
\item $\M$ is a manifold of dimension $d+2$.
\item $\xi$ is a nowhere vanishing vector field. 
\item $g$ is a lorentzian metric such that $\xi$ is lightlike with respect to $g$.
\ei
A bargmannian structure endowed with a compatible connection
is called a  \emph{bargmannian manifold}. Whenever the connection is torsionfree, the latter coincides with the Levi--Civita connection associated with the metric $g$. 
}:
\bea
\label{equation:euclidean/lorentzian bargmannian structure}
 \begin{tikzcd}[row sep = small,column sep = small,ampersand replacement=\&]
{\xi}=\p_u\& g=-\epsilon\, \textbf{x}^2\, dt\otimes dt+ du\otimes dt+dt\otimes du +\delta_{ij}\, dx^i\otimes dx^j
\end{tikzcd}
\eea
corresponding to a Hpp-wave endowed with a privileged lightlike parallel direction. The corresponding isometry algebra
reproduces the unique (non-trivial) central extension $\alg=\alg_0\inplus\mathbb R$ of the Lie algebra \eqref{equation:euclidean/lorentzian newton algebra} by an extra generator $M$ and is characterised by the additional commutator:
\bea
\label{equation:euclidean/lorentzian bargmannian algebra}
 \begin{tikzcd}[row sep = small,column sep = small,ampersand replacement=\&]
\br{\tK}{\tP}=M\, .
\end{tikzcd}
\eea
Crucial for the present analysis is the fact that the additional generator spans a scalar ideal $\ali=\Span M$ of the extended algebra $\alg$. The corresponding quotient is thus a well-defined Lie algebra being isomorphic to the euclidean/lorentzian newton algebra, \ie  $\alg_0=\alg/\ali$. This fact can be seen as the algebraic rationale underlying the geometric fact that the lightlike dimensional reduction of the above bargmannian geometry is well-defined and yields the euclidean/lorentzian newtonian spacetime \eqref{equation:euclidean/lorentzian newton structure} as quotient $\bar\M=\M/\mathbb R$ [\cf \eg \cite{Morand2018} for details]. 
\medskip

Note that the Duval--Eisenhart construction presupposes the ambient geometry to be bargmannian. An alternative, more minimal, embedding scheme has been put forward in \cite{Bekaert2015b} where the corresponding ambient geometry is assumed to be a \emph{leibnizian manifold}\footnote{\label{footnote:leibnizian manifold}
A \emph{leibnizian structure} \cite{Bekaert2015b} is a quadruplet $(\M,\xi,\psi,\gamma)$ where:
\bi
\item ${\M}$ is a manifold of dimension $d+2$.
\item $\xi$ is a nowhere vanishing vector field. 
\item $\psi$ is a nowhere vanishing 1-form annihilating ${\xi}$.
\item $\gamma$ is a {positive semi-definite} covariant metric of rank $d$ acting on $\Milnegh$ whose radical is spanned by ${\xi}$. 
\ei
A leibnizian structure endowed with a compatible connection
is called a \emph{leibnizian manifold}. Note that any bargmannian manifold induces a leibnizian manifold by defining $\psi=g(\xi,-)$ and $\gamma=g|_{\Ker\psi}$. For example, the leibnizian manifold associated with the bargmannian manifold \eqref{equation:euclidean/lorentzian bargmannian structure} reads explicitly:
\bea
\label{equation:induced leibnizian structure}
 \begin{tikzcd}[row sep = small,column sep = small,ampersand replacement=\&]
{\xi}=\p_u\& \psi=dt\&\gamma=\delta_{ij}\, dx^i\otimes dx^j\&\Gamma^u_{ti}=\Gamma^u_{it}=-\epsilon\, x^i\&\Gamma^i_{tt}=\epsilon\, x^i\nn
\end{tikzcd}
\eea
and can be checked to admit a well-defined projection isomorphic to the galilean manifold \eqref{equation:euclidean/lorentzian newton structure}. 
}. Specifically, defining the leibnizian manifold $\M\simeq\mathbb R^{d+2}$ endowed with:
\bea
\label{equation:leibnizian structure}
 \begin{tikzcd}[row sep = small,column sep = small,ampersand replacement=\&]
{\xi}=\p_u\& \psi=dt\&\gamma=\delta_{ij}\, dx^i\otimes dx^j\&\Gamma^u_{tt}=\epsilon\, u\&\Gamma^i_{tt}=\epsilon\, x^i\, ,
\end{tikzcd}
\eea
the quotient manifold $\M/\mathbb R$ [\ie the space of integral curves of $\xi$] can be checked to be a well-defined galilean manifold isomorphic to the above euclidean/lorentzian newtonian spacetime \eqref{equation:euclidean/lorentzian newton structure}. Furthermore, the isometry algebra of \eqref{equation:leibnizian structure}
is of maximal dimension\footnote{
Explicitly, the isometry algebra $\alg$ has underlying vector space $\alg=M\oplus H\oplus C\oplus \tP\oplus\tD\oplus\tK\oplus\tJ$ and is thus of maximal dimension $(d+2)(d+3)/2$, in contradistinction with the centrally extended euclidean/lorentzian newtonian algebra $\alg=M\oplus H\oplus \tP\oplus\tK\oplus\tJ$ being of non-maximal dimension $(d+1)(d+2)/2+1$. 
}
 in dimension $d+2$ and characterised by the following non-trivial commutators:
 \bea
 \label{equation:euclidean/lorentzian leibnizian algebra}
 \begin{tikzcd}[row sep = small,column sep = small,ampersand replacement=\&]
\br{\tD}{\tK}= C\&\br{\tK}{H}=\tP\&\br{\tD}{\tP}= M\&\br{C}{H}= M\&\br{H}{\tP}=\epsilon\, \tK\&\br{H}{M}=\epsilon\, C\, .
\end{tikzcd}
 \eea
 
Defining the abelian ideal $\ali= M\oplus C\oplus\tD$, the quotient $\alg/\ali$ is again isomorphic to the euclidean/lorentzian newton algebra $\alg_0$, thus providing an \emph{a posteriori} justification of the fact that the geometric projection is well-defined. 

\medskip


The two constructions mentioned above serve as examples of lifts of our original homogeneous galilean manifold into higher-dimensional homogeneous geometries.\footnote{Note that the two above constructions are independent in the following two (roughly equivalent) senses:
\be
\item The euclidean/lorentzian leibnizian algebra \eqref{equation:euclidean/lorentzian leibnizian algebra} does not contain the euclidean/lorentzian bargmannian algebra \eqref{equation:euclidean/lorentzian newton algebra}+\eqref{equation:euclidean/lorentzian bargmannian algebra} [corresponding to the central extension of the  euclidean/lorentzian newton algebra \eqref{equation:euclidean/lorentzian newton algebra}] as a subalgebra. 
\item The leibnizian manifold underlying \big[see footnote \ref{footnote:leibnizian manifold} as well as \cite{Bekaert2015b,Morand2023} for details\big] the above euclidean/lorentzian bargmannian manifold \eqref{equation:euclidean/lorentzian bargmannian structure} is not isomorphic to the above euclidean/lorentzian leibnizian manifold \eqref{equation:leibnizian structure}. 
\ee
} The objective of this series of works is to investigate and, where possible, classify such $(d+2)$-dimensional homogeneous ambient geometries capable of lifting $(d+1)$-dimensional homogeneous kinematical geometries. To achieve this, we reduce the problem to an algebraic one, in the spirit of the original work \cite{Bacry1968} which dealt with $(d+1)$-dimensional kinematical geometries.
\medskip

Specifically, our focus lies in classifying geometrically realisable `ambient' Klein pairs that admit kinematical Klein pairs as quotients (such ambient Klein pairs will hereafter be referred to as \emph{projectable triplets} \UnskipRef{Definition:Projectable triplet}). These ambient Klein pairs come in two flavors, depending on the assumed dimensionality of the corresponding algebras (\eg bargmann \vs leibniz). Importantly, we refrain from making assumptions about the underlying geometric structure, our main aim being precisely to classify these geometries.

\pagebreak
Below is an informal list of our main results, as illustrated in terms of the above example:

\be
\item We classify projective triplets having same dimensionality as the centrally extended euclidean/lorentzian newton algebra (or equivalently, as the bargmann algebra). The classified Klein pairs are compiled in Table \ref{Table:Projectable ambient triplets} which constitutes the main result of this first volume. Remarkably, every (effective) galilean Klein pair, as classified in \cite[Table 1]{Figueroa-OFarrill2018} [see also Table \ref{Table:Effective kinematical Klein pairs} below] is shown to admit a lift into two distinct geometric types [\cf Table \ref{Table:Possible lifts of effective galilean Klein pairs}]. 

Coming back to our example, the euclidean/lorentzian newton algebra is shown to admit:
\be
\item a unique bargmannian lift into its central extension \eqref{equation:euclidean/lorentzian newton algebra}+\eqref{equation:euclidean/lorentzian bargmannian algebra}, realised geometrically as the isometry algebra of the bargmannian manifold \eqref{equation:euclidean/lorentzian bargmannian structure}. 
\item a 1-parameter family of lifts into a novel ambient class, dubbed \emph{$\mathsf{G}$-ambient}, with commutators: 
 \bea
 \label{equation:euclidean/lorentzian G-ambient}
 \begin{tikzcd}[row sep = small,column sep = small,ampersand replacement=\&]
\br{\tK}{H}=\tP\&\br{H}{\tP}=\epsilon\, \tK\&\br{H}{M}=\lambda M
\end{tikzcd}
 \eea
 where the free parameter $\lambda\in\mathbb R$ encodes the arbitrariness in the choice of ambient torsion. This new family of $\mathsf{G}$-ambient algebras can be geometrically realised as isometry algebras of leibnizian manifolds endowed with a privileged parallel \emph{Ehresmann connection} \big[\ie a 1-form $A\in\ff$ satisfying $A(\xi)=1$ and $\nabla A=0$\big]. We refer to the companion paper \cite{Morand2023} for explicit constructions. 
\ee

\item We show that every galilean Klein pair admits leibnizian lifts, both in the reductive [Table \ref{Table:Reductive leibnizian lifts of galilean algebras}] and non-necessarily reductive case [Table \ref{Table:Leibnizian lifts of galilean algebras}]. In particular, the euclidean/lorentzian newton algebra \eqref{equation:euclidean/lorentzian newton algebra} is shown to admit:
\be
\item a unique reductive leibnizian lift given by \eqref{equation:euclidean/lorentzian leibnizian algebra}, geometrically realised as isometry algebra of the leibnizian manifold \eqref{equation:leibnizian structure}.
\item a 1-parameter family of non-necessarily reductive leibnizian lifts, with additional commutators:
 \bea
 \label{equation:additional non-reductive commutators}
 \begin{tikzcd}[row sep = small,column sep = small,ampersand replacement=\&]
\br{H}{M}=r\, M\&\br{H}{C}= r\, C\&\br{H}{\tD}=r\, \tD
\end{tikzcd}
 \eea
 where the free parameter $r\in\mathbb R$ encodes the obstruction to reductivity.\footnote{The centrally extended lorentzian newton algebra \eqref{equation:euclidean/lorentzian newton algebra}+\eqref{equation:euclidean/lorentzian bargmannian algebra} (\ie with $\epsilon=-1$) can be realised as subalgebra of the non-reductive \eqref{equation:euclidean/lorentzian leibnizian algebra}+\eqref{equation:additional non-reductive commutators}, with $\epsilon=-1$ and $r=\pm1$. Similarly, the $\mathsf{G}$-ambient algebra \eqref{equation:euclidean/lorentzian G-ambient} with $\epsilon=-1$ is realised as subalgebra of  \eqref{equation:euclidean/lorentzian leibnizian algebra}+\eqref{equation:additional non-reductive commutators} for $r=\lambda\pm1$. Note however that such realisations do not carry in the euclidean case, \cf Section \ref{section:Non-reductive lifts of galilean Klein pairs} for details. }
\ee
\ee
\renewcommand{\theequation}{\arabic{equation}}
\numberwithin{equation}{section}
\paragraph{Plan of the paper}
\hfill

\medskip
We start by reviewing the basic definitions of kinematical algebras and Klein pairs in Section \ref{section:Beginning at the beginning: kinematical algebras}, as well as their classification\footnote{Throughout this work, we will focus on so-called \emph{universal} algebras \ie algebras whose commutators are valid in generic spatial dimension $d$, without relying on any dimension-specific structure. } from \cite{Bacry1968,Bacry1986,Figueroa-OFarrill2017c,Figueroa-OFarrill2017a,Andrzejewski2018,Figueroa-OFarrill2018,Figueroa-OFarrill2019}. Particular attention will be devoted to the partition of kinematical Klein pairs according to the type of invariant metric-like structures they admit. We then proceed in Section \ref{section:Curiouser and curiouser: aristotelian algebras} by mimicking the review of the above-mentioned results in the aristotelian and lifshitzian cases.

\medskip
Section \ref{section:A Tangled Tale: (s,v)-Lie algebras} serves to both justify the diverse array of algebras and Klein pairs employed in the present study [as succinctly presented in Table \ref{Table:Summary of Klein pairs}] and unify the corresponding terminology. [As such, it may be skipped during the initial reading.] We introduce the notion of $(s,v)$-Lie algebras and discuss the particular subfamily of $k$-kinematical algebras and related Klein pairs \UnskipRefs{Proposition:k-kinematical algebras}{Definition:k-kinematical Klein pairs}, and similarly in the aristotelian case \UnskipRefs{Proposition:k-aristotelian algebras}{Definition:k-aristotelian Klein pairs}. The latter will be argued to provide a natural generalisation of the notions of kinematical and aristotelian Klein pairs, respectively, on a spacetime of dimension $d+k$. We conclude the section by discussing a scheme allowing to relate different families of algebras via dimensional reduction that will later prove relevant \SectionRef{section:Drink me} in order to generate generalisations of the leibniz algebra. 

\medskip
Section \ref{section:Climbing up one leg of the table: an ambient perspective on Klein pairs} will see us \textit{try \emph{[our]} best to climb up one of the legs of the table}, an endeavour hopefully made less slippery by the preliminary work engaged in the preceding sections.  More precisely, our aim will be to provide a comprehensive account of the ambient approach at the algebraic level of Klein pairs.  We will start by unpacking the previously defined `$k$-generalisations' of usual notions in the specific $k=2$ case---which will prove to be the case suitable for the ambient approach---and for which the restriction of the previous notions will be dubbed \emph{ambient kinematical} \UnskipRefs{Definition:Ambient kinematical algebra}{Definition:Ambient kinematical Klein pairs} and \emph{ambient aristotelian} \UnskipRefs{Definition:Ambient aristotelian algebra}{Definition:Ambient aristotelian Klein pairs}, respectively. We will then make a detour into the realm of Klein pair reduction along an ideal in order to define the notion of \emph{projectable triplets} \UnskipRef{Definition:Projectable triplet}, which, we shall argue, constitutes the natural algebraic abstraction underlying the usual Duval--Eisenhart lift of nonrelativistic structures into codimension one `ambient' geometries. We conclude the section by discussing two prominent examples of lifts of the galilei Klein pair---\ie into the leibniz \UnskipRef{Example:Leibniz projectable triplet} and bargmann \UnskipRef{Example:Bargmann projectable triplet} Klein pairs, respectively---thus providing two potential paths for generalisation. At this point, we will tuck the leibnizian side of the story in our collective pocket-watch, only to revisit it when the clock strikes the appropriate hour [approximately, when approaching Section \ref{section:Drink me}] and focus in the meantime on the bargmannian route and its ambient aristotelian generalisation.

\medskip
Section \ref{section:Bargmann and his modern rivals} contains our main results. After defining the notion of \emph{projectable ambient triplets} \UnskipRef{Definition:Projectable ambient triplets} at the intersection of projectable triplets \UnskipRef{Definition:Projectable triplet} and ambient aristotelian Klein pairs \UnskipRef{Definition:Ambient aristotelian Klein pairs}, we undertake the classification of the former, proceeding in three steps. We start by classifying all ambient aristotelian algebras admitting a scalar ideal [\cf Table \ref{Table:Ambient aristotelian algebras admitting a scalar ideal}] and provide comparison with the parts of this classification already available in the literature. This first step constitutes the starting point to classify all associated Klein pairs \UnskipRef{Definition:Ambient kinematical Klein pairs}, both effective [\cf Table \ref{Table:Effective ambient aristotelian Klein pairs with scalar ideal}] and non-effective, in which case we exhibit the associated effective lifshitzian Klein pair  [\cf Table \ref{Table:Effective lifshitzian Klein pairs associated with ambient aristotelian algebras}]. We proceed by identifying effective ambient aristotelian Klein pairs [with scalar ideal] whose projection along the scalar ideal yields an effective [kinematical] Klein pair. Such projectable ambient Klein pairs are collected in Table \ref{Table:Projectable ambient triplets} and further categorised according to their invariant ambient metric structure. The classification admits a five-fold partition. Three of these classes coincide with the unique [and trivial] scalar extension of non-galilean kinematical algebras [via direct sum $\alg\oplus\mathbb R$]. One additional class corresponds to the unique [albeit non-trivial] lift of galilean kinematical algebras into the bargmannian algebras recently exhibited and studied in \cite{FigueroaOFarrill2022,FigueroaOFarrill2022e}. Apart from the paradigmatic lift of the galilei Klein pair into the bargmann Klein pair studied by Duval \etal \cite{Duval1977,Duval1985,Duval1991} as well as its curved [euclidean/lorentzian newton] generalisations into Hpp-waves \cite{Gibbons:2003rv}, this class includes embedding algebras for the torsional avatars of the galilei Klein pair, as recently classified in \cite{Figueroa-OFarrill2018} and further studied in \cite{FigueroaOFarrill2022d}. 
This bargmannian class is privileged among all five, being the only one admitting an invariant \emph{non-degenerate} bilinear form defined on the whole $(d+2)$-dimensional tangent space [in the corresponding Cartan geometry]. One primary motivation for the present series of works indeed consists in integrating such bargmannian Klein pairs into an exhaustive classification, as well as to revisit the corresponding geometric realisation and projection \cite{Morand2023}. The last class of the classification is given by so-called \emph{$\mathsf{G}$-ambient} Klein pairs which provide an alternative ambient lift for all galilean kinematical Klein geometries. This $\mathsf{G}$-ambient class is distinguished as the only class allowing to lift kinematical Klein pairs in a non-unique way, the arbitrariness being encoded in the ambient torsion. The galilean family is thus privileged from an ambient perspective\footnote{We should stress that atop the ambient perspective adopted in the present work, which indeed accords primacy to galilean structures, there exists an alternative---dual---perspective \cite{Duval:2014uoa,Morand2018} that instead privileges carrollian structures. We hope to delve deeper into this latter approach in future work. } as the only family of kinematical Klein pairs admitting a non-unique [and non-trivial] lift into two different classes of ambient structures [\cf Table \ref{Table:Possible lifts of effective galilean Klein pairs}]. 

\medskip
Not unlike Alice, we will \textit{give \emph{[ourselves]} very good advice} [specifically in this case, `one should aim at integrating ambient Klein pairs into rigorous classifications'] \emph{only to not follow it} in Section \ref{section:Drink me} by displaying projectable leibnizian Klein pairs obtained through a non-classificatory method, namely by combining the idea of dimensional reduction \SectionRef{section:A Tangled Tale: (s,v)-Lie algebras} together with the one of \IW contraction \SectionRef{section:Contraction of Klein pairs}. The output of this procedure will be a set of Klein geometries endowed with a leibnizian metric structure and allowing to lift all kinematical galilean structures in a maximally symmetric fashion [\cf Table \ref{Table:Reductive leibnizian lifts of galilean algebras}]. Finally, we will go one step forward into our transgression by displaying novel---yet largely \emph{ex nihilo}---lifts of galilean Klein geometries into \emph{non-reductive} leibnizian Klein pairs [\cf Table \ref{Table:Leibnizian lifts of galilean algebras}], the introduction of which will allow to generalise the above-mentioned realisation of the bargmann algebra as subalgebra of the leibniz algebra to some of the previously classified ambient aristotelian Klein pairs. 

\section{Beginning at the beginning: kinematical algebras}
\label{section:Beginning at the beginning: kinematical algebras}

In a seminal paper \cite{Bacry1968}, Bacry and L\'evy--Leblond introduced the notion of kinematical algebras, namely Lie algebras encoding the infinitesimal kinematical symmetries of any free particle. A precise definition can be articulated as follows \cite{Bacry1986,Figueroa-OFarrill2017c}:

\begin{Definition}[Kinematical algebra]\label{Definition:Kinematical algebra}A kinematical algebra with $d$-dimensional spatial isotropy is a Lie algebra $\alg$ containing a $\so(d)$-subalgebra and admitting the following decomposition of $\so(d)$-modules:
\vspace{-6mm}

\bea
\alg=\tS\oplus\bigoplus_{a=1}^2\tVs{a}\oplus\tJ\label{equation:decomposition of kinematical algebras}
\eea
\vspace{-8mm}

where:
\be
\item The generators $\tJ$ span the $\so(d)$-subalgebra\footnote{Here and in the following, $\br{\tJ}{\tJ}\sim \tJ$ stands for $\br{J_{ij}}{J_{kl}}=\delta_{jk}J_{il}-\delta_{jl}J_{ik}+\delta_{il}J_{jk}-\delta_{ik}J_{jl}$, with $i,j,k,l\in\pset{1,\ldots,l}$, while $\br{\tJ}{\tV}\sim \tV$ stands for $\br{J_{ij}}{V_k}=-\delta_{ik}V_j+\delta_{j k}V_i$, with $\delta_{ij}$ the Kronecker delta in $d$ dimensions.}  \ie $\br{\tJ}{\tJ}\sim \tJ$.
\item The generators $\tVs{a}$'s are in the vector representation of $\mathfrak{so}(d)$, \ie $\br{\tJ}{\tVs{a}}\sim \tVs{a}$.
\item The generator $\tS$ is in the scalar representation of $\mathfrak{so}(d)$, \ie $\br{\tJ}{\tS}\sim0$.
\ee
\end{Definition}

Two kinematical algebras $\alg_1,\alg_2$ are isomorphic if there is an isomorphism of Lie algebras $\varphi\From\alg_1\to\alg_2$ preserving the $\so(d)$-decomposition, \ie $\varphi(\tJ_1)= \tJ_2$, $\varphi(S_1)= S_2$ and $\varphi(\tVs{1}_1\oplus\tVs{2}_1)= \tVs{1}_2\oplus\tVs{2}_2$.
\begin{Remark}
\label{Remark:indiscernable}
Note that the $\so(d)$-decomposition treats the two vector modules $\tVs{a}$ as `indiscernible'. In particular, isomorphisms of kinematical algebras are \emph{not} assumed to preserve $\tVs{1}$ and $\tVs{2}$ individually. 
\end{Remark}
It is customary to relabel the above generators as $S=H$, $\tVs{1}=\tP$ and $\tVs{2}=\tK$ so that the above $\so(d)$-decomposition reads more familiarly:
\[
\alg=H\oplus\tP\oplus\tK\oplus\tJ\, .
\]
Let us stress, consistently with the above remark, that the labels of the vectorial generators are given for definiteness only but are not imparted with any geometric meaning at this stage. Such geometric interpretation will require additional structure, namely a subalgebra $\alh\subset\alg$ making $(\alg,\alh)$ a Klein pair, see Definition \ref{Definition:Infinitesimal Klein pair} below. 

\medskip
Kinematical algebras in all dimension $d$ have been classified up to isomorphisms in a series of works due to J.~Figueroa--O'Farrill and collaborators [\cf \cite{FigueroaOFarrill2017} for a summary]:
\bi
\item $d=0$: There is a unique kinematical algebra, which is $1$-dimensional and hence abelian.
\item $d=1$: There is no skewsymmetric tensor in dimension $d=1$, hence no candidate for the generator $\tJ$ of spatial rotations. The defining conditions of a kinematical algebra are then automatically satisfied, so that any 3-dimensional real Lie algebra is a kinematical algebra. The classification of $d=1$ kinematical algebras thus identifies with the Bianchi classification \cite{Bianchi1898}. 
\item $d=2$: Kinematical algebras in $d=2$ were classified in \cite{Andrzejewski2018}.
\item $d=3$: Kinematical algebras in $d=3$ were classified in \cite{Bacry1986,Figueroa-OFarrill2017c}.
\item $d>3$: Kinematical algebras in $d>3$ were classified in \cite{Bacry1968,Figueroa-OFarrill2017a}. 
\ei
A \emph{universal kinematical algebra} is  a sequence of kinematical algebras indexed by an integer $d\in\mathbb N$ -- referred to as the spatial isotropy dimension -- such that the associated commutation relations take the same form for any value of $d$. In plain words, kinematical algebras are said to be universal if their commutators do not rely on any dimension-specific structure (\eg the Levi--Civita tensor) and are hence valid in generic dimension. Universal kinematical algebras are classified\footnote{Our choice of parameterisation differs from the one used in \cite[Table 17]{Figueroa-OFarrill2017a}. The Rosetta stone between the two characterisations is given by the following table: 
\medskip

\begin{center}
\begin{tabular}{l|l}
Table 17 of \cite{Figueroa-OFarrill2017a}&Table \ref{Table:Universal kinematical algebras}\\\hline
Eq.5 (static)& \hyperlink{Table:K1}{$\mathsf{K1}$}\\
Eq.22 (galilei)\cellcolor{Gray}&\cellcolor{Gray}\hyperlink{Table:K3}{$\mathsf{K3}$}\\
Eq.20 with $\gamma=-1$ (lorentzian newton)&\hyperlink{Table:K6}{$\mathsf{K6}$}\\
Eq.20 with $-1<\gamma<0$\cellcolor{Gray}&\cellcolor{Gray}\hyperlink{Table:K8}{$\mathsf{K8}_{\alpha_-}$} with $\alpha_->0$ where $\alpha_-=\frac{\gamma+1}{\sqrt{-\gamma}}$\\
Eq.20 with $\gamma=0$&\hyperlink{Table:K4}{$\mathsf{K4}$}\\
Eq.20 with $0<\gamma<1$\cellcolor{Gray}&\cellcolor{Gray}\hyperlink{Table:K7}{$\mathsf{K7}_{\alpha_+}$} with $\alpha_+>2$ where $\alpha_+=\frac{\gamma+1}{\sqrt{\gamma}}$\\
Eq.20 with $\gamma=1$&\hyperlink{Table:K2}{$\mathsf{K2}$}\\
Eq.23 \cellcolor{Gray}&\cellcolor{Gray}\hyperlink{Table:K7}{$\mathsf{K7}_{\alpha_+}$} with $\alpha_+=2$\\
Eq.21 with $\alpha=0$ (euclidean newton)&\hyperlink{Table:K5}{$\mathsf{K5}$}\\
Eq.21 with $\alpha>0$\cellcolor{Gray}&\cellcolor{Gray}\hyperlink{Table:K7}{$\mathsf{K7}_{\alpha_+}$} with $0<\alpha_+<2$ where $\alpha_+=\frac{2\alpha}{\sqrt{1+\alpha^2}}$
  \end{tabular}
  \end{center}
}
 in \cite[Table 17]{Figueroa-OFarrill2017a}, classification that we reproduce in Table \ref{Table:Universal kinematical algebras}.

\begin{table}[ht]
\centering
\resizebox{14cm}{!}{
\begin{tabular}{l|l|lllll}
\multicolumn{1}{c|}{\textbf{Label}}&\multicolumn{1}{c|}{\textbf{Comments}}&\multicolumn{5}{c}{\textbf{Non-trivial commutators}}\\\hline
\hypertarget{Table:K1}{$\mathsf{K1}$}&static&&&&&\\
\rowcolor{Gray}
\hypertarget{Table:K2}{$\mathsf{K2}$}&&&$\br{H}{\tK}= \tK$&$\br{H}{\tP}=\tP$&&\\
\hypertarget{Table:K3}{$\mathsf{K3}$}&galilei&&&$\br{H}{\tP}=\tK$&&\\
\rowcolor{Gray}
\hypertarget{Table:K4}{$\mathsf{K4}$}&&&&$\br{H}{\tP}=\tP$&&\\
\hypertarget{Table:K5}{$\mathsf{K5}$}&euclidean newton&&$\br{H}{\tK}=-\tP$&$\br{H}{\tP}=\tK$&&\\
\rowcolor{Gray}
\hypertarget{Table:K6}{$\mathsf{K6}$}&lorentzian newton&&$\br{H}{\tK}=\tP$&$\br{H}{\tP}=\tK$&&\\
\hypertarget{Table:K7}{$\mathsf{K7}_{\alpha_+}$}&$\alpha_+>0$&&$\br{H}{\tK}=-\tP$&$\br{H}{\tP}=\tK+\alpha_+\tP$&&\\
\rowcolor{Gray}
\hypertarget{Table:K8}{$\mathsf{K8}_{\alpha_-}$}&$\alpha_->0$&&$\br{H}{\tK}=\tP$&$\br{H}{\tP}=\tK+\alpha_-\tP$&&\\
\hypertarget{Table:K9}{$\mathsf{K9}$}&carroll&$\br{\tK}{\tP}=H$&&&&\\
\rowcolor{Gray}
\hypertarget{Table:K10}{$\mathsf{K10}$}&$\mathfrak{iso}(d+1)$&$\br{\tK}{\tP}=H$&&$\br{H}{\tP}=-\tK$&$\br{\tP}{\tP}=-\tJ$&\\
\hypertarget{Table:K11}{$\mathsf{K11}$}&$\mathfrak{iso}(d,1)$&$\br{\tK}{\tP}=H$&&$\br{H}{\tP}=\tK$&$\br{\tP}{\tP}=\tJ$&\\
\rowcolor{Gray}
\hypertarget{Table:K12}{$\mathsf{K12}$}&$\so(d+1,1)$&$\br{\tK}{\tP}=H$&$\br{H}{\tK}=\tP$&$\br{H}{\tP}=\tK$&$\br{\tP}{\tP}=\tJ$&$\br{\tK}{\tK}=-\tJ$\\
\hypertarget{Table:K13}{$\mathsf{K13}$}&$\so(d+2)$&$\br{\tK}{\tP}=H$&$\br{H}{\tK}=\tP$&$\br{H}{\tP}=-\tK$&$\br{\tP}{\tP}=-\tJ$&$\br{\tK}{\tK}=-\tJ$\\
\rowcolor{Gray}
\hypertarget{Table:K14}{$\mathsf{K14}$}&$\so(d,2)$&$\br{\tK}{\tP}=H$&$\br{H}{\tK}=-\tP$&$\br{H}{\tP}=\tK$&$\br{\tP}{\tP}=\tJ$&$\br{\tK}{\tK}=\tJ$
  \end{tabular}
   }
        \caption{Universal kinematical algebras}
      \label{Table:Universal kinematical algebras}
\end{table}

\paragraph{Kinematical Klein pairs}
Our interest for kinematical algebras lies more precisely in geometrical incarnations thereof as homogeneous spacetimes. Infinitesimally, such homogeneous spacetimes are described by kinematical algebras endowed with an additional structure---namely, a particular subalgebra---thus forming an \emph{infinitesimal Klein pair}.
\begin{Definition}[Infinitesimal Klein pair]
\label{Definition:Infinitesimal Klein pair}
An {infinitesimal Klein pair} (hereafter Klein pair for short) is a pair $(\alg,\alh)$ where $\alg$ is a Lie algebra and $\alh\subset\alg$ a Lie subalgebra.
\bi
\item A Klein pair  is said \emph{reductive} if $\alg$ admits a canonical decomposition $\alg=\alh\oplus \alp$ as $\alh$-modules \ie $\br{\alh}{\alp}\subset\alp$.
\item A reductive Klein pair will be said to be \emph{symmetric} (resp. \emph{flat}) if $\br{\alp}{\alp}\subseteq\alh$ (resp. if $\br{\alp}{\alp}\subseteq\alp$).
\item The \emph{ideal core} $\mathfrak{n}$ of a Klein pair is the largest ideal of $\alg$ contained in $\alh$.
\item A Klein pair is said \emph{effective} if it has zero ideal core \ie $\alh$ does not contain any non-trivial ideal of $\alg$.
\item Let $(\alg,\alh)$ be a Klein pair with ideal core $\mathfrak{n}$. The Klein pair $(\alg/\aln,\alh/\aln)$ is effective and called the \emph{associated effective} Klein pair. 
\item A morphism of Klein pairs  $\varphi\From(\alg_1,\alh_1)\to(\alg_2,\alh_2)$ is a homomorphism of Lie algebras $\varphi\From\alg_1\to\alg_2$ such that $\varphi(\alh_1)\subseteq\alh_2$.
\item Two Klein pairs $(\alg_1,\alh_1)$ and $(\alg_2,\alh_2)$ are isomorphic if there is an isomorphism of Lie algebras $\varphi\From\alg_1\to\alg_2$ such that $\varphi(\alh_1)=\alh_2$. 
\ei
\end{Definition}

Let $(\alg,\alh)$ be a Klein pair. There is a short exact sequence of vector spaces\footnote{We do not assume that $\alh$ is an ideal in $\alg$, so that $\gh$ is not assumed to be a Lie algebra and the sequence is not a sequence of Lie algebras {\it a priori}.}:
\[
\begin{tikzcd}
\label{equation:SHS}
0 \arrow[r] & \alh  \arrow[r, hook, "i"] & \alg \arrow[r, two heads,"\pi"]        & \alg/\alh\arrow[r]  & 0
\end{tikzcd}
\]
where $i\From\alh\hookrightarrow\alg$ is the natural embedding of Lie algebras and $\pi\From\alg\twoheadrightarrow\alg/\alh$ the projection map on the quotient of $\alg$ by $\alh$. The adjoint action of $\alh$ on $\alg$ projects on $\alg/\alh$ so that one can define an action $\ad_\alh\From\alh\to\End(\gh)$.

We will now consider $\ad_\alh$-invariant structures defined on $\gh$\footnote{According to the holonomy principle of Cartan geometries  \cite{Sharpe1997}, $\ad_\alh$-invariant tensors on $\gh$ are associated with (canonical or gauge-invariant) tensors on the base manifold of Cartan geometries modelled on the Klein pair $(\alg,\alh)$. Furthermore, whenever the Klein geometry is reductive \UnskipRef{Definition:Infinitesimal Klein pair}, such canonical tensors are parallelised by the canonical Koszul connection induced on the base manifold by the Ehresmann connection obtained from the part of the reductive Cartan connection taking values in $\alh$.} as a means to partition Klein pairs.

\begin{Definition}[Kinematical metric structures on Klein pairs]
\label{Definition:Metric structures on Klein pairs}
A Klein pair $(\alg,\alh)$ will be said:
\begin{description}
\item[galilean] if $\gh$ is endowed with a pair $(\bs{\psi},\bs{h})$ where:
\be
\item $\bs{{\psi}}\in(\gh)^*$ is a non-zero $\ad_\alh$-invariant linear form
\item $\bs{h}\in(\gh)^{\otimes2}$ is a symmetric, positive semi-definite and $\ad_\alh$-invariant bilinear form whose radical is spanned by $\bs{\psi}$ \ie $\bs{h}(\alpha,\beta)=0$ for all $\beta\in(\gh) ^*$ $\Leftrightarrow$ $\alpha\sim\bs{\psi}$
\ee
\item[carrollian] if $\gh$ is endowed with a pair $(\bs{\xi},\bs{\gamma})$ where:
\be
\item $\bs{\xi}\in\gh$ is a non-zero $\ad_\alh$-invariant vector
\item $\bs{\gamma}\in(\gh^*){}^{\otimes2}$ is a symmetric, positive semi-definite and $\ad_\alh$-invariant bilinear form whose radical is spanned by $\bs{\xi}$ \ie $\bs{\gamma}(X,Y)=0$ for all $Y\in\gh$ $\Leftrightarrow$ $X\sim\bs{\xi}$
\ee
\item[riemannian] if $\gh$ is endowed with a symmetric, positive-definite and $\ad_\alh$-invariant bilinear form $\bs{g}\in(\gh^*){}^{\otimes2}$
\item[lorentzian] if $\gh$ is endowed with a symmetric, non-degenerate and $\ad_\alh$-invariant bilinear form $\bs{g}\in(\gh^*){}^{\otimes2}$ of signature $(-1,\underbrace{+1,\ldots,+1}_d)$ with $\dim\alg-\dim\alh=d+1$
\end{description}
\end{Definition}

\begin{Remark}
\label{Remark:galilean metric2}
It is sometimes useful to equivalently define the galilean metric structure as a pair $(\bs{\psi},\bs{\gamma})$ where $\bs{\gamma}\in(\Ker\bs{\psi})^*{}^{\otimes2}$ is a symmetric, positive-definite and $\ad_\alh$-invariant bilinear form on the kernel of the $\ad_\alh$-invariant linear form $\bs{\psi}$.
\end{Remark}

\begin{Remark}
Note that the epithets galilean,  carrollian, riemannian and lorentzian adopted in Definition \ref{Definition:Metric structures on Klein pairs} only make sense for Klein pairs, as the corresponding metric structures live on the quotient space $\gh$. This is in contradistinction with the definitions adopted in \cite{FigueroaOFarrill2022f} which hold for kinematical algebras, as the corresponding metric structures live on the full space $\alg$ rather that on a quotient thereof.
\end{Remark}

A {Klein pair} $(\alg,\alh)$ such that the Lie algebra $\alg$ is a kinematical algebra and the subalgebra $\alh$ admits the $\so(d)$-decomposition $\alh=\tK\oplus\tJ$ is called a \emph{kinematical Klein pair}.\footnote{Note that the definition of kinematical Klein pair adopted here \big[following \cite{Figueroa-OFarrill2018}\big] is less constraining that the one elicited in the original work \cite{Bacry1968} which furthermore required the following endomorphisms of vector spaces:
\bi
\item \emph{time-reversal} $\Theta \From\alg\to\alg$
\[
 \begin{tikzcd}[row sep = small,column sep = small,ampersand replacement=\&]
\Theta :\& H\mapsto -H\& \tP\mapsto \tP\& \tK\mapsto -\tK\& \tJ\mapsto \tJ
\end{tikzcd}
\]
\item \emph{space-reversal} $\Pi\From\alg\to\alg$
\[
 \begin{tikzcd}[row sep = small,column sep = small,ampersand replacement=\&]
\Pi:\& H\mapsto H\& \tP\mapsto -\tP\& \tK\mapsto -\tK\& \tJ\mapsto \tJ
\end{tikzcd}
\]
\ei 
to be Lie algebra automorphisms of $\alg$. In more modern terms, these two conditions are equivalent for the Klein pair $(\alg,\alh)$ to be both reductive and symmetric \UnskipRef{Definition:Infinitesimal Klein pair}. We will consistently relax these two conditions throughout this work.} An isomorphism of kinematical Klein pairs is an isomorphism of Klein pairs that preserves the $\so(d)$-decomposition. A Klein pair $(\alg,\alh)$ is called \emph{universal} whenever both $\alg$ and $\alh$ are universal algebras.

Having geometric realisations in mind, we are more particularly interested in \textit{effective}\footnote{\label{footnote:effectiveness}It can be checked that a kinematical Klein pair is effective if and only if the generators $\tK$ do not span an ideal of $\alg$ [since the canonical relation $\br{\tJ}{\tP}=\tP$ prevents $\tJ$ to belong to an ideal contained in $\alh$]. The possible obstructions for $\tK$ to span an ideal of $\alg$ take the form: 
\[
 \begin{tikzcd}[row sep = small,column sep = small,ampersand replacement=\&]
\br{\tK}{\tK}\sim \tJ \& \br{\tK}{H}\sim \tP\&\br{\tK}{\tP}\sim H\&\br{\tK}{\tP}\sim \tJ\, .
\end{tikzcd}
\]
Letting $\alp$ be a supplementary of $\alh$ [\ie $\alg=\alh\oplus\alp$ as a vector space], the set of commutators $\br{\alh}{\alp}\subset\alp$ determine what (possibly degenerate) metric structure on $\alp$ is $\ad_\alh$-invariant. Hence, different choices of obstruction yield different $\ad_\alh$-invariant metric structures \UnskipRef{Definition:Metric structures on Klein pairs}:
\bi
\item (pseudo)-riemannian: $\br{\tK}{H}\sim \tP$ \qq $\br{\tK}{\tP}\sim H$
\item galilean:  $\br{\tK}{H}\sim \tP$
\item carrollian: $\br{\tK}{\tP}\sim H$.
\ei
} kinematical Klein pairs \cite[Appendix B]{Figueroa-OFarrill2018}. Universal effective kinematical Klein pairs have been classified in \cite[Table 1]{Figueroa-OFarrill2018}, classification that we reproduce\footnote{Our choice of parameterisation for galilean Klein pairs differs from the one used in \cite[Table 1]{Figueroa-OFarrill2018}. The Moabite Stone between the two characterisations can be found below: 
\medskip

\begin{center}
\begin{tabular}{l|l}
Table 1 of \cite{Figueroa-OFarrill2018}&Table \ref{Table:Effective kinematical Klein pairs}\\\hline
galilean $\mathsf{G}$& \hyperlink{Table:E1}{$\mathsf{E1}$}\\
galilean de sitter $\mathsf{dSG} = \mathsf{dSG}_{\gamma=-1}$\cellcolor{Gray}&\cellcolor{Gray}\hyperlink{Table:E4}{$\mathsf{E4}$}\\
torsional galilean de sitter $\mathsf{dSG}_{\gamma\in(-1,0)}$&\hyperlink{Table:E6}{$\mathsf{E6}_{\alpha_-}$} with $\alpha_->0$ where $\alpha_-=\frac{\gamma+1}{\sqrt{-\gamma}}$\\
torsional galilean de sitter $\mathsf{dSG}_{\gamma=0}$\cellcolor{Gray}&\cellcolor{Gray}\hyperlink{Table:E2}{$\mathsf{E2}$}\\
torsional galilean de sitter $\mathsf{dSG}_{\gamma\in(0,1)}$&\hyperlink{Table:E5}{$\mathsf{E5}_{\alpha_+}$} with $\alpha_+>2$ where $\alpha_+=\frac{\gamma+1}{\sqrt{\gamma}}$\\
torsional galilean de sitter $\mathsf{dSG}_{\gamma=1}$\cellcolor{Gray}&\cellcolor{Gray}\hyperlink{Table:E5}{$\mathsf{E5}_{\alpha_+}$} with $\alpha_+=2$\\
galilean anti de sitter $\mathsf{AdSG} = \mathsf{AdSG}_{\chi=0}$&\hyperlink{Table:E3}{$\mathsf{E3}$}\\
torsional galilean anti de sitter $\mathsf{AdSG}_{\chi>0}$\cellcolor{Gray}&\cellcolor{Gray}\hyperlink{Table:E5}{$\mathsf{E5}_{\alpha_+}$} with $0<\alpha_+<2$ where $\alpha_+=\frac{2\chi}{\sqrt{1+\chi^2}}$
  \end{tabular}
  \end{center}
Note that the commutators presented in Table \ref{Table:Effective kinematical Klein pairs} (as well as in all subsequent classifying tables of Klein pairs) are partitioned in columns so that to reflect the vector space decomposition of the Lie algebra $\alg$ as $\alg=\alh\oplus\alp$ [this split being furthermore a decomposition of $\ad_\alh$-modules whenever $\alg$ is reductive, \ie whenever $\br{\alh}{\alp}\subset\alp$]. Although this notation deviates slightly from the standard conventions, it offers the advantage of emphasising the corresponding geometric interpretation. For instance, what would conventionally be denoted as $\br{H}{\tP}=\tK+\alpha_+\tP$ is now split into a \emph{curvature} component $\br{H}{\tP}=\tK$ and a \emph{torsion} component $\br{H}{\tP}=\alpha_+\tP$. We hope that this notation can provide readers with a clearer insight into the underlying geometry and an easier way to visually discriminate between pairs of the same family. 
}
 in Table \ref{Table:Effective kinematical Klein pairs} along the partition according to the metric structure on $\gh$ \UnskipRef{Definition:Metric structures on Klein pairs}.

\begin{table}[ht]
\begin{adjustwidth}{-0.5cm}{}
\resizebox{18cm}{!}{
\begin{tabular}{c|l|l|l|l|ll|ll|l|ll}
\multicolumn{1}{c|}{\textbf{Metric structure}}&\multicolumn{1}{c|}{\textbf{Label}}&\multicolumn{1}{c|}{\textbf{Algebra}}&\multicolumn{1}{c|}{\textbf{Comments}}&\multicolumn{1}{c|}{$\br{\alh}{\alh}\subset\alh$}&\multicolumn{2}{c|}{$\br{\alh}{\alp}\subset\alp$}&\multicolumn{2}{c|}{$\br{\alp}{\alp}\subset\alh$ (Curvature)}&\multicolumn{1}{c|}{$\br{\alp}{\alp}\subset\alp$ (Torsion)}&\multicolumn{2}{c}{$\br{\alh}{\alp}\subset\alh$ (Non-reductivity)}\\\hline
\multirow{6}{*}{\textnormal{galilean}}&\hypertarget{Table:E1}{$\mathsf{E1}$}\label{E1}&\hyperlink{Table:K3}{$\mathsf{K3}$}&{galilei}&&&$\br{\tK}{H}= \tP$&&&&&\\
&\cellcolor{Gray}\hypertarget{Table:E2}{$\mathsf{E2}$}&\hyperlink{Table:K4}{$\mathsf{K4}$}\cellcolor{Gray}&\cellcolor{Gray}&\cellcolor{Gray}&\cellcolor{Gray}&\cellcolor{Gray}$\br{\tK}{H}= \tP$&\cellcolor{Gray}&\cellcolor{Gray}&$\br{H}{\tP}= \tP$\cellcolor{Gray}&\cellcolor{Gray}&\cellcolor{Gray}\\
&\hypertarget{Table:E3}{$\mathsf{E3}$}&\hyperlink{Table:K5}{$\mathsf{K5}$}&euclidean newton&&&$\br{\tK}{H}= \tP$&$\br{H}{\tP}= \tK$&&&&\\
&\cellcolor{Gray}\hypertarget{Table:E4}{$\mathsf{E4}$}&\cellcolor{Gray}\hyperlink{Table:K6}{$\mathsf{K6}$}&\cellcolor{Gray}lorentzian newton&\cellcolor{Gray}&\cellcolor{Gray}&\cellcolor{Gray}$\br{\tK}{H}= \tP$&\cellcolor{Gray}$\br{H}{\tP}= -\tK$&\cellcolor{Gray}&\cellcolor{Gray}&\cellcolor{Gray}&\cellcolor{Gray}\\
&\hypertarget{Table:E5}{$\mathsf{E5}_{\alpha_+}$}&\hyperlink{Table:K7}{$\mathsf{K7}_{\alpha_+}$}&$\alpha_+>0$&&&$\br{\tK}{H}= \tP$&$\br{H}{\tP}= \tK$&&$\br{H}{\tP}= \alpha_+\tP$&&\\
&\cellcolor{Gray}\hypertarget{Table:E6}{$\mathsf{E6}_{\alpha_-}$}&\cellcolor{Gray}\hyperlink{Table:K8}{$\mathsf{K8}_{\alpha_-}$}&\cellcolor{Gray}$\alpha_->0$&\cellcolor{Gray}&\cellcolor{Gray}&\cellcolor{Gray}$\br{\tK}{H}= \tP$&\cellcolor{Gray}$\br{H}{\tP}= -\tK$&\cellcolor{Gray}&$\br{H}{\tP}= \alpha_-\tP$\cellcolor{Gray}&\cellcolor{Gray}&\cellcolor{Gray}\\
\hline
\multirow{4}{*}{\textnormal{carrollian}}&\hypertarget{Table:E7}{$\mathsf{E7}$}&\hyperlink{Table:K9}{$\mathsf{K9}$}&{carroll}&&$\br{\tK}{\tP}= H$&&&&&&\\
&\cellcolor{Gray}\hypertarget{Table:E8}{$\mathsf{E8}$}&\cellcolor{Gray}\hyperlink{Table:K10}{$\mathsf{K10}$}&\cellcolor{Gray}{ds carroll}&\cellcolor{Gray}&\cellcolor{Gray}$\br{\tK}{\tP}= H$&\cellcolor{Gray}&\cellcolor{Gray}$\br{H}{\tP}= -\tK$&$\br{\tP}{\tP}=-\tJ$\cellcolor{Gray}&\cellcolor{Gray}&\cellcolor{Gray}&\cellcolor{Gray}\\
&\hypertarget{Table:E9}{$\mathsf{E9}$}&\hyperlink{Table:K11}{$\mathsf{K11}$}&{ads carroll}&&$\br{\tK}{\tP}= H$&&$\br{H}{\tP}= \tK$&$\br{\tP}{\tP}=\tJ$&&&\\
&\cellcolor{Gray}\hypertarget{Table:E10}{$\mathsf{E10}$}&\cellcolor{Gray}\hyperlink{Table:K12}{$\mathsf{K12}$}&\cellcolor{Gray}{carroll light cone}&\cellcolor{Gray}&\cellcolor{Gray}$\br{\tK}{\tP}= H$&\cellcolor{Gray}&\cellcolor{Gray}&\cellcolor{Gray}&$\br{H}{\tP}= -\tP$\cellcolor{Gray}&\cellcolor{Gray}$\br{\tK}{\tP}= \tJ$&\cellcolor{Gray}$\br{\tK}{H}= -\tK$\\
\hline
&\hypertarget{Table:E11}{$\mathsf{E11}$}&\hyperlink{Table:K10}{$\mathsf{K10}$}&{euclidean}&$\br{\tK}{\tK}=-\tJ$&$\br{\tK}{\tP}= H$&$\br{\tK}{H}= -\tP$&&&&&\\
\textnormal{riemannian}&\cellcolor{Gray}\hypertarget{Table:E12}{$\mathsf{E12}$}&\cellcolor{Gray}\hyperlink{Table:K13}{$\mathsf{K13}$}&\cellcolor{Gray}{sphere}&\cellcolor{Gray}$\br{\tK}{\tK}=-\tJ$&\cellcolor{Gray}$\br{\tK}{\tP}= H$&\cellcolor{Gray}$\br{\tK}{H}= -\tP$&\cellcolor{Gray}$\br{H}{\tP}=-\tK$&$\br{\tP}{\tP}=-\tJ$\cellcolor{Gray}&\cellcolor{Gray}&\cellcolor{Gray}&\cellcolor{Gray}\\
&\hypertarget{Table:E13}{$\mathsf{E13}$}&\hyperlink{Table:K12}{$\mathsf{K12}$}&{hyperbolic}&$\br{\tK}{\tK}=-\tJ$&$\br{\tK}{\tP}= H$&$\br{\tK}{H}= -\tP$&$\br{H}{\tP}=\tK$&$\br{\tP}{\tP}=\tJ$&&&\\
\hline
&\cellcolor{Gray}\hypertarget{Table:E14}{$\mathsf{E14}$}&\cellcolor{Gray}\hyperlink{Table:K11}{$\mathsf{K11}$}&\cellcolor{Gray}{minkowski}&\cellcolor{Gray}$\br{\tK}{\tK}=\tJ$&\cellcolor{Gray}$\br{\tK}{\tP}= H$&\cellcolor{Gray}$\br{\tK}{H}= \tP$&\cellcolor{Gray}&\cellcolor{Gray}&\cellcolor{Gray}&\cellcolor{Gray}&\cellcolor{Gray}\\
\textnormal{lorentzian}&\hypertarget{Table:E15}{$\mathsf{E15}$}&\hyperlink{Table:K12}{$\mathsf{K12}$}&{de sitter}&$\br{\tK}{\tK}=\tJ$&$\br{\tK}{\tP}= H$&$\br{\tK}{H}= \tP$&$\br{H}{\tP}=-\tK$&$\br{\tP}{\tP}=-\tJ$&&&\\
&\cellcolor{Gray}\hypertarget{Table:E16}{$\mathsf{E16}$}&\cellcolor{Gray}\hyperlink{Table:K14}{$\mathsf{K14}$}&\cellcolor{Gray}{anti de sitter}&\cellcolor{Gray}$\br{\tK}{\tK}=\tJ$&\cellcolor{Gray}$\br{\tK}{\tP}= H$&\cellcolor{Gray}$\br{\tK}{H}= \tP$&\cellcolor{Gray}$\br{H}{\tP}=\tK$&$\br{\tP}{\tP}=\tJ$\cellcolor{Gray}&\cellcolor{Gray}&\cellcolor{Gray}&\cellcolor{Gray}
  \end{tabular}
  }
      \end{adjustwidth}
      \caption{Universal effective kinematical Klein pairs}
      \label{Table:Effective kinematical Klein pairs}
\end{table}

\begin{Remark}
To illustrate the difference between kinematical algebras and Klein pairs, we note that, while the kinematical algebras underlying \eg the minkowski and ads carroll Klein pairs are isomorphic \textit{as kinematical algebras}, they are not isomorphic as \textit{kinematical Klein pairs} since the isomorphism of kinematical algebras:
\[
 \begin{tikzcd}[row sep = small,column sep = small,ampersand replacement=\&]
H\mapsto -H\& \tP\mapsto \tK\& \tK\mapsto \tP\& \tJ\mapsto \tJ
\end{tikzcd}
\]
exchanges $\tP\leftrightarrow \tK$ [see  Remark \ref{Remark:indiscernable} above] and as such fails to preserve the subalgebra $\alh=\tK\oplus\tJ$, hence the need to distinguish those two Klein pairs in Table \ref{Table:Effective kinematical Klein pairs}.
\end{Remark}

Letting $\alg=H\oplus\tP\oplus\tK\oplus\tJ$ be a basis of $\alg$ with dual basis $\alg^*=H^*\oplus\tP^*\oplus\tK^*\oplus\tJ^*$, representatives for the  $\ad_\alh$-invariant metric structures associated with the kinematical Klein pairs of Table \ref{Table:Effective kinematical Klein pairs} are compiled in Table \ref{Table:Metric structures on kinematical Klein pairs} \big[see \eg \cite{Figueroa-OFarrill2018,Figueroa-OFarrill2019}\big]. 

\begin{table}[ht]
\centering
\resizebox{7cm}{!}{
\begin{tabular}{l|l}
\textbf{Metric structure}&\textbf{Invariant tensors}\\\hline
{galilean}&$\bs{\psi}=-H^*$\qq$\bs{h}=\tP\otimes \tP$\\
{carrollian}\cellcolor{Gray}&\cellcolor{Gray}$\bs{\xi}=H$\ \qq\ $\bs{\gamma}=\tP^*\otimes \tP^*$\\
{riemannian}&$\bs{g}=H^*\otimes H^*+\tP^*\otimes \tP^*$\\
{lorentzian}\cellcolor{Gray}&\cellcolor{Gray}$\bs{g}=-H^*\otimes H^*+\tP^*\otimes \tP^*$
  \end{tabular}
  }
      \caption{Metric structure on kinematical Klein pairs}
      \label{Table:Metric structures on kinematical Klein pairs}
\end{table}

\section{Curiouser and curiouser: aristotelian algebras}
\label{section:Curiouser and curiouser: aristotelian algebras}

Another algebraic structure of interest is given by so-called \emph{aristotelian algebras} or boostless kinematical algebras \cite{Penrose:1968ar,Horava2016,Grosvenor2016,Figueroa-OFarrill2018} which have recently known a surge of interest as curved backgrounds of `boostless hydrodynamics' and fractonic physical systems [\cf \cite{Bidussi:2021nmp,Jain:2021ibh,Glodkowski:2022xje,Armas:2023ouk,Jain:2023nbf,Marotta:2023ayw} and references therein]. 

\begin{Definition}[Aristotelian algebra]\label{Definition:aristotelian algebra}An aristotelian algebra (or boostless kinematical algebra) with $d$-dimensional spatial isotropy is a Lie algebra $\alg$ containing a $\so(d)$-subalgebra and admitting the following decomposition of $\so(d)$-modules:
\vspace{-10mm}

\bea
\alg=H\oplus\tP\oplus\tJ\nn
\eea
\vspace{-6mm}

where the generators are in the same $\so(d)$ representation as in \Defi{Definition:Kinematical algebra}. 
\end{Definition}

Note that the relevant notion of \emph{aristotelian Klein pair} is given by a pair $(\alg,\alh)$ where $\alh$ is nothing but the $\so(d)$-subalgebra spanned by $\tJ$. Since the latter is a canonical part of the structure, aristotelian Klein pairs are in bijective correspondence with aristotelian algebras. Furthermore, all aristotelian Klein pairs are trivially effective since the imposed commutator $\br{\tJ}{\tP}\sim\tP$ prevents $\tJ$ to span an ideal of $\alg$. Universal aristotelian algebras have been classified as follows \cite[Table 17]{Figueroa-OFarrill2018}:
\begin{table}[ht]
\centering
\resizebox{6cm}{!}{
\begin{tabular}{l|l|l}
\textbf{Klein pair}&$\br{\alp}{\alp}\subset\alh$&$\br{\alp}{\alp}\subset\alp$\\\hline
$\mathfrak{iso}(d)\oplus\mathbb R$&&\\
$\mathfrak{iso}(d)\mirroredinplus\mathbb R$\cellcolor{Gray}&\cellcolor{Gray}&$\br{H}{\tP}=\tP$\cellcolor{Gray}\\
$\so(d,1)\oplus\mathbb R$&$\br{\tP}{\tP}=\tJ$&\\
$\so(d+1)\oplus\mathbb R$\cellcolor{Gray}&\cellcolor{Gray}$\br{\tP}{\tP}=-\tJ$&\cellcolor{Gray}
  \end{tabular}
  }
           \caption{Universal aristotelian algebras}
      \label{Table:Universal aristotelian algebras}
\end{table}

Compared with kinematical Klein pairs, the loss of the boost generators $\tK$ is concomitant with the addition of an extra $\ad_\alh$-invariant structure living on the quotient $\gh$.
\begin{Definition}[Aristotelian metric structure on Klein pairs]
\label{Definition:aristotelian metric structures on Klein pairs}
A Klein pair $(\alg,\alh)$ will be said to be endowed with an aristotelian metric structure if the quotient space $\gh$ is endowed with any of the following equivalent structures:
\begin{description}
\item[galilean-like] a galilean structure $(\bs{\psi},\bs{\gamma})$ and an $\ad_\alh$-invariant vector $\bs{N}\in\gh$ satisfying $\bs{\psi(\bs{N}})=1$. 
\item[carrollian-like] a carrollian structure $(\bs{\xi},\bs{\gamma})$ and an $\ad_\alh$-invariant linear form $\bs{A}\in(\gh)^*$ satisfying $\bs{A(\bs{\xi}})=1$. 
\item[riemannian-like] a riemannian metric structure $\bs{g}$ and an $\ad_\alh$-invariant vector $\bs{N}\in\gh$ satisfying $\bs{g}(\bs{N},\bs{N})=1$. 
\item[lorentzian-like] a lorentzian metric structure $\bs{g}$ and an $\ad_\alh$-invariant vector $\bs{N}\in\gh$ satisfying $\bs{g}(\bs{N},\bs{N})=-1$. 
\end{description}
\end{Definition}

\begin{Proposition}
\label{Proposition:Equivalence definitions aristotelian metric structures}
The above definitions are equivalent.
\end{Proposition}
\begin{proof}
Let $(\bs{\psi},\bs{\gamma},\bs{N})$ be an aristotelian metric structure \ala galilei \UnskipRefs{Remark:galilean metric2}{Definition:aristotelian metric structures on Klein pairs} and define the symmetric, positive semi-definite and $\ad_\alh$-invariant bilinear form $\overset{\bs{N}}{\bs{\gamma}}\in(\gh^*){}^{\otimes2}$ as $\overset{\bs{N}}{\bs{\gamma}}(x,y)={\bs{\gamma}}\big(\overset{\bs{N}}{P}(x),\overset{\bs{N}}{P}(y)\big)$ where $\overset{\bs{N}}{P}\From\gh\to\Ker\bs{\psi}$ is defined as $\overset{\bs{N}}{P}(x)=x-\bs{\psi}(x)\cdot\boldsymbol N$, for all $x,y\in\gh$.

Starting from this data, on can define an aristotelian metric structure \ala:
\begin{description}
\item[carroll] by setting $\bs{\gamma}=\overset{\bs{N}}{\bs{\gamma}}$ and $\bs{\xi}=\bs{N}$
\item[riemann] by setting 
$\bs{g}=\overset{\bs{N}}{\bs{\gamma}}+\bs{\psi}\otimes\bs{\psi}$
\item[lorentz] by setting 
$\bs{g}=\overset{\bs{N}}{\bs{\gamma}}-\bs{\psi}\otimes\bs{\psi}$.
\end{description}
Alternatively, starting from an aristotelian metric structure \ala lorentz $(\bs{g},\bs{N})$, one can define a galilean structure $(\bs{\psi},\bs{\gamma})$ as $\bs{\psi}=-\bs{g}(\bs{N},-)$ and $\bs{\gamma}=\bs{g}\big|_{\Ker\bs{\psi}}$ and similarly for the riemannian case [with $\bs{\psi}=\bs{g}(\bs{N},-)$].

Finally, starting from an aristotelian metric structure \ala carroll $(\bs{\xi},\bs{\gamma}_C,\bs{A})$, one can define an aristotelian metric structure \ala galilei $(\bs{\psi},\bs{\gamma}_G,\bs{N})$ as $\bs{\psi}=\bs{A}$, $\bs{\gamma}_G=\bs{\gamma}_C\big|_{\Ker\bs{\psi}}$ and $\bs{N}=\bs{\xi}$. 
\end{proof}

The aristotelian Klein pairs associated with the algebras of Table \ref{Table:Universal aristotelian algebras} are all endowed with an aristotelian metric structure, whose different presentations are given in Table \ref{Table:Metric structure on aristotelian Klein pairs}.

\begin{table}[ht]
\centering
 \resizebox{8.5cm}{!}{
\begin{tabular}{c|c}
\textbf{Presentation}&\textbf{Invariant tensors}\\\hline
{galilean}&$\bs{\psi}=-H^*$\qq$\bs{h}=\tP\otimes \tP$\qq $\bs{N}=-H$\\
{carrollian}\cellcolor{Gray}&\cellcolor{Gray}$\bs{\xi}=-H$\qq$\bs{\gamma}=\tP^*\otimes \tP^*$\qq $\bs{A}=-H^*$\\
{riemannian}&$\bs{g}=H^*\otimes H^*+\tP^*\otimes \tP^*$\qq $\bs{N}=-H$\\
{lorentzian}\cellcolor{Gray}&\cellcolor{Gray}$\bs{g}=-H^*\otimes H^*+\tP^*\otimes \tP^*$\qq $\bs{N}=-H$
  \end{tabular}
  }
      \caption{Metric structure on aristotelian Klein pairs}
      \label{Table:Metric structure on aristotelian Klein pairs}
\end{table}

\paragraph{Lifshitzian algebras}

An interesting variation on the notion of aristotelian algebras was recently put forward in \cite{FigueroaOFarrill2022b} in the guise of \emph{lifshitzian algebras}, being scalar extensions of aristotelian algebras. 
\begin{Definition}[Lifshitzian algebra]
\label{Definition:lifshitzian algebra}
A lifshitzian algebra with $d$-dimensional spatial isotropy is a Lie algebra $\alg$ containing a $\so(d)$-subalgebra and admitting the following decomposition of $\so(d)$-modules:
\bea
\alg=M\oplus H\oplus\tP\oplus\tJ\nn
\eea
where the generators $M$ and $H$ are in the scalar representation of $\mathfrak{so}(d)$ and $\tP$ is in the vector representation of $\mathfrak{so}(d)$. 
\end{Definition}
\begin{Remark}
Similarly to the kinematical algebra case \UnskipRef{Remark:indiscernable}, the $\so(d)$-decomposition of lifshitzian algebras does not distinguish between the two scalar modules, so that isomorphisms of lifshitzian algebras map $M_1\oplus H_1$ to $M_2\oplus H_2$, but not necessarily $M_1$ to $M_2$ (resp. $H_1$ to $H_2$).
\end{Remark}
Lifshitzian algebras have recently been classified in \cite[Table 1]{FigueroaOFarrill2022b}, classification that we reproduce (focusing on universal algebras) in Table \ref{Table:Lifshitzian Klein pairs}.

\begin{table}[ht]
\centering
\resizebox{10.5cm}{!}{
\begin{tabular}{l|l|l|ll}
\multicolumn{1}{c|}{\textbf{Label}}&\multicolumn{1}{c|}{\textbf{Comments}}&\multicolumn{1}{c|}{$\br{\alp}{\alp}\subset\alh$}&\multicolumn{2}{c}{$\br{\alp}{\alp}\subset\alp$}\\\hline
\hypertarget{Table:Li1}{$\mathsf{Li1}$}&$\mathfrak{iso}(d)\oplus\mathbb R^2$&&&\\
\rowcolor{Gray}
\hypertarget{Table:Li2}{$\mathsf{Li2}$}&&&$\br{H}{M}=M$&\\
\hypertarget{Table:Li3}{$\mathsf{Li3}_z$}&$z\in\mathbb R$&&$\br{H}{M}=z\,M$&$\br{H}{\tP}=\tP$\\
\rowcolor{Gray}
\hypertarget{Table:Li4}{$\mathsf{Li4}$}&$\mathfrak{so}(d,1)\oplus\mathbb R^2$&$\br{\tP}{\tP}=\tJ$&&\\
\hypertarget{Table:Li5}{$\mathsf{Li5}$}&$\mathfrak{so}(d+1)\oplus\mathbb R^2$&$\br{\tP}{\tP}=-\tJ$&&\\
\rowcolor{Gray}
\hypertarget{Table:Li6}{$\mathsf{Li6}$}&&$\br{\tP}{\tP}=\tJ$&$\br{H}{M}=M$&\\
\hypertarget{Table:Li7}{$\mathsf{Li7}$}&&$\br{\tP}{\tP}=-\tJ$&$\br{H}{M}=M$&
\end{tabular}
  }
         \caption{Universal lifshitzian algebras}
      \label{Table:Lifshitzian Klein pairs}
\end{table}

The motivating example of lifshitzian algebras is given by the \emph{lifshitz algebra} (denoted \hyperlink{Table:Li3}{$\mathsf{Li3}_z$} in Table \ref{Table:Lifshitzian Klein pairs}) which is characterised by an anisotropic scaling between time and space \cite{Kachru2008,Taylor2015}. More precisely, the generator $H$ acts via dilation on the `time'\footnote{Note that, in  \cite[Table 1]{FigueroaOFarrill2022b}, the `dilation' generator $H$ is denoted $D$ while our `time' generator $M$ is denoted $H$. While this notation is more standard, we nevertheless privilege the present notation, having in mind the ambient approach where $M$ needs to span a scalar ideal of $\alg$ \SectionRef{section:Ambient approach to Klein pairs}. } generator $M$ and `space' generators $\tP$, albeit with a generically different scaling $z\in\mathbb R$. As discussed in \cite{FigueroaOFarrill2022b}, one can distinguish between two classes of Klein pairs $(\alg,\alh)$, where $\alg$ is a lifshitzian algebra, depending on the choice of $\alh$. One possible choice is to pick $\alh$ to be the rotational subalgebra $\so(d)$, such that the corresponding quotient $\alg/\alh$ has dimension $d+2$ and is thus endowed with an ambient interpretation. The alternative choice is to select $\alh=\so(d)\oplus \mathbb R$, so that $H$ is considered to be part of the homogeneous subalgebra and takes the interpretation of an anisotropic dilation operator on a $(d+1)$-dimensional geometry \cite{FigueroaOFarrill2022b}. In the present work, we will---perhaps predictably---focus on the `ambient' interpretation $\alh=\so(d)$. Lifshitzian Klein pairs of this first type are readily classified in Table \ref{Table:Lifshitzian Klein pairs}.

\begin{Definition}[Lifshitzian metric structure on Klein pairs]
\label{Definition:Lifshitzian metric structures on Klein pairs}
A Klein pair $(\alg,\alh)$ will be said to be endowed with a lifshitzian metric structure if $\gh$ is endowed with a tuple $(\bs{\xi},\bs{N}, \bs{\psi},\bs{A},\overset{\bs{N}}{\bs{\gamma}})$ where:
\be
\item $\bs{\xi}, \bs{N}\in\gh$ are non-zero $\ad_\alh$-invariant vectors
\item $\bs{{\psi}}, \bs{{A}}\in(\gh)^*$ are non-zero $\ad_\alh$-invariant linear forms satisfying:
 \[
 \begin{tikzcd}[row sep = small,column sep = small,ampersand replacement=\&]
\bs{{\psi}}(\bs{{\xi}})=0\&\bs{{A}}(\bs{{N}})=0\&\bs{{\psi}}(\bs{{N}})=1\&\bs{{A}}(\bs{{\xi}})=1
\end{tikzcd}
 \]
 \vspace{-8mm}
\item $\overset{\bs{N}}{\bs{\gamma}}\in\big((\gh)^*\big){}^{\otimes2}$ is a symmetric, positive semi-definite and $\ad_\alh$-invariant bilinear form whose radical is spanned by $\bs{\xi}$ and $\bs{N}$.
\ee
\end{Definition}
The above definition is straightforwardly exemplified by the lifshitzian Klein pairs compiled in Table \ref{Table:Lifshitzian Klein pairs}, each of them being endowed with the canonical lifshitzian metric structure defined as:
\[
 \begin{tikzcd}[row sep = small,column sep = small,ampersand replacement=\&]
\bs{\xi}=M\&\bs{N}=-H\&\bs{\psi}=-H^*\&\bs{A}=M^*\&\overset{\bs{N}}{\bs{\gamma}}=\tP^*\otimes \tP^*\, .
\end{tikzcd}
\]
As noted in \cite[Section 5]{FigueroaOFarrill2022b}, the rank $d$ bilinear form $\overset{\bs{N}}{\bs{\gamma}}$ can be used to construct various $\ad_\alh$-invariant bilinear forms of rank $d+2$, being of signature $(0,d+2)$, $(1,d+1)$ or $(2,d)$. Of particular interest is the lorentzian bilinear form $\bs{g}:=\overset{\bs{N}}{\bs{\gamma}}+\bs{\psi}\otimes \bs{A}+\bs{A}\otimes\bs{\psi}$, which [together with $\bs{\xi}$] endows $(\alg,\alh)$ with a bargmannian metric structure \big[\cf Definition \ref{Definition:Ambient aristotelian metric structures on Klein pairs} and Figure \ref{Figure:Hierarchy of ambient geometries}\big]. 
\section{A Tangled Tale: $(s,v)$-Lie algebras}
\label{section:A Tangled Tale: (s,v)-Lie algebras}
Having reviewed the notions of kinematical and aristotelian algebras, we are now interested in the generalisation of these concepts to the ambient setting. 
In order to both motivate our definitions and account for the various Lie algebra classes relevant for this work in a unified manner, we introduce the following terminology:
\begin{Definition}[{\footnotesize(}$s,v${\footnotesize)}-Lie algebra]Let $(s,v)\in\mathbb N^{\times 2}$ be a pair of integers. A $(s,v)$-Lie algebra (with $d$-dimensional spatial isotropy) is a Lie algebra $\alg$ containing a $\so(d)$-subalgebra and admitting the following decomposition of $\so(d)$-modules:
\vspace{-6mm}

\bea
\alg=\bigoplus_{a=1}^{s}\tSs{a}\oplus\bigoplus_{b=1}^v\tVs{b}\oplus\tJ\label{equation:decomposition of (s,v)-Lie algebras}
\eea
\vspace{-8mm}

where:
\be
\item The generators $\tJ$ span the $\so(d)$-subalgebra
 \ie $\br{\tJ}{\tJ}\sim \tJ$.
\item The generators $\tVs{b}$'s are in the vector representation of $\mathfrak{so}(d)$, \ie $\br{\tJ}{\tVs{b}}\sim \tVs{b}$.
\item The generators $\tSs{a}$'s are in the scalar representation of $\mathfrak{so}(d)$, \ie $\br{\tJ}{\tSs{a}}\sim0$.
\ee
\end{Definition}

The sequence $(\tSs{1},\ldots,\tSs{s},\tVs{1},\ldots,\tVs{v},\tJ)$ is called an \emph{adapted basis} of $\alg$. The set of adapted bases is a principal homogeneous space for $\GLR{s}\times\GLR{v}\subset\GLR{\dim \alg}$ where $\dim\alg=s+vd+\frac{d(d-1)}{2}$.

Two $(s,v)$-Lie algebras $\alg_1,\alg_2$ are isomorphic if there is an isomorphism of Lie algebras $\varphi\From\alg_1\to\alg_2$ preserving the $\so(d)$-decomposition, \ie:
\[
 \begin{tikzcd}[row sep = small,column sep = small,ampersand replacement=\&]
\varphi(\tJ_1)=\tJ_2\&\displaystyle\varphi(\bigoplus_{a=1}^{s}\tSs{a}_1)=\bigoplus_{a=1}^{s}\tSs{a}_2\&\displaystyle\varphi(\bigoplus_{b=1}^v\tVs{b}_1)=\bigoplus_{b=1}^v\tVs{b}_2\,.
\end{tikzcd}
\]
\begin{Example}
We recover the notions of kinematical \UnskipRef{Definition:Kinematical algebra} and aristotelian \UnskipRef{Definition:aristotelian algebra} algebras as $(1,2)$-Lie algebras and $(1,1)$-Lie algebras, respectively while lifshitzian algebras \UnskipRef{Definition:lifshitzian algebra} identify with $(2,1)$-Lie algebras.
\end{Example}

\begin{Definition}[Subspace of type {\footnotesize(}$t,w,j${\footnotesize)}]
\label{Definition:Subspace of type}
Let $\alg$ be a $(s,v)$-Lie algebra with $d$-dimensional spatial isotropy. A subspace $\alh\subset\alg$ of $\alg$ preserving the $\so(d)$-decomposition of $\alg$ will be said to be of type $(t,w,j)$ if it contains $0\leq t\leq s$ scalar generators, $0\leq w\leq v$ sets of vector generators and $0\leq j\leq1$ set of $\so(d)$-generators.
\end{Definition}

\begin{Example}
A kinematical Klein pair is a pair $(\alg,\alh)$ composed of a $(1,2)$-Lie algebra $\alg$ together with a subalgebra $\alh\subset\alg$ of type $(0,1,1)$.
\end{Example}

\medskip
Different values of the pair of integers $(s,v)$ will yield different classes of Lie algebras. In the following, we will restrict our attention to two particular subclasses obtained by restriction of dimension---dubbed \emph{$k$-kinematical algebras} and \emph{$k$-aristotelian algebras} in the following---generalising the usual notions of kinematical and aristotelian algebras, respectively. 

\paragraph{$k$-kinematical algebras}
In lorentzian geometry, a lorentzian spacetime $(\M,g)$ of dimension $n$ is said to be \emph{maximally symmetric} if the isometry algebra of Killing vector fields  $\displaystyle\alg=\pset{X\in\vfh\ |\ \Lag_X{g}=0}$ has maximal dimension $\dim\alg=\frac{n(n+1)}{2}$. As is well-known, there are three such maximally symmetric spacetimes in every dimension $n$, with associated isometry algebras the lorentzian kinematical algebras of Table \ref{Table:Effective kinematical Klein pairs}, with $d=n-1$. Motivated by this, we define the following notion of $(s,v)$-Lie algebra of maximal dimension. 
\begin{Definition}[{\footnotesize(}$s,v${\footnotesize)}-Lie algebra of maximal dimension]
A $(s,v)$-Lie algebra with $d$-dimensional spatial isotropy will be said to be of maximal dimension if there exists an integer $n\in\mathbb N$ such that $\dim\alg=\frac{n(n+1)}{2}$. 

The integer $n$ will be referred to as the \emph{spacetime dimension} of $\alg$.
\end{Definition}

\begin{Definition}[Universally symmetric pair of integers]
A pair of integers $(s,v)$ will be said \emph{universally} symmetric if, for all $d\in\mathbb N$, any $(s,v)$-Lie algebra with $d$-dimensional spatial isotropy is of maximal dimension.
\end{Definition}

\begin{Proposition}[$k$-kinematical algebras]
\label{Proposition:k-kinematical algebras}
Universally symmetric pairs of integers $(s,v)$ are classified by a single integer $k\in\mathbb N$. Corresponding $(s,v)$-Lie algebras will be called $k$-kinematical algebras.
\end{Proposition}
\begin{proof}
The vector space underlying a $(s,v)$-Lie algebra with $d$-dimensional spatial isotropy has dimension $\dim\alg=\frac{d(d-1)}{2}+v\, d+s$. Setting $\dim\alg=\frac{n(n+1)}{2}$ and assuming $s$ and $v$ are independent of $d$ yields $s=\frac{k(k+1)}{2}$, $v=k+1$ and $n=d+k$.
\end{proof}
A $k$-kinematical algebra thus admits the following $\so(d)$-decomposition:
\bea
\alg=\bigoplus_{a=1}^{\frac{k(k+1)}{2}}\tSs{a}\oplus\bigoplus_{b=1}^{k+1}\tVs{b}\oplus\tJ\nn
\eea
and is of dimension $\dim\alg=\frac{n(n+1)}{2}$ with $n=d+k$.

\begin{Example}[$k$-kinematical algebras]
\hfill
\bi
\item $0$-kinematical algebras are spanned by $\pset{\tVs{1}=\tP,\tJ}$. There are three $0$-kinematical algebras, namely the euclidean, spherical and hyperbolic algebras in $n=d$ spacetime dimensions.
\item $1$-kinematical algebras are spanned by $\pset{\tSs{1}=H,\tVs{1}=\tP,\tVs{2}=\tK,\tJ}$ and thus coincide with kinematical algebras \UnskipRef{Definition:Kinematical algebra}. 
\item $2$-kinematical algebras are spanned by $\pset{\tSs{1}=H,\tSs{2}=M,\tSs{3}=C,\tVs{1}=\tP,\tVs{2}=\tK,\tVs{3}=\tD,\tJ}$ and will be referred to as \emph{ambient kinematical algebras}\footnote{$2$-kinematical algebras were referred to as \emph{generalised conformal algebras} in \cite{Figueroa-OFarrill2018a}.} \UnskipRef{Definition:Ambient kinematical algebra}. Examples include the leibniz algebra \UnskipRef{Example:Leibniz Klein pair} and variations thereof in $n=d+2$ spacetime dimensions [see Table \ref{Table:Reductive leibnizian lifts of galilean algebras}]. Other examples include the notion of \emph{graded conformal algebra} introduced in \cite{Figueroa-OFarrill2018a}.
\ei
\end{Example}

\begin{Definition}[$k$-kinematical Klein pairs]
\label{Definition:k-kinematical Klein pairs}
A Klein pair $(\alg,\alh)$ such that $\alg$ is a $k$-kinematical algebra with $d$-dimensional spatial isotropy and $\alh$ is a $(k-1)$-kinematical algebra with $d$-dimensional spatial isotropy is called a $k$-kinematical Klein pair with $d$-dimensional spatial isotropy.
\end{Definition}

An isomorphism of $k$-kinematical Klein pairs is an isomorphism of Klein pairs that preserves the $\so(d)$-decomposition.

\begin{Remark}
Note that $\dim\alg-\dim\alh=n=d+k$, hence the base manifold of a Cartan geometry modelled on a $k$-kinematical Klein pair with $d$-dimensional spatial isotropy has dimension $d+k$.

The notion of $1$-kinematical Klein pair identifies with the usual one of kinematical Klein pair whose corresponding geometric structures live on a manifold of dimension $d+1$ corresponding to ordinary spacetime. Accordingly, $2$-kinematical Klein pairs will be referred to as \emph{ambient kinematical Klein pairs} \UnskipRef{Definition:Ambient kinematical Klein pairs} where the term ambient is justified by the fact that the underlying geometry is of dimension $d+2$.
\end{Remark}

\paragraph{$k$-aristotelian algebras}
We now transpose the previous discussion from the kinematical case to the  aristotelian case. We will say that a $(s,v)$-Lie algebra with $d$-dimensional spatial isotropy is of \emph{boostless maximal dimension} if there exists an integer $n\in\mathbb N$ such that $\dim\alg=\frac{n(n-1)}{2}+1$. Compared to the maximally symmetric case, we lost $n-1$ generators which corresponds to the intuitive dimension of a boost generator in spacetime dimension $n$.
A pair of integers $(s,v)$ will be said \emph{universally} boostless symmetric if for all $d\in\mathbb N$, $(s,v)$-Lie algebras with $d$-dimensional spatial isotropy are of boostless maximal dimension.

\begin{Proposition}[$k$-aristotelian algebras]
\label{Proposition:k-aristotelian algebras}
Universally boostless symmetric pairs of integers $(s,v)$ are classified by a single integer $k\in\mathbb N$. Corresponding $(s,v)$-Lie algebras will be called $k$-aristotelian algebras.
\end{Proposition}
\begin{proof}
The vector space underlying a $(s,v)$-Lie algebra with $d$-dimensional spatial isotropy has dimension $\dim\alg=\frac{d(d-1)}{2}+v\, d+s$. Setting $\dim\alg=\frac{n(n-1)}{2}+1$ and assuming $s$ and $v$ are independent of $d$ yields $s=\frac{k(k-1)}{2}+1$, $v=k$ and $n=d+k$.
\end{proof}

A $k$-aristotelian algebra is thus a Lie algebra $\alg$ of dimension $\dim\alg=\frac{n(n-1)}{2}+1$ with $n=d+k$ and admitting the following $\so(d)$-decomposition:
\bea
\alg=\bigoplus_{a=1}^{\frac{k(k-1)}{2}+1}\tSs{a}\oplus\bigoplus_{b=1}^{k}\tVs{b}\oplus\tJ\nn.
\eea

\begin{Example}[$k$-aristotelian algebras]
\hfill
\bi
\item There is a unique $0$-aristotelian algebra, namely the direct sum $\mathbb R\oplus\so(d)$.
\item $1$-aristotelian algebras are spanned by $\pset{\tSs{1}=H,\tVs{1}=\tP,\tJ}$ and thus coincide with aristotelian algebras \UnskipRef{Definition:aristotelian algebra}.
\item $2$-aristotelian algebras are spanned by $\pset{\tSs{1}=H,\tSs{2}=M,\tVs{1}=\tP,\tVs{2}=\tK,\tJ}$ and will be referred to as \emph{ambient aristotelian} algebras\footnote{$2$-aristotelian algebras are referred to as \emph{boost-extended Lifshitz algebras} in \cite{FigueroaOFarrill2023}, following \cite{Gibbons2009}.} \UnskipRef{Definition:Ambient aristotelian algebra}. Examples include the bargmann algebra \UnskipRef{Example:Bargmann Klein pair} and variations thereof in $n=d+2$ spacetime dimensions [see Table \ref{Table:Ambient aristotelian algebras admitting a scalar ideal}]. Other examples include the notion of \emph{generalised Lifshitz algebra} introduced in \cite{Figueroa-OFarrill2018a}.
\ei
\end{Example}

\begin{Definition}[$k$-aristotelian Klein pairs]
\label{Definition:k-aristotelian Klein pairs}
A Klein pair $(\alg,\alh)$ such that $\alg$ is a $k$-aristotelian algebra with $d$-dimensional spatial isotropy and $\alh$ is a $(k-2)$-kinematical algebra with $d$-dimensional spatial isotropy is called a $k$-aristotelian Klein pair with $d$-dimensional spatial isotropy.
\end{Definition}

\begin{Remark}
Note that $\dim\alg-\dim\alh=d+k=n$, similarly as in the $k$-kinematical case. The notion of $1$-aristotelian Klein pair identifies with the usual notion of aristotelian Klein pair whereas $2$-aristotelian Klein pairs will be referred to as \emph{ambient aristotelian Klein pairs} \UnskipRef{Definition:Ambient aristotelian Klein pairs}. 
\end{Remark}

\paragraph{Branching rules}
We conclude this section by discussing how to relate Lie algebras across distinct spatial isotropy dimensions through the use of branching rules. Such dimensional reduction scheme will prove useful \SectionRef{section:Drink me} in order to obtain ambient kinematical algebras from ordinary kinematical algebras.

\begin{Proposition}
\label{Proposition:general branching rules}
Any $(s,v)$-Lie algebra $\alg$ with $(d+1)$-dimensional spatial isotropy is isomorphic as a Lie algebra to a $(s+v,v+1)$-Lie algebra $\overline \alg$ with $d$-dimensional spatial isotropy. Any subspace $\alh\subset\alg$ of type $(t,w,j)$ is isomorphic as a vector space to a subspace $\overline\alh\subset\overline\alg$ of type $(s+v,v+j,j)$. 
\end{Proposition}
\begin{proof}
Using branching rules allows to decompose representations of $\so(d+1)$ into representations of $\so(d)$. Explicitly, we decompose the generators of $\so(d+1)$ as $J_{\mu\nu}=\pset{V^{(v+1)}_{i}=J_{0i},J_{ij}}$ and the vectorial generators as $V_{\mu}^{(b)}=\pset{\tSs{s+b}=V_{0}^{(b)},V_{i}^{(b)}}$, respectively, with $\mu,\nu\in\pset{0,\ldots,d}$ and $i,j\in\pset{1,\ldots d}$. 

Straightforward counting shows that a $(s,v)$-Lie algebra \big(resp. a subspace of type $(t,w,j)$\big) with $(d+1)$-dimensional spatial isotropy induces a $(s+v,v+1)$-Lie algebra \big(resp. a subspace of type $(s+v,v+j,j)$\big) with $d$-dimensional spatial isotropy.
\end{proof}

\begin{Remark}
\label{Remark:subalgebra}
Starting from a $(s,v)$-Lie algebra $\alg$ with $(d+1)$-dimensional spatial isotropy:
\bea
\alg=\bigoplus_{a=1}^{s}\tSs{a}\oplus\bigoplus_{b=1}^v\tVs{b}\oplus\tJ\nn
\eea
and using branching rules yields the $(s+v,v+1)$-Lie algebra $\overline \alg$ with $d$-dimensional spatial isotropy:
\bea
\bar\alg=\bigoplus_{a=1}^{s}\tSs{a}\oplus \bigoplus_{b=1}^{v}\Spannn{V^{(b)}_0}\oplus\bigoplus_{b=1}^v\tVs{b}\oplus \Spannn{J_{0i}}\oplus\tJ\, .\nn
\eea
Note that the subspace:
\bea
\alg_0=\bigoplus_{a=1}^{s}\tSs{a}\oplus\bigoplus_{b=1}^v\tVs{b}\oplus\tJ\nn
\eea
is a $(s,v,1)$-subalgebra of $\overline \alg$.
\end{Remark}

\begin{Corollary}
\label{corollary:dimensional reduction of Klein pairs}
Any $k$-kinematical Klein pair (resp. $k$-aristotelian Klein pair) with $(d+1)$-dimensional spatial isotropy is isomorphic as a Klein pair to a $(k+1)$-kinematical Klein pair (resp. $(k+1)$-aristotelian Klein pair) with $d$-dimensional spatial isotropy.
\end{Corollary}
\begin{proof}
Let $(\alg,\alh)$ be a $k$-kinematical Klein pair with $(d+1)$-dimensional spatial isotropy. By definitions \UnskipRefs{Proposition:k-kinematical algebras}{Definition:k-kinematical Klein pairs}, $\alg$ is a $(\frac{k(k+1)}{2},k)$-algebra and $\alh$ a $(\frac{k(k-1)}{2},k,1)$-subalgebra thereof. According to \Prop{Proposition:general branching rules}, the Lie algebra $\alg$ (resp. subspace $\alh$) is isomorphic to a $(\frac{(k+1)(k+2)}{2},k+2)$-algebra \big(resp. a $(\frac{k(k+1)}{2},k+1,1)$-subspace\big) with $d$-dimensional spatial isotropy. Hence the pair $(\alg,\alh)$ is isomorphic as a Klein pair to a $(k+1)$-kinematical Klein pair with $d$-dimensional spatial isotropy.

Now, letting $(\alg,\alh)$ be a $k$-aristotelian Klein pair with $(d+1)$-dimensional spatial isotropy, $\alg$ is a $(\frac{k(k-1)}{2}+1,k)$-algebra isomorphic to a $(\frac{k(k+1)}{2},k+1)$-algebra with $d$-dimensional spatial isotropy \ie to a $(k+1)$-aristotelian algebra with $d$-dimensional spatial isotropy. Furthermore, by the previous argument, the $(k-2)$-kinematical algebra $\alh$ is isomorphic to a $(k-1)$-kinematical algebra with $d$-dimensional spatial isotropy, so that the pair $(\alg,\alh)$ is indeed isomorphic to a $(k+1)$-aristotelian Klein pair with $d$-dimensional spatial isotropy.
\end{proof}
\begin{Example}[Dimensional reduction of $k$-kinematical algebras]
 \label{Example:dimensional reduction}
 \hfill
 
\medskip
\textbf{From $k=0$ to $k=1$}

\medskip
There are three universal $0$-kinematical Klein pairs $(\alg,\alh)$ with spatial isotropy dimension $d+1$, where $\alg=\alh\oplus\alp$  [with $\alp=\Spannn{P_\mu}$ and $\alh=\Spannn{J_{\mu\nu}}$, with $\mu,\nu\in\pset{0,\ldots,d+1}$] with non-trivial commutation relations\footnote{On top of the canonical ones $\br{\tJ}{\tP}\sim\tP$ and $\br{\tJ}{\tJ}\sim\tJ$.} given by:
\bea
 \resizebox{4cm}{!}{
\begin{tabular}{l|l}
\textbf{Klein pair}&$\br{\alp}{\alp}\subset\alh$\\\hline
\textnormal{euclidean}&\\
\textnormal{sphere}\cellcolor{Gray}&\cellcolor{Gray}$\br{\tP}{\tP}=-\tJ$\\
\textnormal{hyperbolic}&$\br{\tP}{\tP}=\tJ$
  \end{tabular}
  }
  \label{equation:0-kinematical Klein pairs}
  \eea
   
Performing the decomposition $\mu=\pset{0,i}$, where $i\in\pset{1,\ldots,d}$ together with the following relabelling:
\bea
P_0\mapsto H\quad,\quad  J_{0i}\mapsto K_i\nn
\eea
allows to map the three above $0$-kinematical Klein pairs with $(d+1)$-dimensional spatial isotropy to three $1$-kinematical Klein pairs with $d$-dimensional spatial isotropy:
\bea
\alg=\alh\oplus\alp \text{ with }\alp=\Spannn{H,P_i}\text{ and }\alh=\Spannn{K_i,J_{ij}}\nn
\eea
together with commutation relations\footnote{The commutation relations displayed in the table are obtained from the following replacements:
\bea
&&\br{P_\mu}{P_\nu}=\epsilon J_{\mu\nu}\quad \mapsto\quad \br{H}{P_i}=\epsilon K_i\qq\br{P_i}{P_j}=\epsilon J_{ij}\nn\\
&&\br{J_{\mu\nu}}{P_\rho}=-\delta_{\mu\rho}P_\nu+\delta_{\nu \rho}P_\mu\quad\mapsto\quad \br{K_i}{H}=-P_i\qq\br{K_i}{P_j}=\delta_{ij}H\qq \br{J_{ij}}{P_k}=-\delta_{ik}P_j+\delta_{j k}P_i\nn\\
&&\br{J_{\mu\nu}}{J_{\rho\sigma}}=\delta_{\nu\rho}J_{\mu\sigma}-\delta_{\nu\sigma}J_{\mu\rho}+\delta_{\mu\sigma}J_{\nu\rho}-\delta_{\mu\rho}J_{\nu\sigma}
\quad\mapsto\quad\nn\\
&&\br{K_i}{K_j}=-J_{ij}\qq\br{J_{ij}}{K_k}=-\delta_{ik}K_j+\delta_{j k}K_i\qq\br{J_{ij}}{J_{kl}}=\delta_{jk}J_{il}-\delta_{jl}J_{ik}+\delta_{il}J_{jk}-\delta_{ik}J_{jl}\, .\nn
\eea
 }:
\bea
 \resizebox{12cm}{!}{
\begin{tabular}{l|l|ll|ll}
\textbf{Klein pair}&$\br{\alh}{\alh}\subset\alh$&\multicolumn{2}{c|}{$\br{\alh}{\alp}\subset\alp$}&\multicolumn{2}{c}{$\br{\alp}{\alp}\subset\alh$}\\
\hline
\textnormal{euclidean}&$\br{\tK}{\tK}=-\tJ$&$\br{\tK}{\tP}= H$&$\br{\tK}{H}= -\tP$&&\\
\textnormal{sphere}\cellcolor{Gray}&$\br{\tK}{\tK}=-\tJ$\cellcolor{Gray}&$\br{\tK}{\tP}= H$\cellcolor{Gray}&$\br{\tK}{H}= -\tP$\cellcolor{Gray}&$\br{H}{\tP}=-\tK$\cellcolor{Gray}&\cellcolor{Gray}$\br{\tP}{\tP}=-\tJ$\\
\textnormal{hyperbolic}&$\br{\tK}{\tK}=-\tJ$&$\br{\tK}{\tP}= H$&$\br{\tK}{H}= -\tP$&$\br{H}{\tP}=\tK$&$\br{\tP}{\tP}=\tJ$
  \end{tabular}
  }
  \hspace{1cm}
    \label{equation:induced 1-kinematical Klein pairs}
\eea
\end{Example}

\begin{Remark}
Allowing for a different signature of $\alh$ allows to recover the minkowski, de sitter and anti de sitter Klein pairs of Table \ref{Table:Effective kinematical Klein pairs} as the $1$-kinematical Klein pairs with $d$-dimensional spatial isotropy associated with the $0$-kinematical Klein pairs with $(d+1)$-dimensional spatial isotropy \eqref{equation:0-kinematical Klein pairs} with $\tJ=\so(1,d-1)$. 
\end{Remark}

Moving from the familiar case of Example \ref{Example:dimensional reduction} to the next to simple case (\ie from $k=1$ to $k=2$) will allow us to obtain ambient kinematical Klein pairs (specifically, leibnizian Klein pairs) from ordinary kinematical Klein pair (namely, galilean Klein pairs), \cf Section \ref{section:Drink me} for details.

\section{Climbing up one leg of the table: an ambient perspective on Klein pairs}
\label{section:Climbing up one leg of the table: an ambient perspective on Klein pairs}
Pray, let us now turn our gaze to the current section, wherein we shall aim at particularising the general discussion of Section \ref{section:A Tangled Tale: (s,v)-Lie algebras} by taking a closer look at the notions of 2-kinematical algebras \UnskipRef{Proposition:k-kinematical algebras} and 2-aristotelian algebras \UnskipRef{Proposition:k-aristotelian algebras}, whose associated notions of Klein pairs \UnskipRefs{Definition:k-kinematical Klein pairs}{Definition:k-aristotelian Klein pairs} will henceforth constitute our main object of interest.
\subsection{Ambient Klein pairs}

We start by unfolding the notion of 2-kinematical algebra---suggestively dubbed \emph{ambient kinematical algebra}---by displaying a definition modelled on the one of kinematical algebra \UnskipRef{Definition:Kinematical algebra}:

\begin{Definition}[Ambient kinematical algebra]
\label{Definition:Ambient kinematical algebra}
An ambient kinematical algebra (or 2-kinematical algebra) with $d$-dimensional spatial isotropy is a Lie algebra $\alg$ containing a $\so(d)$-subalgebra and admitting the following decomposition of $\so(d)$-modules:
\vspace{-4mm}

\bea
\alg=M\oplus H\oplus C\oplus \tP\oplus\tD\oplus\tK\oplus\tJ\label{equation:ambient kinematical algebra}
\eea
\vspace{-8mm}

where:
\be
\item The generators $\tJ$ span the $\so(d)$-subalgebra.
\item The generators $\tP$, $\tD$ and $\tK$ are in the vector representation of $\mathfrak{so}(d)$.
\item The generators $M,H$ and $C$ are in the scalar representation of $\mathfrak{so}(d)$.
\ee
\end{Definition}

Two ambient kinematical algebras $\alg_1,\alg_2$ are isomorphic if there is an isomorphism of Lie algebras $\varphi\From\alg_1\to\alg_2$ preserving the $\so(d)$-decomposition.\footnote{Similarly to the kinematical algebra case \UnskipRef{Remark:indiscernable}, the $\so(d)$-decomposition of ambient kinematical algebras  treats the three scalar modules (resp. the three vector modules) as indiscernable, so that isomorphisms of ambient kinematical algebras map $M_1\oplus H_1\oplus C_1$ to $M_2\oplus H_2\oplus C_2$ (resp. $\tP_1\oplus\tD_1\oplus\tK_1$ to $\tP_2\oplus\tD_2\oplus\tK_2$), but not necessarily $M_1$ to $M_2$, $H_1$ to $H_2$ nor $C_1$ to $C_2$ (resp. $\tP_1$ to $\tP_2$, $\tD_1$ to $\tD_2$ nor $\tK_1$ to $\tK_2$).} Similarly as their (non-ambient) kinematical counterpart, ambient kinematical algebras can be upgraded to a notion of ambient kinematical Klein pairs, defined as 2-kinematical Klein pairs \UnskipRef{Definition:k-kinematical Klein pairs}, for which the subalgebra is a kinematical algebra. Unfolding the above definition yields:
\begin{Definition}[Ambient kinematical Klein pair]
\label{Definition:Ambient kinematical Klein pairs}
A Klein pair $(\alg,\alh)$ such that $\alg$ (resp. $\alh)$ is an ambient kinematical (resp. kinematical) Lie algebra with $d$-dimensional spatial isotropy is called an ambient kinematical Klein pair with $d$-dimensional spatial isotropy.
\end{Definition}
By convention, we will choose the following decomposition of $\so(d)$-modules for the subalgebra $\alh$:
\bea
\alh=C\oplus\tD\oplus\tK\oplus\tJ\label{equation:ambient kinematical subalgebra}.
\eea

The motivating example of ambient kinematical Klein pair is given by the \emph{leibniz Klein pair} with $\alg$ the leibniz algebra introduced in \cite{Bekaert2015b} as isometry algebra of the flat leibniz manifold [see \cite{Bekaert2015b,Morand2023} for details]. 
\begin{Example}[Leibniz Klein pair]
\label{Example:Leibniz Klein pair}
The leibniz Klein pair is defined as the pair $(\alg,\alh)$ where $\alg$ stands for the {leibniz} algebra, defined as the ambient kinematical algebra \eqref{equation:ambient kinematical algebra}
with non-trivial commutators:
\bea
\label{equation:leibniz algebra}
\begin{tikzcd}
\br{\tD}{\tK}= C&\br{\tK}{H}= \tP& \br{\tD}{\tP}= M&\br{C}{H}= M
\end{tikzcd}
\eea
and $\alh$ the subalgebra \eqref{equation:ambient kinematical subalgebra}. The leibniz Klein pair can be checked to be reductive as it admits the decomposition of $\alh$-modules given by $\alg=\alh\oplus\alp$, where $\alp=M\oplus H\oplus \tP$. The $\alh$-module $\alp$ being abelian, the leibniz Klein pair is flat and symmetric [\cf \UnskipRef{Definition:Infinitesimal Klein pair} for definitions]. The pair is also effective and can be geometrically realised as isometry algebra of the flat leibniz manifold in $(d+2)$-dimensions \cite{Bekaert2015b,Morand2023}. Denoting $\ali$ the abelian ideal $\ali= M\oplus C\oplus\tD$, the following sequence of Lie algebras:
\[
\begin{tikzcd}
0 \arrow[r] &\ali  \arrow[r, hook, "i"] & \alg \arrow[r, two heads,"\pi"]        & \alg_0\arrow[r]  & 0
\end{tikzcd}
\]
where $\alg_0$ stands for the galilei algebra \hyperlink{Table:K3}{$\mathsf{K3}$}, is exact [\cf Example \ref{Example:Leibniz projectable triplet} below].
\end{Example}

Ambient kinematical algebras, and specifically the leibniz algebras and generalisations thereof, will play a  key role in Section \ref{section:Drink me}  as alternative ways to lift galilean algebras. 

We now introduce the main notion relevant for the present work, namely \emph{ambient aristotelian algebras}, whose definition as 2-aristotelian algebras \UnskipRef{Proposition:k-aristotelian algebras} can be unfolded as follows:

\begin{Definition}[Ambient aristotelian algebra]
\label{Definition:Ambient aristotelian algebra}
An ambient aristotelian algebra (or 2-aristotelian algebra) with $d$-dimensional spatial isotropy is a Lie algebra $\alg$ containing a $\so(d)$-subalgebra and admitting the following decomposition of $\so(d)$-modules:
\vspace{-4mm}

\bea
\label{equation:ambient aristotelian algebra}
\alg=M\oplus H\oplus \tP\oplus\tK\oplus\tJ
\eea
\vspace{-8mm}

where:
\be
\item The generators $\tJ$ span the $\so(d)$-subalgebra.
\item The generators $\tP$ and $\tK$ are in the vector representation of $\mathfrak{so}(d)$.
\item The generators $M$ and $H$ are in the scalar representation of $\mathfrak{so}(d)$.
\ee
\end{Definition}

The associated notion of Klein pair coincides with the notion of 2-aristotelian Klein pair \UnskipRef{Definition:k-aristotelian Klein pairs}:
\begin{Definition}[Ambient aristotelian Klein pair]
\label{Definition:Ambient aristotelian Klein pairs}
A Klein pair $(\alg,\alh)$ such that $\alg$ (resp. $\alh)$ is an ambient aristotelian (resp. 0-kinematical) Lie algebra with $d$-dimensional spatial isotropy is called an ambient aristotelian Klein pair with $d$-dimensional spatial isotropy.
\end{Definition}
By convention, we will choose the following decomposition of $\so(d)$-modules for the subalgebra $\alh$:
\bea
\label{equation:aristotelian subalgebra}
\alh=\tK\oplus\tJ\, .
\eea
The paradigmatic example of ambient aristotelian Klein pair is given by the \emph{bargmann Klein pair}:
\begin{Example}[Bargmann Klein pair]
\label{Example:Bargmann Klein pair}
Let $(\alg,\alh)$ be the bargmann Klein pair with $\alg$ the bargmann algebra \cite{Bargmann1952,Duval1977} with underlying vector space \eqref{equation:ambient aristotelian algebra} and non-trivial commutators:
\bea
\label{equation:bargmann algebra}
\begin{tikzcd}
\br{\tK}{H}= \tP&\br{\tK}{\tP}= M
\end{tikzcd}
\eea
and $\alh$ the subalgebra \eqref{equation:aristotelian subalgebra}. The resulting Klein pair is reductive, flat, symmetric and effective and can be realised as the isometry algebra of the flat bargmann manifold.
Note that the generator $M$ belongs to the center of $\alg$, hence the Lie algebra extension encoded by the following short exact sequence of Lie algebras:
\[
\begin{tikzcd}
0 \arrow[r] & \Span M  \arrow[r, hook, "i"] & \alg \arrow[r, two heads,"\pi"]        & \alg_0\arrow[r]  & 0
\end{tikzcd}
\]
where $\alg_0$ stands for the galilei algebra \hyperlink{Table:K3}{$\mathsf{K3}$}, is central.
\end{Example}

The various classes of Klein pairs used in this work are summed up in Table \ref{Table:Summary of Klein pairs}.

\begin{table}[ht]
\centering
\resizebox{14cm}{!}{
\begin{tabular}{c|l|l|l|l|c|c}
\multicolumn{1}{c|}{\textbf{Dimension of $\alg/\alh$}}&\multicolumn{1}{c|}{\textbf{Label}}&\multicolumn{1}{c|}{\textbf{Example}}&\multicolumn{1}{c|}{$\alg$}&\multicolumn{1}{c|}{$\alh$}&\multicolumn{1}{c}{$\dim\alg$}&\multicolumn{1}{c}{$\dim\alh$}\\\hline
\multirow{3}{*}{$d+2$}&ambient kinematical&leibniz&$M\oplus H\oplus C\oplus \tP\oplus\tD\oplus\tK\oplus\tJ$&$C\oplus \tD\oplus \tK\oplus\tJ$&$\frac{(d+2)(d+3)}{2}$&$\frac{(d+1)(d+2)}{2}$\\
&\cellcolor{Gray}ambient aristotelian&\cellcolor{Gray}bargmann&\cellcolor{Gray}$M\oplus H\oplus \tP\oplus\tK\oplus\tJ$&\cellcolor{Gray}$\tK\oplus\tJ$&\cellcolor{Gray}$\frac{(d+1)(d+2)}{2}+1$&\cellcolor{Gray}$\frac{d(d+1)}{2}$\\
&lifshitzian&lifshitz&$M\oplus H\oplus\tP\oplus\tJ$&$\tJ$&$\frac{d(d+1)}{2}+2$&$\frac{d(d-1)}{2}$\\
\hline
\multirow{2}{*}{$d+1$}&\cellcolor{Gray}kinematical&\cellcolor{Gray}galilei&\cellcolor{Gray}$H\oplus \tP\oplus\tK\oplus\tJ$&\cellcolor{Gray}$\tK\oplus\tJ$&\cellcolor{Gray}$\frac{(d+1)(d+2)}{2}$&\cellcolor{Gray}$\frac{d(d+1)}{2}$\\
&aristotelian&static&$H\oplus\tP\oplus\tJ$&$\tJ$&$\frac{d(d+1)}{2}+1$&$\frac{d(d-1)}{2}$
  \end{tabular}
   }
        \caption{Summary of Klein pairs}
      \label{Table:Summary of Klein pairs}
\end{table}

\paragraph{Ambient metric structures on Klein pairs}

Definition \ref{Definition:Metric structures on Klein pairs} above reviewed four possible metric structures that can live on Klein pairs. The latter metric structures allowed to partition effective kinematical Klein pairs into four classes (\ie galilean, carrollian, riemannian and lorentzian). We now enrich this list with a novel notion of metric structure on Klein pairs that will be proved to live on a particular subclass of ambient kinematical Klein pairs \UnskipRef{Proposition:Leibnizian structure on Klein pairs}.

\begin{Definition}[Leibnizian metric structure on Klein pairs]
\label{Definition:leibnizian metric structures on Klein pairs}
A Klein pair $(\alg,\alh)$ will be said leibnizian if $\gh$ is endowed with a triplet $(\bs{\xi},\bs{\psi},\bs{\gamma})$ where:
\be
\item $\bs{\xi}\in\gh$ is a non-zero $\ad_\alh$-invariant vector
\item $\bs{{\psi}}\in(\gh)^*$ is a non-zero $\ad_\alh$-invariant linear form satisfying $\bs{{\psi}}(\bs{{\xi}})=0$
\item $\bs{\gamma}\in\big((\Ker\bs{\psi})^*\big){}^{\otimes2}$ is a symmetric, positive semi-definite and $\ad_\alh$-invariant bilinear form whose radical is spanned by $\bs{\xi}$ \ie $\bs{\gamma}(v,w)=0$ for all $w\in\Ker\bs{\psi}$ $\Leftrightarrow$ $v\sim\bs{\xi}$.
\ee
\end{Definition}

\begin{Remark}
At the risk of appearing relatively unsurprising considering the chosen terminology, let us note nevertheless that the leibniz Klein pair \UnskipRef{Example:Leibniz Klein pair} is leibnizian, when endowed with the canonical leibnizian metric structure\footnote{\label{footnote:leibnizian criteria}More generally, an ambient kinematical Klein pair can be shown to preserve the canonical leibnizian metric structure \eqref{equation:leibniz metric structure on leibniz Klein pair} if and only if the following commutators vanish:
 \[
 \begin{tikzcd}[row sep = small,column sep = small,ampersand replacement=\&]
\ad_\gh(\bs{\xi})=0\&\br{\tD}{M}\sim\tP\&\br{\tK}{M}\sim\tP\&\br{C}{M}\sim M\&\br{C}{M}\sim H\\
\ad_\gh(\bs{\psi})=0\&\br{\tD}{\tP}\sim H\&\br{\tK}{\tP}\sim H\&\br{C}{M}\sim H\&\br{C}{H}\sim H\\
\ad_\gh(\bs{\gamma})=0\&\br{\tD}{M}\sim \tP\&\br{\tK}{M}\sim \tP\&\br{C}{\tP}\sim \tP\, .\&
\end{tikzcd}
 \]
}:
\bea
\label{equation:leibniz metric structure on leibniz Klein pair}
\begin{tikzcd}[row sep = small,column sep = small,ampersand replacement=\&]
\bs{\xi}=M\&\bs{\psi}=-H^*\&\bs{\gamma}=\tP^*\otimes \tP^*\, .
\end{tikzcd}
\eea
The holonomy principle thus ensures that base manifolds of Cartan geometries modelled on the leibniz Klein pair are leibnizian manifolds \cite{Bekaert2015b,Morand2023}.
\end{Remark}

As previously reviewed, adding supplementary structures to metric structures living on kinematical Klein pairs \UnskipRef{Definition:Metric structures on Klein pairs} yields a notion of aristotelian metric structure \UnskipRef{Definition:aristotelian metric structures on Klein pairs} which naturally lives on aristotelian Klein pairs [\cf Table \ref{Table:Metric structure on aristotelian Klein pairs}]. We now mimic this construction in the ambient case by displaying three subclasses of leibnizian Klein pairs, each obtained by requiring the existence of an extra compatible and $\ad_\alh$-invariant structure on $\gh$:

\begin{Definition}[Ambient aristotelian metric structures on Klein pairs]
\label{Definition:Ambient aristotelian metric structures on Klein pairs}
A leibnizian Klein pair $(\alg,\alh, \bs{\xi},\bs{\psi},\bs{\gamma})$ will be said:
\begin{description}
\item[$\mathsf G$-ambient] if $\gh$ is endowed with an $\ad_\alh$-invariant linear form $\bs{{A}}\in(\gh)^*$ satisfying $\bs{{A}}(\bs{{\xi}})=1$.
\item[$\mathsf C$-ambient] if $\gh$ is endowed with an $\ad_\alh$-invariant vector $\bs{{N}}\in\gh$ satisfying $\bs{{\psi}}(\bs{{N}})=1$.
\item[bargmannian] if $\gh$ is endowed with  an $\ad_\alh$-invariant bilinear form $\bs{g}\in{\vee^2 (\gh)^*}$ satisfying $\bs{g}(\bs{\xi},-)=\bs{\psi}$ and $\bs{g}|_{\Ker\bs{\psi}}=\bs{\gamma}$.\footnote{The second condition reads more explicitly $\bs{g}(v,w)=\bs{\gamma}(v,w)$ for all $v,w\in\Ker\bs{\psi}$.}
\end{description}
\end{Definition}
The intersection between any two of these three classes\footnote{That from the intersection of any two of these three classes, one can construct an instance of the third one can be seen as follows:
\bi
\item Given a bargmannian structure $\bs{g}$ and a $\mathsf C$-ambient structure $\bs{N}$, one defines a $\mathsf G$-ambient structure $\bs{A}:=\bs{g}(\bs{N},-)$.
\item Given a bargmannian structure $\bs{g}$ and a $\mathsf G$-ambient structure $\bs{A}$, one can create a $\mathsf C$-ambient structure $\bs{N}:=\bs{g}\un(\bs{A},-)$.
\item Given a $\mathsf G$-ambient structure $\bs{A}$ and a $\mathsf C$-ambient structure $\bs{N}$, one can construct a bargmannian structure $\bs{g}:=\overset{\bs{N}}{\bs{\gamma}}+\bs{\psi}\otimes \bs{A}+\bs{A}\otimes\bs{\psi}$, where $\overset{\bs{N}}{\bs{\gamma}}$ is the $\ad_\alh$-invariant bilinear form defined in the proof of Proposition \ref{Proposition:Equivalence definitions aristotelian metric structures}. 
\ei
} coincides with the notion of lifshitzian metric structure \UnskipRef{Definition:Lifshitzian metric structures on Klein pairs}, see Figure \ref{Figure:Hierarchy of ambient geometries}.\footnote{Note that the domain at the intersection of any two disks which is not contained in lifshitz is empty.}

  \definecolor{antiquewhite}{rgb}{0.98, 0.92, 0.84}
\begin{center}
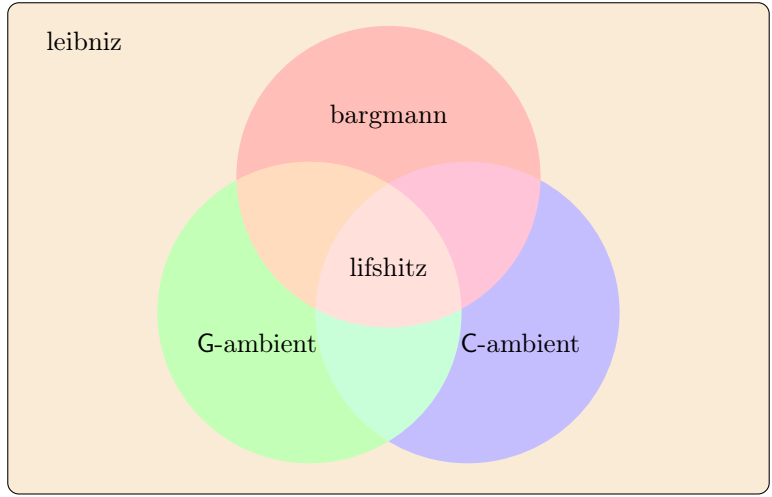

\begin{tikzpicture}
\draw[rounded corners,thick] (-5,-3) rectangle (5,3.5) ;
  \begin{scope}[blend group = soft light]
    \fill[red!40!white]   ( 90:1.2) circle (2);
    \fill[green!40!white] (210:1.2) circle (2);
    \fill[blue!40!white]  (330:1.2) circle (2);
        \fill[antiquewhite] [rounded corners,thick]   (-5,-3) rectangle (5,3.5);  
  \end{scope}
  \node at ( 90:2)  {bargmann};
   \node at ( 210:2) {$\mathsf G$-ambient};
    \node at ( 330:2)  {$\mathsf C$-ambient};
    \node {lifshitz};

    \node at ( -4,3)  {leibniz};
\end{tikzpicture}
      \captionof{figure}{Hierarchy of ambient geometries}
      \label{Figure:Hierarchy of ambient geometries}
\end{center}

\begin{Example}[Ambient aristotelian Klein pairs]
\label{Example:Ambient aristotelian Klein pairs}
Starting from the leibniz Klein pair \UnskipRef{Example:Leibniz Klein pair}, we consider the following subalgebras:
\be
\item The subalgebra spanned by $M\oplus H\oplus\tP\oplus\tK\oplus\tJ$ is isomorphic to the trivial extension of the \emph{galilei} algebra $\hyperlink{Table:K3}{\mathsf{K3}}\oplus\mathbb R$ and can be shown to induce a $\mathsf G$-ambient Klein pair, with $\bs{A}=M^*$.
\item The subalgebra spanned by $M\oplus H\oplus\tD\oplus\tP\oplus\tJ$ is isomorphic to the trivial extension of the \emph{carroll} algebra $\hyperlink{Table:K9}{\mathsf{K9}}\oplus\mathbb R$ and induces a $\mathsf C$-ambient Klein pair, with $\bs{N}=-H$.
\item The subalgebra spanned by $M\oplus H\oplus\tilde\tK\oplus\tP\oplus\tJ$, where $\tilde\tK=\tK+\tD$\footnote{\label{footnote:alternative form of leibniz algebra}Performing this change of basis, the non-trivial commutators of the leibniz algebra take the alternative form:
 \[
 \begin{tikzcd}[row sep = small,column sep = small,ampersand replacement=\&]
\br{\tD}{\tilde\tK}= C\&\br{\tilde\tK}{H}=\tP\&\br{\tD}{\tP}= M\&\br{\tilde\tK}{\tP}= M\&\br{C}{H}= M\, .
\end{tikzcd}
 \]
}, is isomorphic to the \emph{bargmann} algebra \UnskipRef{Example:Bargmann Klein pair} which induces a bargmannian Klein pair, with $\bs{g}=-M^*\otimes H^*-H^*\otimes M^*+\tP^*\otimes \tP^*$.
\ee
Finally, the subalgebra spanned by $M\oplus H\oplus\tP\oplus\tJ$ induces a lifshitzian Klein pair, isomorphic to $\mathfrak{iso}(d)\oplus \mathbb R^2$ [\cf  \hyperlink{Table:Li1}{$\mathsf{Li1}$} in Table \ref{Table:Lifshitzian Klein pairs}]. 
\end{Example}

\subsection{Lifting kinematical Klein pairs}
\label{section:Ambient approach to Klein pairs}

The present section aims at formalising the algebraic structures underlying the geometric ambient approach of Duval--Eisenhart, the main objective being to justify the notion of \emph{projectable ambient triplets} \UnskipRef{Definition:Projectable ambient triplets} to be introduced and classified 
in the forthcoming sections. We begin by collecting elementary facts regarding ideals of Klein pairs:
\bi
\item Let $(\alg,\alh)$ be a Klein pair. A \emph{sub-Klein pair} of $(\alg,\alh)$ is a Klein pair $(\ali,\alj)$ such that $\ali\subset\alg$ is a subalgebra of $\alg$ and $\alj\subset\alh$ is a subalgebra of $\alh$.
\item The \emph{quotient of a Klein pair} $(\alg,\alh)$ by an ideal $\ali\subset\alg$ is the Klein pair defined as the quotient  $(\alg/\ali,\alh/\alj)$ where $\alj$ is the ideal of $\alh$ defined as the intersection  $\alj=\ali\cap\alh$. The sequence:
\[
\begin{tikzcd}
0 \arrow[r] & (\ali,\alj)  \arrow[r, hook, "i"] & (\alg,\alh) \arrow[r, two heads,"\pi"]        & (\alg/\ali,\alh/\alj)\arrow[r]  & 0
\end{tikzcd}
\]
is thus an exact sequence of Klein pairs.
\ei

We are now in position to articulate the main notion of the present section:
\begin{Definition}[Projectable triplet]
\label{Definition:Projectable triplet}
An effective Klein pair $(\alg,\alh)$ is said to be \emph{projectable} if there exists a non-trivial ideal $\ali\subset\alg$ such that the quotient Klein pair $(\alg/\ali,\alh/\alj)$ is effective, where $\alj$ denotes the ideal of $\alh$ defined as $\alj:=\alh\cap \ali$. The triplet $(\alg,\alh,\ali)$ is hence referred to as a \emph{projectable triplet} and $(\alg,\alh)$ is said to project on $(\alg/\ali,\alh/\alj)$.
Conversely, an effective Klein pair $(\bar\alg,\bar\alh)$ for which there exists a projectable triplet $(\alg,\alh,\ali)$ such that $\bar \alg\simeq\alg/\ali$ and $\bar\alh\simeq\alh/\alj$ [with $\alj:=\alh\cap \ali$] is said to admit a \emph{lift} into $(\alg,\alh,\ali)$.

\end{Definition}
\begin{Remark}
\hfill
\bi
\item The assumption that $(\alg,\alh)$ is effective ensures that $\alh$ does not admit any non-trivial ideal, hence $\ali$ is not contained in $\alh$ [\ie $\ali \not\subset \alh$] and $\alj$ is a proper subset of $\ali$ [\ie $\alj\subset\ali$]. This ensures in particular that $\dim\alj<\dim \ali$, hence the quotient spaces $\alg/\alh$ and $(\alg/\ali)/(\alh/\alj)$ satisfy $\dim \alg/\alh> \dim (\alg/\ali)/(\alh/\alj)$, thus foreshadowing the geometric interpretation of this quotient procedure as dimensional reduction. 
\item The following diagram of vector spaces commutes:
\[
\begin{tikzcd}
            & 0 \arrow[d]                                    & 0 \arrow[d]                                    & 0 \arrow[d]                                 &   \\
0 \arrow[r] & \alj \arrow[d, hook] \arrow[r, hook]           & \ali \arrow[d, hook] \arrow[r, two heads]      & \ali/\alj \arrow[r] \arrow[d, hook]         & 0 \\
0 \arrow[r] & \alh \arrow[r, hook] \arrow[d, two heads]      & \alg \arrow[r, two heads] \arrow[d, two heads] & \alg/\alh \arrow[r] \arrow[d, two heads]    & 0 \\
0 \arrow[r] & \alh/\alj \arrow[d] \arrow[r, hook] & \alg/\ali \arrow[r, two heads] \arrow[d]       & (\alg/\ali)/(\alh/\alj) \arrow[r] \arrow[d] & 0 \\
            & 0                                              & 0                                              & 0                                           &  
\end{tikzcd}
\]
\item Whenever $(\alg,\alh,\ali)$ is a projectable triplet, the two left-most columns are sequences of Lie algebras.
\ei
\end{Remark}

The subsequent pair of examples provides two motivating instances of projectable Klein pairs, based on ambient kinematical Klein pairs \UnskipRef{Definition:Ambient kinematical Klein pairs} and ambient aristotelian Klein pairs \UnskipRef{Definition:Ambient aristotelian Klein pairs}, respectively. 

\begin{Example}[Leibniz projectable triplet]
\label{Example:Leibniz projectable triplet}
Let $(\alg,\alh)$ be the leibniz Klein pair \UnskipRef{Example:Leibniz Klein pair}.
Defining the ideal $\ali=M\oplus C\oplus \tD$, so that $\alj=C\oplus \tD$, the triplet $(\alg,\alh,\ali)$ is a projectable triplet, with quotient $(\alg/\ali,\alh/\alj)$ isomorphic to the galilei Klein pair \hyperlink{Table:E1}{$\mathsf{E1}$}.
\end{Example}

\begin{Example}[Bargmann projectable triplet]
\label{Example:Bargmann projectable triplet}
Let $(\alg,\alh)$ be the bargmann Klein pair \UnskipRef{Example:Bargmann Klein pair}.
Defining the ideal $\ali=M$, so that $\alj=\varnothing$, the triplet $(\alg,\alh,\ali)$ is a projectable triplet, with $(\alg/\ali,\alh/\alj)$ isomorphic to the galilei Klein pair \hyperlink{Table:E1}{$\mathsf{E1}$}.
\end{Example}
From the above examples, one can conclude that the galilei Klein pair \hyperlink{Table:E1}{$\mathsf{E1}$} can be lifted either to the leibniz Klein pair \UnskipRef{Example:Leibniz Klein pair} or to the bargmann Klein pair \UnskipRef{Example:Bargmann Klein pair}. Such statements can be understood as the algebraic counterparts of the fact that the flat galilei manifold [whose isometry algebra identifies with the galilei algebra \hyperlink{Table:K3}{$\mathsf{K3}$}] can be lifted either to the flat leibniz manifold \cite{Bekaert2015b,Morand2023} or to the flat bargmann manifold \cite{Duval1985} [see \eg \cite{Morand2018} for a review], whose isometry algebras identify with the leibniz \eqref{equation:leibniz algebra} and bargmann \eqref{equation:bargmann algebra} algebras, respectively. [These two constructions are in fact related since the leibnizian manifold underlying\footnote{Recall that any bargmannian manifold $(\M,\xi,g,\nabla)$ induces a leibnizian manifold $(\M,\xi,\psi,\gamma,\nabla)$  upon the identification \cite{Bekaert2015b}:
\[
 \begin{tikzcd}[row sep = small,column sep = small,ampersand replacement=\&]
{\psi}:=  g(\xi,-)\&\gamma:=  g|_{\Ker{\psi}}
\end{tikzcd}
\]
where the second definition reads more explicitly $\gamma(V,W):= g(V,W)$ for all $V,W\in\Milnegh$.} the flat bargmann manifold is precisely the flat leibniz manifold, see the companion paper  \cite{Morand2023} for details.] 

The two previous examples offer as many alternative paths for extending the ambient approach: either following the ambient aristotelian path, with the promise to meet bargmann along the way, or the ambient kinematical road, with leibniz as a prominent signpost. At this point in our journey, we shall elect the path marked with a single letter $M$ and pursue our classificatory ambition in the ambient aristotelian direction, before coming back to the ambient kinematical path (actually via a shortcut) in Section \ref{section:Drink me}.
\section{Bargmann and his modern rivals: classifying ambient kinematics}
\label{section:Bargmann and his modern rivals}
As mentioned in the introduction, generalisations of the bargmannian projectable triplet \UnskipRef{Example:Bargmann projectable triplet} have recently been exhibited and studied in \cite{FigueroaOFarrill2022,FigueroaOFarrill2022e} where it was in particular shown that these bargmannian \UnskipRef{Definition:Ambient aristotelian metric structures on Klein pairs} Klein pairs allowed to lift all galilean \UnskipRef{Definition:Metric structures on Klein pairs} kinematical Klein pairs [see Table \ref{Table:Effective kinematical Klein pairs}].

In broad strokes, such generalisations can be characterised as lying at the intersection of projectable triplets \UnskipRef{Definition:Projectable triplet} and ambient aristotelian Klein pairs \UnskipRef{Definition:Ambient aristotelian Klein pairs}. We will call such intersection \emph{projectable ambient triplets}, which we claim constitutes the natural algebraic abstraction underlying the usual lift of nonrelativistic structures into codimension one higher `ambient' geometries. The aim of the present section (and main objective of the present paper) consists precisely in classifying such projectable ambient triplets, the results being collected in Table \ref{Table:Projectable ambient triplets}. 

\begin{Definition}[Projectable ambient triplets]
\label{Definition:Projectable ambient triplets}
A projectable ambient triplet $(\alg,\alh,\ali)$ is composed of the following elements:
\be
\item $(\alg,\alh)$ is an effective ambient aristotelian Klein pair \UnskipRef{Definition:Ambient aristotelian Klein pairs}.
\item $\ali\subset\alg$ is a scalar ideal\footnote{In other words, $\ali\subset\alg$ is a non-trivial ideal of type $(1,0,0)$ \UnskipRef{Definition:Subspace of type}. In the following, we will denote $\ali=\Span M$.} in $\alg$ such that the quotient Klein pair $(\alg/\ali,\alh)$ is effective.
\ee
\end{Definition}

A morphism of projectable ambient triplets $\varphi:(\alg_1,\alh_1,\ali_1)\to(\alg_2,\alh_2,\ali_2)$ is a homomorphism of Lie algebras $\varphi\From\alg_1\to\alg_2$ preserving the corresponding $\so(d)$-decomposition and such that $\varphi(\alh_1)\subseteq\alh_2$ and $\varphi(\ali_1)\subseteq\ali_2$. 
\medskip

The motivating example of projectable ambient triplets is given by the bargmannian Klein pair \UnskipRef{Example:Bargmann projectable triplet}, as well as its generalisations described in  \cite{FigueroaOFarrill2022,FigueroaOFarrill2022e}.
Our main result consists in integrating these examples in a classification of universal projectable ambient triplets. 
\begin{Theorem}[Possible ambient kinematics]
\label{Theorem:Classification of universal projectable ambient triplets}
Universal projectable ambient triplets are classified in {\normalfont Table \ref{Table:Projectable ambient triplets}}. 
\end{Theorem}
The proof of Theorem \ref{Theorem:Classification of universal projectable ambient triplets} will proceed in three steps:
\be
\item We start by classifying the pairs $(\alg,\ali)$ where $\alg$ is an ambient aristotelian algebra \UnskipRef{Definition:Ambient aristotelian algebra} and $\ali\subset\alg$ a scalar ideal [\cf Table \ref{Table:Ambient aristotelian algebras admitting a scalar ideal}].
\item For each of the previously classified algebras $\alg$, we 
 classify the corresponding effective Klein pairs $(\alg,\alh)$ [\cf Table \ref{Table:Effective ambient aristotelian Klein pairs with scalar ideal}].
\item For each of the previously classified effective Klein pairs $(\alg,\alh)$, we classify the associated scalar ideals $\ali$ such that $(\alg/\ali,\alh)$ is effective. The output of the procedure is a classification of all projectable ambient triplets [\cf Table \ref{Table:Projectable ambient triplets}].
\ee
\subsection{Ambient aristotelian algebras}
\label{section:Ambient aristotelian algebras admitting a scalar ideal}

We start by classifying the ambient aristotelian algebras 
\bea
\alg=M\oplus H\oplus \tP\oplus\tK\oplus\tJ\nn
\eea
admitting a scalar ideal $\ali=\Span M$.
Adapted bases for $\alg$ take the form $(M, H, \tK,\tP,\tJ)$. The set of adapted bases is a principal homogeneous space \SectionRef{section:A Tangled Tale: (s,v)-Lie algebras} for $\GLR{2}\times\GLR{2}\subset\GLR{\dim \alg}$ where $\dim\alg=\frac{(d+1)(d+2)}{2}+1$. The subgroup preserving $\ali$ is given by $\TR{2}\times\GLR{2}$, where $\TR{2}=(\GroupR)^{\times 2}\ltimes\mathbb R$ stands for the subgroup of invertible lower triangular 2-dimensional matrices, acting on generators as:
\bea
\label{equation:subgroup}
\begin{pmatrix}
M\\
H
\end{pmatrix}
\mapsto
\underbrace{\begin{pmatrix}
\mu&0\\
\chi&\lambda
\end{pmatrix}}_{\mathsf{R}}
\cdot
\begin{pmatrix}
M\\
H
\end{pmatrix}
\text{ and }
\begin{pmatrix}
\tK\\
\tP
\end{pmatrix}
\mapsto
\underbrace{\begin{pmatrix}
\alpha&\beta\\
\gamma&\delta
\end{pmatrix}}_{\mathsf{G}}
\cdot \begin{pmatrix}
\tK\\
\tP
\end{pmatrix}
\text{ where }\mathsf{G}\in\GLR{2}\text{ and }\mathsf{R}\in\TR{2}\, .
\eea

The possible non-trivial commutators compatible with the $\so(d)$ structure of $\alg$ read as:\footnote{Explicitly, the condition that $\ali$ is an ideal imposes that all commutators of the form $\br{\alg}{M}\not\subset M$ vanish.}
\begin{equation}
\begin{tikzcd}[row sep = small,column sep = small]
\br{H}{\tK}= \alt_1 \tK+\alt_2 \tP&\br{H}{\tP}=\alt_3 \tK+\alt_4 \tP&\br{H}{M}= \alt_5 M\\
\br{\tK}{\tP}= \alt_6 H+ \alt_7 M+\alt_8 \tJ&\br{\tK}{\tK}= \alt_9 \tJ&\br{\tP}{\tP}=\alt_{10} \tJ\, .\label{equation:commutators row}
\end{tikzcd}
\end{equation}
We need to determine the space of solutions of the corresponding Jacobi equation and quotient by the group $\TR{2}\times\GLR{2}$ \eqref{equation:subgroup}.
For each orbit, we will explicit a particular representative [with respect to an adapted basis $(M, H, \tK,\tP,\tJ)$] to be collected in Table \ref{Table:Ambient aristotelian algebras admitting a scalar ideal}. The non-trivial values of the corresponding Jacobiator $\mathsf{J}(x,y,z):=\br{x}{\br{y}{z}}+\br{y}{\br{z}{x}}+\br{z}{\br{x}{y}}$ take the following form:
\bi
\item $\mathsf{J}(M,P_i,K_j)=\delta_{ij}M\alt_5\alt_6$
\item $\mathsf{J}(H,K_i,K_j)=-2J_{ij}(\alt_2\alt_8+\alt_1\alt_9)$
\item $\mathsf{J}(H,P_i,P_j)=-2J_{ij}(\alt_4\alt_{10}+\alt_3\alt_8)$
\item $\mathsf{J}(P_i,P_j,K_k)=(\delta_{jk}P_i-\delta_{ik}P_j)(\alt_8+\alt_4\alt_6)-(\delta_{jk}K_i-\delta_{ik}K_j)(\alt_{10}-\alt_3\alt_6)$
\item $\mathsf{J}(P_i,K_j,K_k)=(\delta_{ij}P_k-\delta_{ik}P_j)(\alt_9+\alt_2\alt_6)-(\delta_{ij}K_k-\delta_{ik}K_j)(\alt_8-\alt_1\alt_6)$
\item $\mathsf{J}(H,P_i,K_j)=\delta_{ij}M\alt_7(\alt_1+\alt_4-\alt_5)+\delta_{ij}H\alt_6(\alt_1+\alt_4)-J_{ij}(\alt_8(\alt_1+\alt_4)+\alt_2\alt_{10}+\alt_3\alt_9)$.
\ei
The corresponding system gets simplified to:
\[
\begin{tikzcd}[row sep = small,column sep = small]
&\alt_8=-\alt_4\alt_6&\alt_9=-\alt_2\alt_6&\alt_{10}=\alt_3\alt_6\\
&\alt_6(\alt_1+\alt_4)=0&\alt_5\alt_6=0
&\alt_7(\alt_1+\alt_4-\alt_5)=0\, .\nn
\end{tikzcd}
\]
The first row only contains equations involving linear terms, and hence can be interpreted as providing a definition for $\alt_8$, $\alt_9$ and $\alt_{10}$, definition to be substituted back in \eqref{equation:commutators row}. The second row involves only quadratic equations, whose solutions can be partitioned
into the three following branches:
\be
\item $\alt_6=\alt_7=0$
\item $\alt_6\neq0$, $\alt_5=0$, $\alt_1+\alt_4=0$
\item $\alt_6=0$, $\alt_7\neq0$, $\alt_1+\alt_4-\alt_5=0$.
\ee
Let us address those three branches in turn and classify the corresponding ambient aristotelian algebras $\alg$. 

\paragraph{First branch ($\alt_6=\alt_7=0$)}
\hfill

\medskip
Imposing $\alt_6=\alt_7=0$ in \eqref{equation:commutators row}, the remaining non-trivial commutators can be arranged as:
\[
\begin{tikzcd}[row sep = small,column sep = small]
\br{H}{\tK}= \alt_1 \tK+\alt_2 \tP&\br{H}{\tP}=\alt_3 \tK+\alt_4 \tP&\br{H}{M}= \alt_5 M\, .
\end{tikzcd}
\]
Performing a change of adapted basis \eqref{equation:subgroup} induces a representation $\TR{2}\times\GLR{2}\to\GLR{5}$ on the space of parameters $\mathbb R^5=\Spannn{\alt_1,\alt_2,\alt_3,\alt_4,\alt_5}$, whose kernel is given by elements of the form: 
\[
\begin{pmatrix}
\mu&0\\
\chi&1
\end{pmatrix}
\times
\begin{pmatrix}
\kappa&0\\
0&\kappa
\end{pmatrix}
\text{ where }
\mu,\kappa\in\GroupR
\text{ and }
\chi\in\mathbb R\, .
\]
Quotienting by the latter yields a faithful representation $\GroupR\times\PGLR{2}\to\GLR{5}$, where $\PGLR{2}$ stands for the \emph{projective general linear group} $\PGLR{2}=\GLR{2}/Z(\mathbb R^2)$, with $Z(\mathbb R^2)=\pset{\kappa\, \Id_2\, \big|\, \kappa\in\GroupR}$ the subgroup of nonzero scalar transformations of $\mathbb R^2$. Explicitly, denoting: 
\bea
\label{equation:matrix M}
\mathsf{M}=\begin{pmatrix}
\alt_1&\alt_2\\
\alt_3&\alt_4
\end{pmatrix}
\eea the induced representation acts on $\mathbb R^5$ as $\mathsf{M}\mapsto \lambda\, \mathsf{G}\, \mathsf{M}\, \mathsf{G}\un$ and $\alt_5\mapsto\lambda\, \alt_5$, with $\lambda\in\GroupR$ and $\mathsf{G}\in\PGLR{2}$.
Classifying algebras in the first branch thus boils down to classify similarity classes of 2-dimensional matrices $\mathsf{M}$, up to the rescaling induced by $\lambda$.
Recall that two matrices are similar [\ie belong to the same similarity class] if and only if they have same \emph{Frobenius normal form}. 
The Frobenius normal form $\mathsf{F}$ for 2-dimensional matrices is given by:
\bea
\mathsf{F}=\begin{cases}
\begin{pmatrix}
\alt_1&0\\
0&\alt_1
\end{pmatrix} \text{ if }\alt_2=\alt_3=0\text{ and } \alt_4=\alt_1\\
\nn\\
\begin{pmatrix}
0&-\text{det}\\
1&\text{tr}
\end{pmatrix} \text{ otherwise}
\end{cases}
\eea
where $\text{tr}=\alt_1+\alt_4$ and $\text{det}=\alt_1\alt_4-\alt_2\alt_3$. We start by considering the case for which $\mathsf{M}$ is proportional to the identity \ie $\mathsf{M}=\mathsf{F}=\alt_1 \Id_2$.
There are two subcases to consider corresponding to $\alt_1=0$ and $\alt_1\neq0$:
\bi
\item $\alt_1=0$
\hfill

The only non-trivial commutator is given by $\br{H}{M}= \alt_5 M$. There are then two subcases to examine:
\be
\item $\alt_5=0$: This case corresponds to the trivial extension of the static algebra $\hyperlink{Table:K1}{\mathsf{K1}}\oplus\mathbb R$  [\cf Table \ref{Table:Universal kinematical algebras}] with no additional commutators on top of the canonical ones \UnskipRef{Definition:Ambient aristotelian algebra}:
\mybox{WarmGray}{
$\hypertarget{Eq:AK1}{\hyperlink{Table:AK1}{\mathsf{AK1}}}$
}
{
\begin{tikzcd}[row sep = small,column sep = small,ampersand replacement=\&]
\varnothing
\end{tikzcd}
}
\item $\alt_5\neq0$: Performing a rescaling via $\lambda=\frac{1}{\alt_5}$ allows to set the only non-trivial commutator to the form: 
\mybox{WarmGray}{
$\hypertarget{Eq:AK1a}{\hyperlink{Table:AK1a}{\mathsf{AK1a}}}$
}
{
\begin{tikzcd}[row sep = small,column sep = small,ampersand replacement=\&]
\br{H}{M}=M
\end{tikzcd}
}

This corresponds to the semi-direct sum $\hyperlink{Table:K1}{\mathsf{K1}}\inplus\mathbb R$ where $H$ acts on $M$ by dilation.
\ee
\item $\alt_1\neq0$
\hfill

Rescaling via $\lambda=\frac{1}{\alt_1}$ allows to put $\mathsf{F}$ in the form $\mathsf{F}=\Id_2$. Since the rescaling freedom has already been used, the coefficient in the commutator $\br{H}{M}\sim M$ is left unfixed, yielding the semi-direct sum $\hyperlink{Table:K2}{\mathsf{K2}}\inplus_\kappa\mathbb R$:
\bi
\item[]
\mybox{WarmGray}{
$\hypertarget{Eq:AK2}{\hyperlink{Table:AK2}{\mathsf{AK2}_\kappa}}$
}
{
\begin{tikzcd}[row sep = small,column sep = small,ampersand replacement=\&]
\br{H}{\tK}= \tK\&\br{H}{\tP}=\tP\&\br{H}{M}=\kappa M\&\text{ where }\kappa\in\mathbb R\text{ is arbitrary}
\end{tikzcd}
}
\ei
The case $\kappa=0$ corresponds to the direct sum $\hyperlink{Table:K2}{\mathsf{K2}}\oplus\mathbb R$.
\ei

Whenever the matrix $\mathsf{M}$ is \textit{not} proportional to the identity, similarity classes are classified by pairs $(\tr,\det)$ modulo the equivalence relation $(\tr,\det)\sim(\lambda\, \tr,\lambda^2\det)$. 

We first consider the case $\det=0$, which contains two equivalence classes represented by: 
\bi
\item $(0,0)$ if $\tr=0$ so that $\mathsf{F}=\begin{pmatrix}
0&0\\
1&0
\end{pmatrix}$. There are two further subcases to distinguish:
\be
\item $\alt_5=0$: This corresponds to the trivial extension of the galilei algebra $\hyperlink{Table:K3}{\mathsf{K3}}\oplus\mathbb R$:
\mybox{WarmGray}{
$\hypertarget{Eq:AK3}{\hyperlink{Table:AK3}{\mathsf{AK3}}}$
}
{
\begin{tikzcd}[row sep = small,column sep = small,ampersand replacement=\&]
\br{H}{\tP}=\tK
\end{tikzcd}
}

\item $\alt_5\neq0$: Rescaling via $\lambda=\frac{1}{\alt_5}$ allows to set the non-trivial commutators to the form: 
\mybox{WarmGray}{
$\hypertarget{Eq:AK3a}{\hyperlink{Table:AK3a}{\mathsf{AK3a}}}$
}
{
\begin{tikzcd}[row sep = small,column sep = small,ampersand replacement=\&]
\br{H}{\tP}=\tK\&\br{H}{M}= M
\end{tikzcd}
}

This corresponds to the semi-direct sum $\hyperlink{Table:K3}{\mathsf{K3}}\inplus\mathbb R$ where $H$ acts on $M$ by dilation.
\ee
\item $(1,0)$ if $\tr\neq0$, hence rescaling via $\lambda=\frac1\tr$ yields $\mathsf{F}=\begin{pmatrix}
0&0\\
1&1
\end{pmatrix}$ so that:
\[
\begin{tikzcd}[row sep = small,column sep = small,ampersand replacement=\&]
\br{H}{\tP}=\tP+\tK\&\br{H}{M}=\kappa M\&\text{ where }\kappa\in\mathbb R\text{ is arbitrary.}
\end{tikzcd}
\]
Performing a redefinition $\tP\mapsto\tK+\tP$ yields the semi-direct product $\hyperlink{Table:K4}{\mathsf{K4}}\inplus_\kappa\mathbb R$:
\bi
\item[]
\mybox{WarmGray}{
$\hypertarget{Eq:AK4}{\hyperlink{Table:AK4}{\mathsf{AK4}_\kappa}}$
}
{
\begin{tikzcd}[row sep = small,column sep = small,ampersand replacement=\&]
\br{H}{\tP}=\tP\&\br{H}{M}=\kappa M\&\text{ where }\kappa\in\mathbb R\text{ is arbitrary}
\end{tikzcd}
}
\ei
\ei
Finally, we consider the case $\det\neq0$. We distinguish between two cases:
\bi
\item $\tr=0$, so that rescaling via $\lambda=\frac{\epsilon}{\sqrt{|\det|}}$, with $\epsilon=\pm1$, yields two classes:
\be
\item $(0,1)$ if $\det>0$  so that $\mathsf{F}=\begin{pmatrix}
0&-1\\
1&0
\end{pmatrix}$. 

The remaining arbitrariness in the sign $\epsilon$ can be used to set the coefficient in the commutator $\br{H}{M}\sim M$ to be non-negative. The resulting algebra corresponds to the semi-direct product $\hyperlink{Table:K5}{\mathsf{K5}}\inplus_\kappa\mathbb R$ where $\hyperlink{Table:K5}{\mathsf{K5}}$ stands for the euclidean newton algebra:
\mybox{WarmGray}{
$\hypertarget{Eq:AK5}{\hyperlink{Table:AK5}{\mathsf{AK5}_{\kappa}}}$
}
{
\begin{tikzcd}[row sep = small,column sep = small,ampersand replacement=\&]
\br{H}{\tK}=-\tP\&\br{H}{\tP}=\tK\&\br{H}{M}=\kappa M\&\text{ where }\kappa\geq0
\end{tikzcd}
}
\item $(0,-1)$ if $\det<0$  so that $\mathsf{F}=\begin{pmatrix}
0&1\\
1&0
\end{pmatrix} $ corresponding to the semi-direct product $\hyperlink{Table:K6}{\mathsf{K6}}\inplus_\kappa\mathbb R$ where $\hyperlink{Table:K6}{\mathsf{K6}}$ stands for the lorentzian newton algebra:
\mybox{WarmGray}{
$\hypertarget{Eq:AK6}{\hyperlink{Table:AK6}{\mathsf{AK6}_{\kappa}}}$
}
{
\begin{tikzcd}[row sep = small,column sep = small,ampersand replacement=\&]
\br{H}{\tK}=\tP\&\br{H}{\tP}=\tK\&\br{H}{M}=\kappa M\&\text{ where }\kappa\geq0
\end{tikzcd}
}
\ee

\item $\tr\neq0$, so that rescaling via $\lambda=\frac{\tr}{|\tr|\sqrt{|\det|}}$ yields two classes:
\be
\item $(\alpha_+,1)$ if $\det>0$, where $\alpha_+>0$ is defined as $\alpha_+=\frac{|\tr|}{\sqrt{\det}}$, so that $\mathsf{F}=\begin{pmatrix}
0&-1\\
1&\alpha_+
\end{pmatrix} $ corresponding to the semi-direct product $\hyperlink{Table:K7}{\mathsf{K7}_{\alpha_+}}\inplus_\kappa\mathbb R$:
\mybox{WarmGray}{
$\hypertarget{Eq:AK7}{\hyperlink{Table:AK7}{\mathsf{AK7}_{\alpha_+,\kappa}}}$
}
{
\begin{tikzcd}[row sep = small,column sep = small,ampersand replacement=\&]
\br{H}{\tK}=-\tP\&\br{H}{\tP}=\tK+\alpha_+\tP\&\br{H}{M}=\kappa M\&\alpha_+>0
,\ 
\kappa\in\mathbb R
\end{tikzcd}
}
Contrarily to the case where $\tr=0$, the sign of the scaling factor $\lambda$ is fixed to cancel the sign of the trace so that $\alpha_+>0$ hence the sign of $\kappa$ is left unfixed. The Lie algebra $\hyperlink{Table:AK7}{\mathsf{AK7}_{\alpha_+,\kappa}}$ hence splits into 3 different cases, $\hyperlink{Table:AK7}{\mathsf{AK7}_{\alpha_+}}$, $\hyperlink{Table:AK7+}{\mathsf{AK7}^+_{\alpha_+,\kappa}}$ and $\hyperlink{Table:AK7-}{\mathsf{AK7}^-_{\alpha_+,\kappa}}$ according to $\kappa=0$, $\kappa>0$ and $\kappa<0$, respectively.
\item $(\alpha_-,-1)$ if $\det<0$, where $\alpha_->0$ is defined as $\alpha_-=\frac{|\tr|}{\sqrt{-\det}}$, so that $\mathsf{F}=\begin{pmatrix}
0&1\\
1&\alpha_-
\end{pmatrix} $ corresponding to the semi-direct product $\hyperlink{Table:K8}{\mathsf{K8}_{\alpha_-}}\inplus_\kappa\mathbb R$:
\mybox{WarmGray}{
$\hypertarget{Eq:AK8}{\hyperlink{Table:AK8}{\mathsf{AK8}_{\alpha_-,\kappa}}}$
}
{
\begin{tikzcd}[row sep = small,column sep = small,ampersand replacement=\&]
\br{H}{\tK}=\tP\&\br{H}{\tP}=\tK+\alpha_-\tP\&\br{H}{M}=\kappa M\&\alpha_->0
,\ 
\kappa\in\mathbb R
\end{tikzcd}
}
The Lie algebra $\hyperlink{Table:AK8}{\mathsf{AK8}_{\alpha_-,\kappa}}$ splits into 3 different cases, $\hyperlink{Table:AK8}{\mathsf{AK8}_{\alpha_-}}$, $\hyperlink{Table:AK8+}{\mathsf{AK8}^+_{\alpha_-,\kappa}}$ and $\hyperlink{Table:AK8-}{\mathsf{AK8}^-_{\alpha_-,\kappa}}$ according to $\kappa=0$, $\kappa>0$ and $\kappa<0$, respectively.
\ee
\ei

This concludes the classification of Lie algebras living in the first branch.

\paragraph{Second branch ($\alt_6\neq0$, $\alt_5=0$, $\alt_1+\alt_4=0$)}
\hfill
\medskip

Since $\alt_6\neq0$ by assumption, one can always perform the invertible transformation $H\mapsto \alt_6H+\alt_7M$ yielding:
\bea
\label{equation:second branch}
\hspace{-1.5cm}\begin{tikzcd}[row sep = small,column sep = small]
\br{H}{\tK}=\alt'_1 \tK+\alt'_2 \tP&\br{H}{\tP}=\alt'_3 \tK-\alt'_1\tP&\br{\tK}{\tP}=  H+\alt'_1 \tJ&\br{\tK}{\tK}= -\alt'_2 \tJ&\br{\tP}{\tP}=\alt'_3 \tJ\hspace{-1cm}
\end{tikzcd}
\eea
where we relabelled $\alt'_1=\alt_1\alt_6$, $\alt'_2=\alt_2\alt_6$ and $\alt'_3=\alt_3\alt_6$. Note that the generator $M$ disappears from the commutation relations. The investigation of the second branch in the present classification hence boils down to the investigation of the second branch of \cite[Section 2.4]{Figueroa-OFarrill2017a} in the classification of kinematical algebras that we now briefly review.

Performing a change of adapted basis \eqref{equation:subgroup} preserves the form of the commutation relations \eqref{equation:second branch} provided that $\chi=0$ and $\lambda=\det \mathsf{G}=\alpha\delta-\beta\gamma$. This induces a representation of the corresponding subgroup on the space of parameters $\mathbb R^3=\Spannn{\alt'_1,\alt'_2,\alt'_3}$ whose kernel is given by elements of the form: 
\[
\begin{pmatrix}
\mu&0\\
0&1
\end{pmatrix}
\times
\begin{pmatrix}
\pm1&0\\
0&\pm1
\end{pmatrix}
\text{ where }
\mu\in\GroupR\, .
\]

Quotienting by the latter yields a faithful representation\footnote{Remember that $\GLR{2}\simeq\GroupR\ltimes\SLR{2}$ so that each element of 
$\GLR{2}$
can be uniquely written as a product:
\[
\underbrace{\begin{pmatrix}
\alpha&\beta\\
\gamma&\delta
\end{pmatrix}}_{\in\GLR{2}}=\underbrace{\begin{pmatrix}
\Delta&0\\
0&1
\end{pmatrix}}_{\in\GroupR}\underbrace{\begin{pmatrix}
\frac{\alpha}{\Delta}&\frac{\beta}{\Delta}\\
\gamma&\delta
\end{pmatrix}}_{\in\SLR{2}}
\text{ where }\Delta=\alpha\delta-\beta\gamma\, .
 \]
 This decomposition induces an action: $\varphi\From\GroupR\times\SLR{2}\to\SLR{2}\From\kappa\times\begin{pmatrix}
a&b\\
c&d
\end{pmatrix}\mapsto
\begin{pmatrix}
a&b/\kappa\\
\kappa c&d
\end{pmatrix}$ so that the group operation on $\GroupR\ltimes\SLR{2}$ is defined as: $(\kappa_1,\mathsf{S}_1)\times (\kappa_2,\mathsf{S}_2)\to\big(\kappa_1\kappa_2,\varphi_{\kappa_2}(\mathsf{S}_1)\, \mathsf{S}_2\big)$.
} $\GroupR\ltimes\PSLR{2}\to\GLR{3}$---where $\PSLR{2}$ stands for the \emph{projective special linear group} $\PSLR{2}=\SLR{2}/\mathbb Z_2$ with $\mathbb Z_2=\pset{\pm\Id_2}$ the subgroup of scalar transformations with unit determinant---acting on $\mathbb R^3$ as:
\bea
\label{equation:matrix representation}
\begin{pmatrix}
\alt'_1\\
\alt'_2\\
\alt'_3
\end{pmatrix}
\mapsto
\underbrace{\begin{pmatrix}
 \kappa &0&0\\
0& \kappa ^2&0\\
0&0&1
\end{pmatrix}}_{\hat{\mathsf R}}
\underbrace{\begin{pmatrix}
 a  d + b  c &- a  c & b  d \\
-2 a  b & a ^2&- b ^2\\
2 c  d &- c ^2& d ^2
\end{pmatrix}
}_{\hat{\mathsf{S}}}
\begin{pmatrix}
\alt'_1\\
\alt'_2\\
\alt'_3
\end{pmatrix}
\text{ where }
 \kappa \in\mathbb R^\times
\text{ and }
\mathsf{S}=\begin{pmatrix}
 a & b \\
 c & d 
\end{pmatrix}\in\SLR{2}\, .\nn
\eea
Indeed, the matrix $\hat{\mathsf{S}}$ provides a representation $\SLR{2}\to\GLR{3}$ with kernel $\mathbb Z_2$, hence, it induces a faithful representation of the projective special linear group $\PSLR{2}\to \GL(3|\mathbb R)$. The latter is isomorphic to the \emph{proper orthochronous Lorentz group} $\SO^+(1,2)$, defined as the subgroup of linear transformations on $(2+1)$-dimensional minkowski space preserving both time direction and space orientation. Indeed, defining the matrix 
\[
\eta=\begin{pmatrix}
2&0&0\\
0&0&1\\
0&1&0
\end{pmatrix}
\]
we have $\hat{\mathsf{S}}^T \eta\, \hat{\mathsf{S}}=\eta$ for all ${\mathsf{S}}\in\SLR{2}$. Acting with $\hat{\mathsf{S}}$ on a vector  $v=\begin{pmatrix}
\alt'_1&
\alt'_2&
\alt'_3
\end{pmatrix}$ thus preserves the norm $Q(v)=v\, \eta\,  v^T=2\, (\alt'_1{}^2+\alt'_2\alt'_3)$. This fact, combined with $\det\hat{\mathsf{S}}=1$, ensures that $\hat{\mathsf{S}}\in\SO(1,2)$ [thanks to the fact that $\eta$ has signature $(1,2)$]. It can finally be checked\footnote{This fact can be made manifest by diagonalising $\eta$ via the transformation matrix 
$\mathsf{Q}:=\begin{pmatrix}
0&0&1\\
\frac{\sqrt{2}}{2}&\frac{\sqrt{2}}{2}&0\\
-\frac{\sqrt{2}}{2}&\frac{\sqrt{2}}{2}&0\\
\end{pmatrix}\in\SO(3)$
, so that 
\[
\eta':=\mathsf{Q}^T\eta\, \mathsf{Q}=\begin{pmatrix}
-1&0&0\\
0&1&0\\
0&0&2
\end{pmatrix}
\text{ and }
\hat{\mathsf{S}}':=\mathsf{Q}^T\hat{\mathsf{S}}\, \mathsf{Q}=\begin{pmatrix}
\frac{1}{2}\big( a ^2+ b ^2+ c ^2+ d ^2\big)&\frac{1}{2}\big( a ^2- b ^2+ c ^2- d ^2\big)&-\sqrt{2}\, ( a  b + c  d )\\
\frac{1}{2}\big( a ^2+ b ^2- c ^2- d ^2\big)&\frac{1}{2}\big( a ^2- b ^2- c ^2+ d ^2\big)&-\sqrt{2}\, ( a  b - c  d )\\
-\frac{\sqrt{2}}{2}( a  c + b  d )&-\frac{\sqrt{2}}{2}( a  c - b  d )& a  d + b  c 
\end{pmatrix}
\]
where $\hat{\mathsf{S}}'{}^0{}_{0}=\frac{1}{2}\big( a ^2+ b ^2+ c ^2+ d ^2\big)>0$.
} that acting with $\hat{\mathsf{S}}$ preserves the sign of the time component $T(v)=\frac{\sqrt{2}}{2}(\alt'_2-\alt'_3)$, hence $\hat{\mathsf{S}}\in\SO^+(1,2)$ as expected. Finally, let us note that acting with $\kappa\in\mathbb R^\times$ via $\hat{\mathsf{R}}$ rescales the norm as $Q(\hat{\mathsf{R}}\, v)=\kappa^2 Q(v)$. Our classifying task for this branch then boils down to classify the orbits of the proper orthochronous Lorentz group $\SO^+(1,2)$ on $\mathbb R^3$, modulo the rescaling induced by $\kappa\in\mathbb R^\times$. There are 6 orbits [or surfaces of transitivity] to be considered \cite{Hermann1966}. For each of them, we write a representative vector $v\in\mathbb R^3$ and explicit the corresponding commutation relations:
\bi
\item[]
\bi
\item[$\bullet$] \emph{origin} $v=\begin{pmatrix}0&0&0\end{pmatrix}$

\mybox{WarmGray}{
$\hypertarget{Eq:AK9}{\hyperlink{Table:AK9}{\mathsf{AK9}}}$ \quad $\text{carroll}\oplus\mathbb R$
}
{
\begin{tikzcd}[row sep = small,column sep = small,ampersand replacement=\&]
\br{\tK}{\tP}= H
\end{tikzcd}
}

\item[$\bullet$] \emph{future directed lightlike orbit}  $v=\begin{pmatrix}0&0&-1\end{pmatrix}$, $Q(v)=0$, $T(v)>0$

\mybox{WarmGray}{
$\hypertarget{Eq:AK10}{\hyperlink{Table:AK10}{\mathsf{AK10}}}$\quad$\mathfrak{iso}(d+1)\oplus\mathbb R$
}
{
\begin{tikzcd}[row sep = small,column sep = small,ampersand replacement=\&]
\br{H}{\tP}=-\tK\&\br{\tK}{\tP}=H\&\br{\tP}{\tP}=-\tJ
\end{tikzcd}
}

\item[$\bullet$] \emph{past directed lightlike orbit}  $v=\begin{pmatrix}0&0&1\end{pmatrix}$, $Q(v)=0$, $T(v)<0$

\mybox{WarmGray}{
$\hypertarget{Eq:AK11}{\hyperlink{Table:AK11}{\mathsf{AK11}}}$\quad$\mathfrak{iso}(d,1)\oplus\mathbb R$
}
{
\begin{tikzcd}[row sep = small,column sep = small,ampersand replacement=\&]
\br{H}{\tP}=\tK\&\br{\tK}{\tP}=H\&\br{\tP}{\tP}=\tJ
\end{tikzcd}
}

\item[$\bullet$] \emph{spacelike orbit}  $v=\begin{pmatrix}0&1&1\end{pmatrix}$, $Q(v)>0$

\mybox{WarmGray}{
$\hypertarget{Eq:AK12}{\hyperlink{Table:AK12}{\mathsf{AK12}}}$\quad$\so(d+1,1)\oplus\mathbb R$
}
{
\begin{tikzcd}[row sep = small,column sep = small,ampersand replacement=\&]
\br{H}{\tK}=\tP\&\br{H}{\tP}=\tK\&\br{\tK}{\tP}= H\&\br{\tK}{\tK}= -\tJ\&\br{\tP}{\tP}=\tJ
\end{tikzcd}
}

\item[$\bullet$] \emph{future directed timelike orbit}  $v=\begin{pmatrix}0&1&-1\end{pmatrix}$, $Q(v)<0$, $T(v)>0$

\mybox{WarmGray}{
$\hypertarget{Eq:AK13}{\hyperlink{Table:AK13}{\mathsf{AK13}}}$\quad$\so(d+2)\oplus\mathbb R$
}
{
\begin{tikzcd}[row sep = small,column sep = small,ampersand replacement=\&]
\br{H}{\tK}=\tP\&\br{H}{\tP}=-\tK\&\br{\tK}{\tP}=H\&\br{\tK}{\tK}=-\tJ\&\br{\tP}{\tP}=- \tJ
\end{tikzcd}
}

\item[$\bullet$] \emph{past directed timelike orbit}  $v=\begin{pmatrix}0&-1&1\end{pmatrix}$, $Q(v)<0$, $T(v)<0$

\mybox{WarmGray}{
$\hypertarget{Eq:AK14}{\hyperlink{Table:AK14}{\mathsf{AK14}}}$\quad$\so(d,2)\oplus\mathbb R$
}
{
\begin{tikzcd}[row sep = small,column sep = small,ampersand replacement=\&]
\br{H}{\tK}=-\tP\&\br{H}{\tP}=\tK\&\br{\tK}{\tP}=H\&\br{\tK}{\tK}=\tJ\&\br{\tP}{\tP}=\tJ
\end{tikzcd}
}

\ei
\ei
This concludes the classification of the second branch.

\paragraph{Third branch ($\alt_6=0$, $\alt_7\neq0$, $\alt_1+\alt_4-\alt_5=0$)}
\hfill

\medskip
Imposing $\alt_6=\alt_1+\alt_4-\alt_5=0$ in \eqref{equation:commutators row}, the remaining non-trivial commutators can be arranged as:
\bea
\label{equation:third branch}
\hspace{-1cm}\begin{tikzcd}[row sep = small,column sep = small]
\br{\tK}{\tP}=M&\br{H}{\tK}= \alt_1 \tK+\alt_2 \tP&\br{H}{\tP}=\alt_3 \tK+\alt_4 \tP&\br{H}{M}= (\alt_1+\alt_4) M
\end{tikzcd}
\eea
where we used the assumption that $\alt_7\neq0$ to rescale $M$ as $M\mapsto\frac{1}{\alt_7}M$. 

Performing a change of adapted basis \eqref{equation:subgroup} preserves the form of the commutation relations \eqref{equation:third branch} provided that $\mu=\det \mathsf{G}=\alpha\delta-\beta\gamma$. This induces a representation of the corresponding subgroup on the space of parameters $\mathbb R^4=\Spannn{\alt_1,\alt_2,\alt_3,\alt_4}$ whose kernel is given by elements of the form: 
\[
\begin{pmatrix}
\kappa^2&0\\
\chi&1
\end{pmatrix}
\times
\begin{pmatrix}
\kappa&0\\
0&\kappa
\end{pmatrix}
\text{ where }
\kappa\in\GroupR
\text{ and }
\chi\in\mathbb R\, .
\]
Quotienting by the latter yields a faithful representation $\GroupR\times\PGLR{2}\to\GLR{4}$ acting on $\mathsf{M}\in\mathbb R^4$ \eqref{equation:matrix M} as $\mathsf{M}\mapsto \lambda\, \mathsf{G}\, \mathsf{M}\, \mathsf{G}\un$  with $\lambda\in\GroupR$ and $\mathsf{G}\in\PGLR{2}$.
The classification of the third branch thus mimics the classification of the first branch via the Frobenius normal form $\mathsf F$. We start by considering the case for which $\mathsf{M}$ is proportional to the identity \ie $\mathsf{M}=\mathsf{F}=\alt_1 \Id_2$ and distinguish between the subcases $\alt_1=0$ and $\alt_1\neq0$:
\bi
\item[]
\bi
\item[$\bullet$] $\alt_1=0$

The resulting algebra corresponds to the central extension of the static algebra:
\mybox{WarmGray}{
$\hypertarget{Eq:AK1b}{\hyperlink{Table:AK1b}{\mathsf{AK1b}}}$
}
{
\begin{tikzcd}[row sep = small,column sep = small,ampersand replacement=\&]
\br{\tK}{\tP}= M
\end{tikzcd}
}
Note that the latter is isomorphic as a Lie algebra to the trivial extension of the carroll algebra \hypertarget{Eq:AK9}{\hyperlink{Table:AK9}{$\mathsf{AK9}$}}, upon the relabelling $M\leftrightarrow H$. However, the latter change of basis fails to preserve the ideal $\ali$, hence the need to keep those two distinct.
 \item[$\bullet$] $\alt_1\neq0$

Rescaling via $\lambda=\frac{1}{\alt_1}$ allows to put $\mathsf{F}$ in the form $\mathsf{F}=\Id_2$, yielding:
\mybox{WarmGray}{
$\hypertarget{Eq:AK2b}{\hyperlink{Table:AK2b}{\mathsf{AK2b}}}$
}
{
\begin{tikzcd}[row sep = small,column sep = small,ampersand replacement=\&]
\br{\tK}{\tP}= M\&\br{H}{\tK}= \tK\&\br{H}{\tP}= \tP\&\br{H}{M}=2M
\end{tikzcd}
}
\ei
\ei
We pursue by considering the case where $\mathsf{M}$ is not proportional to the identity, so that similarity classes are classified by the pair $(\tr,\det)$
modulo $(\tr,\det)\sim(\lambda\, \tr,\lambda^2\det)$. The corresponding equivalence classes are the following:

\bi
\item[]
\bi
\item[$\bullet$] $(0,0)$ if $\tr=\det=0$

This corresponds to the centrally extended galilei algebra (or bargmann algebra) with:
\mybox{WarmGray}{
$\hypertarget{Eq:AK3b}{\hyperlink{Table:AK3b}{\mathsf{AK3b}}}$
}
{
\begin{tikzcd}[row sep = small,column sep = small,ampersand replacement=\&]
\br{\tK}{\tP}= M\&\br{H}{\tP}=\tK
\end{tikzcd}
}

\item[$\bullet$] $(1,0)$ if $\tr\neq0$ and $\det=0$, so that rescaling via $\lambda=\frac1\tr$ yields:
\[
\begin{tikzcd}[row sep = small,column sep = small,ampersand replacement=\&]
\br{H}{\tP}=\tK+\tP\&\br{\tK}{\tP}= M\&\br{H}{M}=M
\end{tikzcd}
\]
Performing a redefinition $\tP\mapsto\tK+\tP$ yields:
\mybox{WarmGray}{
$\hypertarget{Eq:AK4b}{\hyperlink{Table:AK4b}{\mathsf{AK4b}}}$
}
{
\begin{tikzcd}[row sep = small,column sep = small,ampersand replacement=\&]
\br{\tK}{\tP}= M\&\br{H}{\tP}=\tP\&\br{H}{M}=M
\end{tikzcd}
}
\item[$\bullet$] $(0,1)$ if $\tr=0$, $\det >0$

This case reproduces the centrally extended euclidean newton algebra:
\mybox{WarmGray}{
$\hypertarget{Eq:AK5b}{\hyperlink{Table:AK5b}{\mathsf{AK5b}}}$
}
{
\begin{tikzcd}[row sep = small,column sep = small,ampersand replacement=\&]
\br{\tK}{\tP}= M\&\br{H}{\tK}=-\tP\&\br{H}{\tP}=\tK
\end{tikzcd}
}
\item[$\bullet$] $(0,-1)$ if $\tr=0$, $\det <0$

We recover the centrally extended lorentzian newton algebra:
\mybox{WarmGray}{
$\hypertarget{Eq:AK6b}{\hyperlink{Table:AK6b}{\mathsf{AK6b}}}$
}
{
\begin{tikzcd}[row sep = small,column sep = small,ampersand replacement=\&]
\br{\tK}{\tP}= M\&\br{H}{\tK}=\tP\&\br{H}{\tP}=\tK
\end{tikzcd}
}
\item[$\bullet$] $(\alpha_+,1)$  if $\tr\neq0$, $\det >0$
\mybox{WarmGray}{
$\hypertarget{Eq:AK7b}{\hyperlink{Table:AK7b}{\mathsf{AK7b}_{\alpha_+}}}$
}
{
\begin{tikzcd}[row sep = small,column sep = small,ampersand replacement=\&]
\br{\tK}{\tP}= M\&\br{H}{\tK}=-\tP\&\br{H}{\tP}=\tK+\alpha_+\tP\&\br{H}{M}=\alpha_+ M\&\alpha_+>0
\end{tikzcd}
}
\item[$\bullet$] $(\alpha_-,-1)$  if $\tr\neq0$, $\det <0$
\mybox{WarmGray}{
$\hypertarget{Eq:AK8b}{\hyperlink{Table:AK8b}{\mathsf{AK8b}_{\alpha_-}}}$
}
{
\begin{tikzcd}[row sep = small,column sep = small,ampersand replacement=\&]
\br{\tK}{\tP}= M\&\br{H}{\tK}=\tP\&\br{H}{\tP}=\tK+\alpha_-\tP\&\br{H}{M}=\alpha_- M\&\alpha_->0
\end{tikzcd}
}
\ei
\ei
We conclude this section by checking whether the ambient aristotelian algebras we classified admit additional scalar ideals [on top of $\ali=\Span M$]. It can be easily verified that, among the previously classified algebras, only \hyperlink{Table:AK1}{$\mathsf{AK1}$}, \hyperlink{Table:AK9}{$\mathsf{AK9}$} and \hyperlink{Table:AK1b}{$\mathsf{AK1b}$} admit more than one scalar ideal, namely the linear combination:
\[
S=x\, M+y\, H\text{ where }(x,y)\in\mathbb R^2\setminus (0,0)\, .
\]
For each of these three algebras, we check whether there exists an automorphism $\varphi\in\Aut(\alg)$ such that $\varphi(\ali')=\ali$, where $\ali'=\Span S$ and $\ali=\Span M$. This boils down to classify the orbits of the intersection $\Aut(\alg)\cap\GLR{2}$ [\ie scalar automorphisms] acting on $\mathbb R^2\setminus (0,0)$. In case $S$ and $M$ do not belong to the same orbit, we perform a change of adapted basis allowing ro rewrite $S$ as $M$ [at the expense of modifying the associated commutation relations]. 
\bi
\item \hyperlink{Table:AK1}{$\mathsf{AK1}$}

There is a unique orbit of $\Aut(\hyperlink{Table:AK1}{\mathsf{AK1}})\cap\GLR{2}=\GLR{2}$ on $\mathbb R^2\setminus (0,0)$, hence there is a unique pair $(\alg,\ali)$ with underlying Lie algebra $\hyperlink{Table:AK1}{\mathsf{AK1}}$ and the ideal $\ali$ can always be chosen to be $\ali=\Span M$.
\item \hyperlink{Table:AK9}{$\mathsf{AK9}$} 

The group of scalar automorphisms $\Aut(\hyperlink{Table:AK9}{\mathsf{AK9}})\cap\GLR{2}$ of the trivial extension of the carroll algebra is parameterised by elements of the form:
\bea
\varphi=
\begin{pmatrix}
\mu&\nu\\
0&1
\end{pmatrix}
\text{ with }\mu\in\GroupR\text{, }\nu\in\mathbb R\, .
\eea
There are two orbits on $\mathbb R^2\setminus (0,0)$:
\be
\item $\begin{pmatrix}
1&
0
\end{pmatrix}$ corresponding to $\ali_1=\ali=\Span M$, with commutation relation $\br{\tK}{\tP}=H$. 
\item $\begin{pmatrix}
0&
1
\end{pmatrix}$ so that $\ali_2=\Span H$. Performing a change of basis, we can map $\ali_2$ to $\ali=\Span M$, with commutation relation $\br{\tK}{\tP}=M$, hence the couple $(\hyperlink{Table:AK9}{\mathsf{AK9}} ,\ali_2)$ is isomorphic to $(\hyperlink{Table:AK1b}{\mathsf{AK1b}},\ali)$.
\ee

\item \hyperlink{Table:AK1b}{$\mathsf{AK1b}$} 

The situation is dual to the one for \hyperlink{Table:AK9}{$\mathsf{AK9}$}.
\ei
We conclude that the set consisting of the three above algebras induces a set of three pairs $(\hyperlink{Table:AK1}{\mathsf{AK1}} ,\ali)$, $(\hyperlink{Table:AK9}{\mathsf{AK9}} ,\ali)$ and $(\hyperlink{Table:AK1b}{\mathsf{AK1b}} ,\ali)$, where the scalar ideal $\ali$ can always be chosen to be $\ali=\Span M$. 
\medskip

The isomorphism classes of pairs $(\alg,\ali)$---where $\alg$ is an ambient aristotelian algebra and $\ali=\Span M$ a scalar ideal---corresponding to the above three branches are collected in Table \ref{Table:Ambient aristotelian algebras admitting a scalar ideal}. 

\begin{Remark}
Deformations of the universal central extension of the static kinematical algebra [\ie the ambient aristotelian algebra whose only non-trivial commutator is $\br{\tK}{\tP}=M$] have been classified for $d>3$ in \cite[Table 18]{Figueroa-OFarrill2017a}. This classification intersects the second and third branches of our classification \UnskipRef{Table:Ambient aristotelian algebras admitting a scalar ideal} [but not the first branch since the latter is characterised by $\alt_6=\alt_7=0$ yielding $\br{\tK}{\tP}=0$]. The following table can be seen as a Pyrgi Tablet between the two classifications:

\begin{center}
\resizebox{12cm}{!}{
\begin{tabular}{l|l}
Table 18 of \cite{Figueroa-OFarrill2017a}&Table \ref{Table:Ambient aristotelian algebras admitting a scalar ideal}\\\hline
Eq.58 (centrally extended static)& \hyperlink{Table:AK1b}{$\mathsf{AK1b}$}, \hyperlink{Table:AK9}{$\mathsf{AK9}$}\\
\rowcolor{Gray}
Eq.89 (central extension of lorentzian Newton)& \hyperlink{Table:AK6b}{$\mathsf{AK6b}$}\\
Eq.104 (central extension of euclidean Newton)& \hyperlink{Table:AK5b}{$\mathsf{AK5b}$}\\
\rowcolor{Gray}
Eq.98 (Bargmann)& \hyperlink{Table:AK3b}{$\mathsf{AK3b}$}\\
\hline
Eq.79 ($\mathfrak{e}\oplus\mathbb R$)& \hyperlink{Table:AK10}{$\mathsf{AK10}$}\\
\rowcolor{Gray}
Eq.79 ($\mathfrak{p}\oplus\mathbb R$)& \hyperlink{Table:AK11}{$\mathsf{AK11}$}\\
Eq.77 ($\so(d+1,1)\oplus\mathbb R$)& \hyperlink{Table:AK12}{$\mathsf{AK12}$}\\
\rowcolor{Gray}
Eq.82 ($\so(d+2)\oplus\mathbb R$)& \hyperlink{Table:AK13}{$\mathsf{AK13}$}\\
Eq.82 ($\so(d,2)\oplus\mathbb R$)& \hyperlink{Table:AK14}{$\mathsf{AK14}$}\\
\rowcolor{Gray}
\hline
Eq.74& \hyperlink{Table:AK2b}{$\mathsf{AK2b}$}\\
Eq.86 with $-1<\gamma<0$& \hyperlink{Table:AK8b}{$\mathsf{AK8b}_{\alpha_-}$} with $\alpha_->0$ where $\alpha_-=\frac{\gamma+1}{\sqrt{-\gamma}}$\\
\rowcolor{Gray}
Eq.86 with $\gamma=0$& \hyperlink{Table:AK4b}{$\mathsf{AK4b}$}\\
Eq.86 with $0<\gamma<1$&\hyperlink{Table:AK7b}{$\mathsf{AK7b}_{\alpha_+}$} with $\alpha_+>2$ where $\alpha_+=\frac{\gamma+1}{\sqrt{\gamma}}$\\
\rowcolor{Gray}
Eq.101&\hyperlink{Table:AK7b}{$\mathsf{AK7b}_{\alpha_+}$} with $\alpha_+=2$\\
Eq.107 with $\alpha>0$&\hyperlink{Table:AK7b}{$\mathsf{AK7b}_{\alpha_+}$} with $0<\alpha_+<2$ where $\alpha_+=\frac{2\alpha}{\sqrt{1+\alpha^2}}$
  \end{tabular}
  }
  \end{center}
\end{Remark}

\begin{table}[ht]
\centering
\resizebox{16cm}{!}{
\begin{tabular}{l|l|llllll}
\multicolumn{1}{c|}{\textbf{Label}}&\multicolumn{1}{c|}{\textbf{Comments}}&\multicolumn{6}{c}{\textbf{Non-trivial commutators}}\\\hline
\hypertarget{Table:AK1}{\hyperlink{Eq:AK1}{$\mathsf{AK1}$}}&$\text{static}\oplus\mathbb R$&&&&&&\\
\rowcolor{Gray}
\hypertarget{Table:AK1a}{\hyperlink{Eq:AK1a}{$\mathsf{AK1a}$}}&$\text{static}\inplus\mathbb R$&&&$\br{H}{M}=M$&&&\\
\hypertarget{Table:AK2}{\hyperlink{Eq:AK2}{$\mathsf{AK2}$}}&&$\br{H}{\tK}= \tK$&$\br{H}{\tP}=\tP$&&&&\\
\rowcolor{Gray}
\hypertarget{Table:AK2+}{\hyperlink{Eq:AK2}{$\mathsf{AK2}^+_{\kappa}$}}&$\kappa>0$&$\br{H}{\tK}= \tK$&$\br{H}{\tP}=\tP$&$\br{H}{M}=\kappa M$&&&\\
\hypertarget{Table:AK2-}{\hyperlink{Eq:AK2}{$\mathsf{AK2}^-_{\kappa}$}}&$\kappa>0$&$\br{H}{\tK}= \tK$&$\br{H}{\tP}=\tP$&$\br{H}{M}=-\kappa M$&&&\\
\rowcolor{Gray}
\hypertarget{Table:AK3}{\hyperlink{Eq:AK3}{$\mathsf{AK3}$}}&$\text{galilei}\oplus\mathbb R$&&$\br{H}{\tP}=\tK$&&&&\\
\hypertarget{Table:AK3a}{\hyperlink{Eq:AK3a}{$\mathsf{AK3a}$}}&$\text{galilei}\inplus\mathbb R$&&$\br{H}{\tP}=\tK$&$\br{H}{M}=M$&&&\\
\rowcolor{Gray}
\hypertarget{Table:AK4}{\hyperlink{Eq:AK4}{$\mathsf{AK4}$}}&&&$\br{H}{\tP}=\tP$&&&&\\
\hypertarget{Table:AK4+}{\hyperlink{Eq:AK4}{$\mathsf{AK4}^+_\kappa$}}&$\kappa>0$&&$\br{H}{\tP}=\tP$&$\br{H}{M}=\kappa M$&&&\\
\rowcolor{Gray}
\hypertarget{Table:AK4-}{\hyperlink{Eq:AK4}{$\mathsf{AK4}^-_\kappa$}}&$\kappa>0$&&$\br{H}{\tP}=\tP$&$\br{H}{M}=-\kappa M$&&&\\
\hypertarget{Table:AK5}{\hyperlink{Eq:AK5}{$\mathsf{AK5}$}}&$\text{euclidean newton}\oplus\mathbb R$&$\br{H}{\tK}=-\tP$&$\br{H}{\tP}=\tK$&&&&\\
\rowcolor{Gray}
\hypertarget{Table:AK5+}{\hyperlink{Eq:AK5}{$\mathsf{AK5}^+_{\kappa}$}}&$\kappa>0$&$\br{H}{\tK}=-\tP$&$\br{H}{\tP}=\tK$&$\br{H}{M}=\kappa M$&&&\\
\hypertarget{Table:AK6}{\hyperlink{Eq:AK6}{$\mathsf{AK6}$}}&$\text{lorentzian newton}\oplus\mathbb R$&$\br{H}{\tK}=\tP$&$\br{H}{\tP}=\tK$&&&&\\
\rowcolor{Gray}
\hypertarget{Table:AK6+}{\hyperlink{Eq:AK6}{$\mathsf{AK6}^+_{\kappa}$}}&$\kappa>0$&$\br{H}{\tK}=\tP$&$\br{H}{\tP}=\tK$&$\br{H}{M}=\kappa M$&&&\\
\hypertarget{Table:AK7}{\hyperlink{Eq:AK7}{$\mathsf{AK7}_{\alpha_+}$}}&$\alpha_+>0$&$\br{H}{\tK}=-\tP$&$\br{H}{\tP}=\tK+\alpha_+\tP$&&&&\\
\rowcolor{Gray}
\hypertarget{Table:AK7+}{\hyperlink{Eq:AK7}{$\mathsf{AK7}^+_{\alpha_+,\kappa}$}}&$\alpha_+>0\qq\kappa>0$&$\br{H}{\tK}=-\tP$&$\br{H}{\tP}=\tK+\alpha_+\tP$&$\br{H}{M}=\kappa M$&&&\\
\hypertarget{Table:AK7-}{\hyperlink{Eq:AK7}{$\mathsf{AK7}^-_{\alpha_+,\kappa}$}}&$\alpha_+>0\qq\kappa>0$&$\br{H}{\tK}=-\tP$&$\br{H}{\tP}=\tK+\alpha_+\tP$&$\br{H}{M}=-\kappa M$&&&\\
\rowcolor{Gray}
\hypertarget{Table:AK8}{\hyperlink{Eq:AK8}{$\mathsf{AK8}_{\alpha_-}$}}&$\alpha_->0$&$\br{H}{\tK}=\tP$&$\br{H}{\tP}=\tK+\alpha_-\tP$&&&&\\
\hypertarget{Table:AK8+}{\hyperlink{Eq:AK8}{$\mathsf{AK8}^+_{\alpha_-,\kappa}$}}&$\alpha_->0\qq\kappa>0$&$\br{H}{\tK}=\tP$&$\br{H}{\tP}=\tK+\alpha_-\tP$&$\br{H}{M}=\kappa M$&&&\\
\rowcolor{Gray}
\hypertarget{Table:AK8-}{\hyperlink{Eq:AK8}{$\mathsf{AK8}^-_{\alpha_-,\kappa}$}}&$\alpha_->0\qq\kappa>0$&$\br{H}{\tK}=\tP$&$\br{H}{\tP}=\tK+\alpha_-\tP$&$\br{H}{M}=-\kappa M$&&&\\
\hline
\hypertarget{Table:AK9}{\hyperlink{Eq:AK9}{$\mathsf{AK9}$}}&$\text{carroll}\oplus\mathbb R$&&&&$\br{\tK}{\tP}=H$&&\\
\rowcolor{Gray}
\hypertarget{Table:AK10}{\hyperlink{Eq:AK10}{$\mathsf{AK10}$}}&$\mathfrak{iso}(d+1)\oplus\mathbb R$&&$\br{H}{\tP}=-\tK$&&$\br{\tK}{\tP}=H$&&$\br{\tP}{\tP}=-\tJ$\\
\hypertarget{Table:AK11}{\hyperlink{Eq:AK11}{$\mathsf{AK11}$}}&$\mathfrak{iso}(d,1)\oplus\mathbb R$&&$\br{H}{\tP}=\tK$&&$\br{\tK}{\tP}=H$&&$\br{\tP}{\tP}=\tJ$\\
\rowcolor{Gray}
\hypertarget{Table:AK12}{\hyperlink{Eq:AK12}{$\mathsf{AK12}$}}&$\so(d+1,1)\oplus\mathbb R$&$\br{H}{\tK}=\tP$&$\br{H}{\tP}=\tK$&&$\br{\tK}{\tP}=H$&$\br{\tK}{\tK}=-\tJ$&$\br{\tP}{\tP}=\tJ$\\
\hypertarget{Table:AK13}{\hyperlink{Eq:AK13}{$\mathsf{AK13}$}}&$\so(d+2)\oplus\mathbb R$&$\br{H}{\tK}=\tP$&$\br{H}{\tP}=-\tK$&&$\br{\tK}{\tP}=H$&$\br{\tK}{\tK}=-\tJ$&$\br{\tP}{\tP}=-\tJ$\\
\rowcolor{Gray}
\hypertarget{Table:AK14}{\hyperlink{Eq:AK14}{$\mathsf{AK14}$}}&$\so(d,2)\oplus\mathbb R$&$\br{H}{\tK}=-\tP$&$\br{H}{\tP}=\tK$&&$\br{\tK}{\tP}=H$&$\br{\tK}{\tK}=\tJ$&$\br{\tP}{\tP}=\tJ$\\
\hline
\hypertarget{Table:AK1b}{\hyperlink{Eq:AK1b}{$\mathsf{AK1b}$}}&$\text{centrally extended static}$&&&&$\br{\tK}{\tP}=M$&&\\
\rowcolor{Gray}
\hypertarget{Table:AK2b}{\hyperlink{Eq:AK2b}{$\mathsf{AK2b}$}}&&$\br{H}{\tK}=\tK$&$\br{H}{\tP}=\tP$&$\br{H}{M}=2M$&$\br{\tK}{\tP}=M$&&\\
\hypertarget{Table:AK3b}{\hyperlink{Eq:AK3b}{$\mathsf{AK3b}$}}&{centrally extended galilei}&&$\br{H}{\tP}=\tK$&&$\br{\tK}{\tP}=M$&&\\
\rowcolor{Gray}
\hypertarget{Table:AK4b}{\hyperlink{Eq:AK4b}{$\mathsf{AK4b}$}}&&&$\br{H}{\tP}=\tP$&$\br{H}{M}=M$&$\br{\tK}{\tP}=M$&&\\
\hypertarget{Table:AK5b}{\hyperlink{Eq:AK5b}{$\mathsf{AK5b}$}}&centrally extended euclidean newton&$\br{H}{\tK}=-\tP$&$\br{H}{\tP}=\tK$&&$\br{\tK}{\tP}=M$&&\\
\rowcolor{Gray}
\hypertarget{Table:AK6b}{\hyperlink{Eq:AK6b}{$\mathsf{AK6b}$}}&centrally extended lorentzian newton&$\br{H}{\tK}=\tP$&$\br{H}{\tP}=\tK$&&$\br{\tK}{\tP}=M$&&\\
\hypertarget{Table:AK7b}{\hyperlink{Eq:AK7b}{$\mathsf{AK7b}_{\alpha_+}$}}&$\alpha_+>0$&$\br{H}{\tK}=-\tP$&$\br{H}{\tP}=\tK+\alpha_+\tP$&$\br{H}{M}=\alpha_+M$&$\br{\tK}{\tP}=M$&&\\
\rowcolor{Gray}
\hypertarget{Table:AK8b}{\hyperlink{Eq:AK8b}{$\mathsf{AK8b}_{\alpha_-}$}}&$\alpha_->0$&$\br{H}{\tK}=\tP$&$\br{H}{\tP}=\tK+\alpha_-\tP$&$\br{H}{M}=\alpha_-M$&$\br{\tK}{\tP}=M$&&
  \end{tabular}
}
      \caption{Universal ambient aristotelian algebras $\alg$ admitting a scalar ideal $\ali=\Span M$}
      \label{Table:Ambient aristotelian algebras admitting a scalar ideal}
\end{table}

\subsection{Effective Klein pairs}
\label{section:Effective Klein pairs}
Having classified ambient aristotelian algebras [admitting a scalar ideal] in Section \ref{section:Ambient aristotelian algebras admitting a scalar ideal}, we now aim to identify the corresponding ambient aristotelian Klein pairs \UnskipRef{Definition:Ambient aristotelian Klein pairs}. Explicitly, we identify for each ambient aristotelian algebra $\alg$ compiled in Table \ref{Table:Ambient aristotelian algebras admitting a scalar ideal} the associated subalgebras $\alh\subset\alg$ of type $(0,1,1)$ \UnskipRef{Definition:Subspace of type}, and determine whether the Klein pair $(\alg,\alh)$ is effective or not. Effective Klein pairs will be collected in Table \ref{Table:Effective ambient aristotelian Klein pairs with scalar ideal}. For each non-effective Klein pair found, we explicit the corresponding  associated Klein pair \UnskipRef{Definition:Infinitesimal Klein pair}, the latter being lifshitzian [see Table \ref{Table:Lifshitzian Klein pairs}]. We will proceed along the lines followed in \cite[Section 3]{Figueroa-OFarrill2018} regarding the classification of effective kinematical Klein pairs \UnskipRef{Table:Effective kinematical Klein pairs}, procedure that we now briefly review:
\be
\item The first step consists in determining all possible subalgebras $\alh$ of type $(0,1,1)$ of $\alg$ \UnskipRef{Definition:Subspace of type}. Since the generator $\tJ$ is part of the $\so(d)$-structure of $\alg$, what is left to determine is the vectorial generator---denoted $\tX$ in the following---being a linear combination of the two vectorial generators $\tK$ and $\tP$. For definiteness, let us denote: 
\[
\tX=x\, \tK+y\, \tP\text{ where }(x,y)\in\mathbb R^2\setminus (0,0)
\]
hence $\tX$ is a non-trivial vectorial generator and $\alh=\tX\oplus\tJ$. By assumption, the canonical generators \UnskipRef{Definition:Ambient aristotelian algebra} ensure that $\br{\tJ}{\tJ}\sim\tJ\in\alh$ and $\br{\tJ}{\tX}\sim\tX\in\alh$, so that the only non-trivial condition for $\alh$ to be a subalgebra is that $\br{\tX}{\tX}\in\alh$. Inspection of the commutators in Table \ref{Table:Ambient aristotelian algebras admitting a scalar ideal} reveals that $\alh$ is a subalgebra for all algebras in the classification and for any value of $(x,y)\in\mathbb R^2\setminus (0,0)$.
\item Recall that two ambient aristotelian Klein pairs $(\alg,\alh_1)$ and $(\alg,\alh_2)$ are isomorphic if there exists an automorphism $\varphi\From\alg\to\alg$ [preserving the underlying $\so(d)$-decomposition] and such that $\varphi(\alh_1)=\alh_2$. For each subalgebra $\alh$ previously found, we thus check whether there exists an automorphism $\varphi\in\Aut(\alg,\ali)$ of $\alg$  [preserving the ideal $\ali=\Span M$] such that $\varphi(\alh)=\alh_0$, where $\alh=\tX\oplus\tJ$ and $\alh_0=\tK\oplus\tJ$. This boils down to classify the orbits of the intersection $\Aut(\alg,\ali)\cap\GLR{2}$ acting on $\mathbb R^2\setminus (0,0)$. In case $\tX$ and $\tK$ do not belong to the same orbit, we perform a change of adapted basis \eqref{equation:subgroup} allowing ro rewrite $\tX$ as $\tK$ [at the expense of modifying the associated commutation relations]. 
\item For all ambient aristotelian Klein pairs $(\alg,\alh)$ found, we check whether the pair is effective or not. Similarly to the case of kinematical Klein pairs [\cf footnote \ref{footnote:effectiveness}], an ambient aristotelian Klein pair is effective if and only if the generators $\tK\in\alh$ do not span an ideal of $\alg$.\footnote{Recall that the canonical relation $\br{\tJ}{\tP}=\tP$ prevents $\tJ$ to belong to an ideal contained in $\alh$.} The possible obstructions\footnote{As will be expanded upon in Section \ref{section:Projectable ambient triplets}, different choices of obstruction yield 
different $\ad_\alh$-invariant metric structures, as the set of commutators $\br{\alh}{\alp}\subset\alp$ determine what (possibly degenerate) metric structure on $\alp$ is $\ad_\alh$-invariant, \eg:
\bi
\item $\mathsf{G}$-ambient:  $\br{\tK}{H}\sim \tP$
\item bargmannian: $\br{\tK}{\tP}\sim M$ \qq $\br{\tK}{H}\sim \tP$
\item carrollian$\oplus\mathbb R$: $\br{\tK}{\tP}\sim H$
\item (pseudo)-riemannian$\oplus\mathbb R$: $\br{\tK}{H}\sim \tP$ \qq $\br{\tK}{\tP}\sim H$.
\ei} for $\tK$ to span an ideal of $\alg$ take the form: 
\bea
\label{equation:obstructions}
\hspace{-1.5cm}\begin{tikzcd}[row sep = small,column sep = small,ampersand replacement=\&]
\br{\tK}{\tK}\sim \tJ \& \br{\tK}{M}\sim \tP \& \br{\tK}{H}\sim \tP\&\br{\tK}{\tP}\sim M\&\br{\tK}{\tP}\sim H\&\br{\tK}{\tP}\sim \tJ\, .\hspace{-1cm}
\end{tikzcd}
\eea
Whenever the above obstructions all vanish, the pair is non-effective and then excluded from Table \ref{Table:Effective ambient aristotelian Klein pairs with scalar ideal}. 
We then identify the associated effective Klein pair as one of the lifshitzian pairs displayed in Table \ref{Table:Lifshitzian Klein pairs}.
\ee
Let us now apply the previously outlined procedure to the Lie algebras classified in Table \ref{Table:Ambient aristotelian algebras admitting a scalar ideal}, starting with the first branch \UnskipRef{section:Ambient aristotelian algebras admitting a scalar ideal}. 

\paragraph{First branch}
\bi
\item $\hyperlink{Table:AK1}{\mathsf{AK1}}$  
\hypertarget{Eq:AK1h0}{}

The group of automorphisms $\Aut(\hyperlink{Table:AK1}{\mathsf{AK1}},\ali)$ of the trivial extension of the static algebra coincides with the full group $\TR{2}\times\GLR{2}$ \eqref{equation:subgroup}, parameterised by elements of the form:
\bea
\label{equation:varphiaction}
\varphi=
\begin{pmatrix}
\mu&0\\
\chi&\lambda
\end{pmatrix}
\times
\underbrace{\begin{pmatrix}
\alpha&\beta\\
\gamma&\delta
\end{pmatrix}}_{\mathsf{G}}
\text{ with }\mu,\lambda\in\GroupR\text{, }\chi\in\mathbb R\text{ and }\mathsf{G}\in\GLR{2}\, .
\eea
There is a unique orbit of $\Aut(\hyperlink{Table:AK1}{\mathsf{AK1}},\ali)\cap\GLR{2}=\GLR{2}$ on $\mathbb R^2\setminus (0,0)$, hence there is a unique ambient aristotelian Klein pair $\hyperlink{Table:AK1h0}{(\mathsf{AK1},\alh_0)}$ with underlying Lie algebra $\hyperlink{Table:AK1}{\mathsf{AK1}}$. The latter is not effective since the boost generators $\tK$ span an ideal of $\hyperlink{Table:AK1}{\mathsf{AK1}}$. As mentioned previously, the generators $\tJ$ do not belong to an ideal of $\alg$ contained in $\alh$ hence $\tK$ spans the ideal core $\aln$ of $(\alg,\alh)$ \UnskipRef{Definition:Infinitesimal Klein pair}. The associated effective Klein pair $(\alg/\aln,\alh/\aln)$ is isomorphic to the lifshitzian Klein pair $\hyperlink{Table:Li1}{\mathsf{Li1}}$.

\item $\hyperlink{Table:AK1a}{\mathsf{AK1a}}$ 
\hypertarget{Eq:AK1ah0}{}

The automorphism group $\Aut(\hyperlink{Table:AK1a}{\mathsf{AK1a}},\ali)$ is given by elements of the form \eqref{equation:varphiaction} with $\lambda=1$. A reasoning similar to the one above thus shows that there is a unique ambient aristotelian Klein pair $\hyperlink{Table:AK1ah0}{(\mathsf{AK1a},\alh_0)}$ with underlying Lie algebra $\hyperlink{Table:AK1a}{\mathsf{AK1a}}$, the latter being non-effective with associated effective Klein pair $\hyperlink{Table:Li2}{\mathsf{Li2}}$.

\item $\hyperlink{Table:AK2}{\mathsf{AK2}_\kappa}$ 
\hypertarget{Eq:AK2h0}{}

Similarly to the previous case, the automorphism group $\Aut(\hyperlink{Table:AK2}{\mathsf{AK2}_\kappa},\ali)$ is given by elements of the form \eqref{equation:varphiaction} with $\lambda=1$ so that there is a unique ambient aristotelian Klein pair $\hyperlink{Table:AK2h0}{(\mathsf{AK2},\alh_0)}$ with underlying Lie algebra $\hyperlink{Table:AK2}{\mathsf{AK2}_\kappa}$, the latter being non-effective with associated effective Klein pair $\hyperlink{Table:Li3}{\mathsf{Li3}_z}$ [with $z=0$ for $\hyperlink{Table:AK2}{\mathsf{AK2}}$, $z=\kappa$ for $\hyperlink{Table:AK2+}{\mathsf{AK2}^+_\kappa}$ and $z=-\kappa$ for $\hyperlink{Table:AK2-}{\mathsf{AK2}^-_\kappa}$].

\item $\hyperlink{Table:AK3}{\mathsf{AK3}}$ 

The automorphism group $\Aut(\hyperlink{Table:AK3}{\mathsf{AK3}},\ali)$ is given by elements of the form \eqref{equation:varphiaction} with $\alpha=\lambda\delta$, $\beta=0$. 

There are two orbits $\begin{pmatrix}
x&
y
\end{pmatrix}$ to distinguish:
\be
\hypertarget{Eq:AK3h0}{}
\item $\begin{pmatrix}
1&
0
\end{pmatrix}$ whenever $y=0$.

All obstructions \eqref{equation:obstructions} vanish hence the pair $\hyperlink{Table:AK3h0}{(\mathsf{AK3},\alh_0)}$ is non-effective with associated effective Klein pair $\hyperlink{Table:Li1}{\mathsf{Li1}}$.

\item $\begin{pmatrix}
0&
1
\end{pmatrix}$ whenever $y\neq0$.

In this case, there is no automorphism of $\hyperlink{Table:AK3}{\mathsf{AK3}}$ mapping $\alh_1=\alh=\tX\oplus\tJ$ to $\alh_0=\tK\oplus\tJ$, so that $\hyperlink{Table:AK3h0}{(\mathsf{AK3},\alh_0)}$ and $\hyperlink{Table:AK3h1}{(\mathsf{AK3},\alh_1)}$ are not isomorphic as Klein pairs. 

To facilitate comparison, we perform a change of adapted basis \eqref{equation:subgroup}
in order to bring $\alh_1$ to the canonical form $\alh_1=\tK\oplus\tJ$. This change of basis does not preserve the commutation relations that now read:
\mybox{Ice}{
$\hypertarget{Eq:AK3h1}{\hyperlink{Table:AK3h1}{(\mathsf{AK3},\alh_1)}}\quad\text{galilei}\oplus\mathbb R$
}
{
\begin{tikzcd}[row sep = small,column sep = small,ampersand replacement=\&]
\br{\tK}{H}=\tP
\end{tikzcd}
}
The unique non-trivial commutator corresponds to one of the obstructions listed in \eqref{equation:obstructions} so that the resulting Klein pair $\hyperlink{Table:AK3h1}{(\mathsf{AK3},\alh_1)}$ is effective.
\ee
\item $\hyperlink{Table:AK3a}{\mathsf{AK3a}}$ 

The treatment of this case is parallel to the previous case $\hyperlink{Table:AK3}{\mathsf{AK3}}$ [with the sole exception that the extra commutator $\br{H}{M}=M$ imposes $\lambda=1$ in \eqref{equation:varphiaction}, which does not affect the analysis]. There are thus two associated Klein pairs:
\be
\item\hypertarget{Eq:AK3ah0}{} a non-effective Klein pair $\hyperlink{Table:AK3ah0}{(\mathsf{AK3a},\alh_0)}$ with commutation relations as in Table \ref{Table:Ambient aristotelian algebras admitting a scalar ideal} and with associated effective Klein pair $\hyperlink{Table:Li2}{\mathsf{Li2}}$.
\item an effective Klein pair  $\hyperlink{Table:AK3ah1}{(\mathsf{AK3a},\alh_1)}$ with commutation relations:
\mybox{Ice}{
$\hypertarget{Eq:AK3ah1}{\hyperlink{Table:AK3ah1}{(\mathsf{AK3a},\alh_1)}}\quad\text{galilei}\inplus\mathbb R$
}
{
\begin{tikzcd}[row sep = small,column sep = small,ampersand replacement=\&]
\br{\tK}{H}=\tP\&\br{H}{M}= M
\end{tikzcd}
}
\ee

\item $\hyperlink{Table:AK4}{\mathsf{AK4}_\kappa}$ 

Computing the automorphism group $\hyperlink{Table:AK4}{\mathsf{AK4}_\kappa}$, the latter is given by elements \eqref{equation:varphiaction} with $\lambda=1$ and $\beta=\gamma=0$, for all values of the parameter $\kappa\in\mathbb R$. There are three orbits to consider:
\be
\item\hypertarget{Eq:AK4h0}{}  $\begin{pmatrix}
1&
0
\end{pmatrix}$ whenever $y=0$ (hence $x\neq0$) so that $\alh=\alh_0$. All obstructions \eqref{equation:obstructions} vanish hence the pair $\hyperlink{Table:AK4h0}{(\mathsf{AK4}_\kappa,\alh_0)}$ is non-effective, with associated effective Klein pair $\hyperlink{Table:Li3}{\mathsf{Li3}_z}$ [with $z=0$ for $\hyperlink{Table:AK4}{\mathsf{AK4}}$, $z=\kappa$ for $\hyperlink{Table:AK4+}{\mathsf{AK4}^+_\kappa}$ and $z=-\kappa$ for $\hyperlink{Table:AK4-}{\mathsf{AK4}^-_\kappa}$]. 
\item \hypertarget{Eq:AK4h1}{}$\begin{pmatrix}
0&
1
\end{pmatrix}$ whenever $x=0$ (hence $y\neq0$). Performing a change of basis 
yields:
\[
\begin{tikzcd}[row sep = small,column sep = small,ampersand replacement=\&]
\hyperlink{Table:AK4h1}{(\mathsf{AK4}_\kappa,\alh_1)}\&\br{H}{\tK}=\tK\&\br{H}{M}=\kappa M\&\text{ where }\kappa\in\mathbb R\text{ is arbitrary.}
\end{tikzcd}
\]
The generator $\tK$ spans a non-trivial ideal hence the pair $\hyperlink{Table:AK4h1}{(\mathsf{AK4}_\kappa,\alh_1)}$ is non-effective, with associated effective Klein pair $\hyperlink{Table:Li1}{\mathsf{Li1}}$ for $\kappa=0$ and $\hyperlink{Table:Li2}{\mathsf{Li2}}$ for $\kappa\neq0$ [after rescaling of $H$].
\item $\begin{pmatrix}
1&
1
\end{pmatrix}$ whenever $x\neq0$ and $y\neq0$. A change of adapted basis 
allows to put the pair in the following form:
\mybox{Ice}{
$\hypertarget{Eq:AK4h2}{\hyperlink{Table:AK4h2}{(\mathsf{AK4}_\kappa,\alh_2)}}$
}
{
\begin{tikzcd}[row sep = small,column sep = small,ampersand replacement=\&]
\br{\tK}{H}=\tP\&\br{H}{\tP}=\tP\&\br{H}{M}=\kappa M\&\text{ where }\kappa\in\mathbb R\text{ is arbitrary}
\end{tikzcd}
}
The obstruction $\br{\tK}{H}=\tP$ is hit, hence the pair $\hypertarget{Eq:AK4h2}{\hyperlink{Table:AK4h2}{(\mathsf{AK4}_\kappa,\alh_2)}}$ is effective.
\ee

\item $\hyperlink{Table:AK5}{\mathsf{AK5}}$ 

The automorphism group $\Aut(\hyperlink{Table:AK5}{\mathsf{AK5}},\ali)$ is given by elements of the form \eqref{equation:varphiaction} with $\alpha=\epsilon\delta$, $\beta=-\epsilon\gamma$ and $\lambda=\epsilon$, with $\epsilon\in\pset{-1,1}$. There is a single orbit corresponding to the effective Klein pair isomorphic to the trivial extension of the euclidean newton Klein pair:
\bi
\item[]
\mybox{Ice}{
$\hypertarget{Eq:AK5h0}{\hyperlink{Table:AK5h0}{(\mathsf{AK5},\alh_0)}}\quad\text{euclidean newton}\oplus\mathbb R$
}
{
\begin{tikzcd}[row sep = small,column sep = small,ampersand replacement=\&]
\br{\tK}{H}=\tP\&\br{H}{\tP}=\tK
\end{tikzcd}
}
\ei
\item $\hyperlink{Table:AK5+}{\mathsf{AK5}^+_\kappa}$ 

The automorphism group $\Aut(\hyperlink{Table:AK5+}{\mathsf{AK5}^+_\kappa},\ali)$ is given by elements of the form \eqref{equation:varphiaction} with $\alpha=\delta$, $\beta=-\gamma$ and $\lambda=1$. There is a single orbit yielding the effective Klein pair:
\bi
\item[]
\mybox{Ice}{
$\hypertarget{Eq:AK5+h0}{\hyperlink{Table:AK5+h0}{(\mathsf{AK5}^+_\kappa,\alh_0)}}$
}
{
\begin{tikzcd}[row sep = small,column sep = small,ampersand replacement=\&]
\br{\tK}{H}=\tP\&\br{H}{\tP}=\tK\&\br{H}{M}=\kappa M\&\text{ where }\kappa>0
\end{tikzcd}
}
\ei

\item $\hyperlink{Table:AK6}{\mathsf{AK6}}$ 

The automorphism group $\Aut(\hyperlink{Table:AK6}{\mathsf{AK6}},\ali)$ is given by elements of the form \eqref{equation:varphiaction} with $\alpha=\epsilon\delta$, $\beta=\epsilon\gamma$ and $\lambda=\epsilon$, with $\epsilon\in\pset{-1,1}$. Contradistinctly to the case $\hyperlink{Table:AK5}{\mathsf{AK5}}$, there are two orbits to distinguish:
\be
\item $\begin{pmatrix}
1&
0
\end{pmatrix}$ whenever $x^2-y^2\neq0$ so that $\alh=\alh_0$. The corresponding effective Klein pair is isomorphic to the trivial extension of the lorentzian newton Klein pair:
\bi
\item[]
\mybox{Ice}{
$\hypertarget{Eq:AK6h0}{\hyperlink{Table:AK6h0}{(\mathsf{AK6},\alh_0)}}
\quad\text{lorentzian newton}\oplus\mathbb R$
}
{
\begin{tikzcd}[row sep = small,column sep = small,ampersand replacement=\&]
\br{\tK}{H}=\tP\&\br{H}{\tP}=-\tK
\end{tikzcd}
}
\ei
\item\hypertarget{Eq:AK6h1}{} $\begin{pmatrix}
1&
1
\end{pmatrix}$ whenever $x^2-y^2=0$.\footnote{Note that there are \textit{a priori} two orbits to distinguish, corresponding to the two branches of solutions $x=\pm y$. However, the sign freedom mediated by $\epsilon$ in $\Aut(\hyperlink{Table:AK6}{\mathsf{AK6}},\ali)$ allows to map one branch to the other hence there is a single orbit $\begin{pmatrix}
1&
1
\end{pmatrix}$ whenever $x^2-y^2=0$. As we will see, this will no longer hold for the case $\hyperlink{Table:AK6+}{\mathsf{AK6}^+_\kappa}$ hence the need to further distinguish those two orbits in this case.}
Performing a change of basis 
 yields:
\[
\begin{tikzcd}[row sep = small,column sep = small,ampersand replacement=\&]
\hyperlink{Table:AK6h1}{(\mathsf{AK6},\alh_1)}\&\br{H}{\tK}=-\tK\&\br{H}{\tP}=\tP.
\end{tikzcd}
\]
Again, $\tK$ spans a non-trivial ideal hence the pair $\hyperlink{Table:AK6h1}{(\mathsf{AK6},\alh_1)}$ is non-effective, with associated effective Klein pair $\hyperlink{Table:Li3}{\mathsf{Li3}_0}$.
\ee
\item $\hyperlink{Table:AK6+}{\mathsf{AK6}^+_\kappa}$  

The automorphism group $\Aut(\hyperlink{Table:AK6+}{\mathsf{AK6}^+_\kappa},\ali)$ is given by elements of the form \eqref{equation:varphiaction} with $\alpha=\delta$, $\beta=\gamma$ and $\lambda=1$. We distinguish between the following orbits:
\be
\item $\begin{pmatrix}
1&
0
\end{pmatrix}$ whenever $x^2-y^2\neq0$ so that $\alh=\alh_0$. The corresponding effective Klein pair reads:
\bi
\item[]
\mybox{Ice}{
$\hypertarget{Eq:AK6+h0}{\hyperlink{Table:AK6+h0}{(\mathsf{AK6}^+_\kappa,\alh_0)}}$
}
{
\begin{tikzcd}[row sep = small,column sep = small,ampersand replacement=\&]
\br{\tK}{H}=\tP\&\br{H}{\tP}=-\tK\&\br{H}{M}=\kappa M\&\text{ where }\kappa>0
\end{tikzcd}
}
\ei
\item\hypertarget{Eq:AK6+h1}{}  $\begin{pmatrix}
1&
\epsilon
\end{pmatrix}$ whenever $x=\epsilon y$, with $\epsilon\in\pset{-1,1}$. Performing a change of basis 
yields:
\[
\begin{tikzcd}[row sep = small,column sep = small,ampersand replacement=\&]
\hyperlink{Table:AK6+h1}{(\mathsf{AK6}^+_\kappa,\alh_{1_\epsilon})}\&\br{H}{\tK}=\epsilon\tK\&\br{H}{\tP}=-\epsilon\tP\&\&\br{H}{M}=\kappa M
\end{tikzcd}
\]
The boost generators $\tK$ span a non-trivial ideal hence the pair $\hyperlink{Table:AK6+h1}{(\mathsf{AK6}^+_\kappa,\alh_{1_\epsilon})}$ is non-effective, with associated effective Klein pair $\hyperlink{Table:Li3}{\mathsf{Li3}_z}$, with $z>0$ (resp. $z<0$) for $\epsilon=-1$ (resp. $\epsilon=1$).
\ee

\item $\hyperlink{Table:AK7}{\mathsf{AK7}_{\alpha_+,\kappa}}$ 

The automorphism group $\Aut(\hyperlink{Table:AK7}{\mathsf{AK7}_{\alpha_+,\kappa}},\ali)$ is given by elements of the form \eqref{equation:varphiaction} with 
$\delta=\alpha+\alpha_+\beta$, $\gamma=-\beta$ and $\lambda=1$. The number of orbits varies according to the values of the parameter $\alpha_+>0$:
\be
\item $\alpha_+<2$: one orbit
\item $\alpha_+=2$: two orbits
\item $\alpha_+>2$: three orbits
\ee
Defining $\Delta_+:=x^2-\alpha_+xy+y^2$, the corresponding orbits are the following:
\be
\item $\begin{pmatrix}
1&
0
\end{pmatrix}$ whenever $\Delta_+\neq0$ so that $\alh=\alh_0$. This orbit is available for any value of $\alpha_+ $ and is unique whenever $0<\alpha_+<2$ [since the equation $\Delta_+=0$ does not possess real solutions then]. The Klein pair is effective and reads:
\mybox{Ice}{
$\hypertarget{Eq:AK7h0}{\hyperlink{Table:AK7h0}{(\mathsf{AK7}_{\alpha_+,\kappa},\alh_0)}}$
}
{
\begin{tikzcd}[row sep = small,column sep = small,ampersand replacement=\&]
\br{\tK}{H}=\tP\&\br{H}{\tP}=\tK+\alpha_+\tP\&\br{H}{M}=\kappa M\&\alpha_+>0
,\ 
\kappa\in\mathbb R
\end{tikzcd}
}
\item\hypertarget{Eq:AK7h1}{} Whenever $\alpha_+=2$, the orbit $\begin{pmatrix}
1&
1
\end{pmatrix}$ is available for every solution of $\Delta_+=0\Leftrightarrow x=y$, on top of the previous one. Performing a change of basis 
yields the non-effective Klein pair:
\[
\begin{tikzcd}[row sep = small,column sep = small,ampersand replacement=\&]
\hyperlink{Table:AK7h1}{(\mathsf{AK7}_{\alpha_+=2,\kappa},\alh_1)}\&\br{H}{\tK}= \tK\&\br{H}{\tP}=\tK+\tP\&\br{H}{M}=\kappa M\&\text{ where }\kappa\in\mathbb R\text{ is arbitrary.}
\end{tikzcd}
\]
The associated effective Klein pair is isomorphic to the lifshitz Klein pair $\hyperlink{Table:Li3}{\mathsf{Li3}_z}$, with $z=\kappa$.
\item\hypertarget{Eq:AK7h2}{} Finally, whenever $\alpha_+>2$, there are two additional orbits on top of the canonical one  $\begin{pmatrix}
1&
0
\end{pmatrix}$, 
 corresponding to the two branches of solutions of $\Delta_+=0$,
  and denoted $\begin{pmatrix}
\frac{1}{2}(\alpha_++\epsilon\sqrt{\alpha_+^2-4})&
1
\end{pmatrix}$, where $\epsilon\in\pset{-1,1}$. Performing a change of basis 
yields the two non-effective families of Klein pairs (where $\kappa\in\mathbb R$ is arbitrary):
\[
\begin{tikzcd}[row sep = small,column sep = small,ampersand replacement=\&]
\hspace{-1cm}\hyperlink{Table:AK7h2}{(\mathsf{AK7}_{\alpha_+>2,\kappa},\alh_{2_\epsilon})}\&\br{H}{\tK}=\frac{1}{2}\big(\alpha_+-\epsilon\sqrt{\alpha_+^2-4}\big)\tK\&\br{H}{\tP}=\frac{1}{2}\big(\alpha_++\epsilon\sqrt{\alpha_+^2-4}\big)\tP\&\br{H}{M}=\kappa M\, .
\end{tikzcd}
\]
The associated effective Klein pair is unique and isomorphic to the lifshitz Klein pair $\hyperlink{Table:Li3}{\mathsf{Li3}_z}$ [after rescaling of $H$ and setting $z=\frac{2\, \kappa}{\alpha_+-\epsilon\sqrt{\alpha_+^2-4}}$].
\ee

\item $\hyperlink{Table:AK8}{\mathsf{AK8}_{\alpha_-,\kappa}}$ 

There are three orbits for any value of $\alpha_->0$. 
Defining $\Delta_-:=x^2+\alpha_-xy-y^2$, the corresponding orbits are the following:
\be
\item $\begin{pmatrix}
1&
0
\end{pmatrix}$ whenever $\Delta_-\neq0$ so that $\alh=\alh_0$. 
The Klein pair is effective and reads:
\mybox{Ice}{
$\hypertarget{Eq:AK8h0}{\hyperlink{Table:AK8h0}{(\mathsf{AK8}_{\alpha_-,\kappa},\alh_0)}}$
}
{
\begin{tikzcd}[row sep = small,column sep = small,ampersand replacement=\&]
\br{\tK}{H}=-\tP\&\br{H}{\tP}=\tK+\alpha_-\tP\&\br{H}{M}=\kappa M\&\alpha_->0
,\ 
\kappa\in\mathbb R
\end{tikzcd}
}
\item\hypertarget{Eq:AK8h1}{}  $\begin{pmatrix}
-\frac{1}{2}(\alpha_-+\epsilon\sqrt{\alpha_-^2+4})&
1
\end{pmatrix}$
where $\epsilon\in\pset{-1,1}$ so that $\Delta_-=0$.
Performing a change of basis 
yields the two non-effective (families of) Klein pairs (where $\kappa\in\mathbb R$ is arbitrary):
\[
\begin{tikzcd}[row sep = small,column sep = small,ampersand replacement=\&]
\hspace{-1cm}\hyperlink{Table:AK8h1}{(\mathsf{AK8}_{\alpha_-,\kappa},\alh_{1_\epsilon})}\&\br{H}{\tK}=\frac{1}{2}\big(\alpha_--\epsilon\sqrt{\alpha_-^2+4}\big)\tK\&\br{H}{\tP}=\frac{1}{2}\big(\alpha_-+\epsilon\sqrt{\alpha_-^2+4}\big)\tP\&\br{H}{M}=\kappa M\, .
\end{tikzcd}
\]

The associated effective Klein pair is unique and isomorphic to the lifshitz Klein pair $\hyperlink{Table:Li3}{\mathsf{Li3}_z}$ [after rescaling of $H$ and $z=\frac{2\, \kappa}{\alpha_-+\epsilon\sqrt{\alpha_-^2+4}}$].

\ee

\ei
\paragraph{Second branch}
Recall that ambient aristotelian algebras belonging to the second branch of the above classification \UnskipRef{section:Ambient aristotelian algebras admitting a scalar ideal} are trivial extensions of the kinematical algebras from the second branch of the classification performed in \cite[Section 2.4]{Figueroa-OFarrill2017a}. The classification of the corresponding Klein pairs then follows \emph{mutatis mutandis} the classification of non-galilean kinematical Klein pairs performed in \cite[Section 3.1]{Figueroa-OFarrill2018}. Note that all Klein pairs are effective in this case.
\bi
\item $\hyperlink{Table:AK9}{\mathsf{AK9}}$

There is a unique orbit:
\bi
\item[]
\mybox{Ice}{
$\hypertarget{Eq:AK9h0}{\hyperlink{Table:AK9h0}{(\mathsf{AK9},\alh_0)}} \quad \text{carroll}\oplus\mathbb R$
}
{
\begin{tikzcd}[row sep = small,column sep = small,ampersand replacement=\&]
\br{\tK}{\tP}=H
\end{tikzcd}
}
\ei
\item $\hyperlink{Table:AK10}{\mathsf{AK10}}$  

There are two orbits:
\be
\item $y=0$
\mybox{Ice}{
$\hypertarget{Eq:AK10h0}{\hyperlink{Table:AK10h0}{(\mathsf{AK10},\alh_0)}} \quad \text{ds carroll}\oplus\mathbb R$
}
{
\begin{tikzcd}[row sep = small,column sep = small,ampersand replacement=\&]
\br{\tK}{\tP}=H\&\br{H}{\tP}=-\tK\&\br{\tP}{\tP}=-\tJ
\end{tikzcd}
}
\item $y\neq0$
\mybox{Ice}{
$\hypertarget{Eq:AK10h1}{\hyperlink{Table:AK10h1}{(\mathsf{AK10},\alh_1)}}\quad \text{euclidean}\oplus\mathbb R$

}
{
\begin{tikzcd}[row sep = small,column sep = small,ampersand replacement=\&]
\br{\tK}{\tP}=H\&\br{\tK}{H}=-\tP\&\br{\tK}{\tK}=-\tJ
\end{tikzcd}
}
\ee
\item $\hyperlink{Table:AK11}{\mathsf{AK11}}$

There are two orbits:
\be
\item $y=0$

\mybox{Ice}{
$\hypertarget{Eq:AK11h0}{\hyperlink{Table:AK11h0}{(\mathsf{AK11},\alh_0)}}\quad \text{ads carroll}\oplus\mathbb R$
}
{
\begin{tikzcd}[row sep = small,column sep = small,ampersand replacement=\&]
\br{\tK}{\tP}=H\&\br{H}{\tP}=\tK\&\br{\tP}{\tP}=\tJ
\end{tikzcd}
}
\item $y\neq0$
\mybox{Ice}{
$\hypertarget{Eq:AK11h1}{\hyperlink{Table:AK11h1}{(\mathsf{AK11},\alh_1)}}\quad \text{minkowski}\oplus\mathbb R$
}
{
\begin{tikzcd}[row sep = small,column sep = small,ampersand replacement=\&]
\br{\tK}{\tP}=H\&\br{\tK}{H}=\tP\&\br{\tK}{\tK}=\tJ
\end{tikzcd}
}
\ee
\item $\hyperlink{Table:AK12}{\mathsf{AK12}}$ 

There are three orbits according to the value of $\Delta:=x^2-y^2$.
\be
\item $\Delta>0$

\mybox{Ice}{
$\hypertarget{Eq:AK12h0}{\hyperlink{Table:AK12h0}{(\mathsf{AK12},\alh_0)}}\quad \text{hyperbolic}\oplus\mathbb R$
}
{
\begin{tikzcd}[row sep = small,column sep = small,ampersand replacement=\&]
\br{\tK}{\tP}=H\&\br{\tK}{H}=-\tP\&\br{H}{\tP}=\tK\&\br{\tP}{\tP}=\tJ\&\br{\tK}{\tK}=-\tJ
\end{tikzcd}
}
\item $\Delta=0$

\mybox{Ice}{
$\hypertarget{Eq:AK12h1}{\hyperlink{Table:AK12h1}{(\mathsf{AK12},\alh_1)}}\quad \text{carrollian light-cone}\oplus\mathbb R$
}
{
\begin{tikzcd}[row sep = small,column sep = small,ampersand replacement=\&]
\br{\tK}{\tP}=H+\tJ\&\br{H}{\tP}=-\tP\&\br{\tK}{H}=-\tK
\end{tikzcd}
}
\item $\Delta<0$
\mybox{Ice}{
$\hypertarget{Eq:AK12h2}{\hyperlink{Table:AK12h2}{(\mathsf{AK12},\alh_2)}}\quad \text{de sitter}\oplus\mathbb R$
}
{
\begin{tikzcd}[row sep = small,column sep = small,ampersand replacement=\&]
\br{\tK}{\tP}=H\&\br{\tK}{H}=\tP\&\br{H}{\tP}=-\tK\&\br{\tP}{\tP}=-\tJ\&\br{\tK}{\tK}=\tJ
\end{tikzcd}
}
\ee

\item $\hyperlink{Table:AK13}{\mathsf{AK13}}$

There is a single orbit:
\bi
\item[]
\mybox{Ice}{
$\hypertarget{Eq:AK13h0}{\hyperlink{Table:AK13h0}{(\mathsf{AK13},\alh_0)}}\quad \text{sphere}\oplus\mathbb R$
}
{
\begin{tikzcd}[row sep = small,column sep = small,ampersand replacement=\&]
\br{\tK}{\tP}=H\&\br{\tK}{H}=-\tP\&\br{H}{\tP}=-\tK\&\br{\tP}{\tP}=-\tJ\&\br{\tK}{\tK}=-\tJ
\end{tikzcd}
}
\ei
\item $\hyperlink{Table:AK14}{\mathsf{AK14}}$ 

There is a single orbit:
\bi
\item[]
\mybox{Ice}{
$\hypertarget{Eq:AK14h0}{\hyperlink{Table:AK14h0}{(\mathsf{AK14},\alh_0)}}\quad \text{anti de sitter}\oplus\mathbb R$
}
{
\begin{tikzcd}[row sep = small,column sep = small,ampersand replacement=\&]
\br{\tK}{\tP}=H\&\br{\tK}{H}=\tP\&\br{H}{\tP}=\tK\&\br{\tP}{\tP}=\tJ\&\br{\tK}{\tK}=\tJ
\end{tikzcd}
}
\ei
\ei
\paragraph{Third branch}

The treatment of the third branch is parallel to the one in the first branch, with the sole exception that all Klein pairs are effective thanks to the presence of the obstruction $\br{\tK}{\tP}=M$. 
\bi
\item $\hyperlink{Table:AK1b}{\mathsf{AK1b}}$

There is a unique orbit:
\bi
\item[]
\mybox{Ice}{
$\hypertarget{Eq:AK1bh0}{\hyperlink{Table:AK1bh0}{(\mathsf{AK1b},\alh_0)}} \quad \text{centrally extended static}$
}
{
\begin{tikzcd}[row sep = small,column sep = small,ampersand replacement=\&]
\br{\tK}{\tP}=M
\end{tikzcd}
}
\ei
\item $\hyperlink{Table:AK2b}{\mathsf{AK2b}}$

There is a unique orbit: 

\bi
\item[]
\mybox{Ice}{
$\hypertarget{Eq:AK2bh0}{\hyperlink{Table:AK2bh0}{(\mathsf{AK2b},\alh_0)}}$
}
{
\begin{tikzcd}[row sep = small,column sep = small,ampersand replacement=\&]
\br{\tK}{\tP}= M\&\br{H}{\tK}= \tK\&\br{H}{\tP}= \tP\&\br{H}{M}=2M
\end{tikzcd}
}
\ei
\item $\hyperlink{Table:AK3b}{\mathsf{AK3b}}$

There are two orbits:
\be
\item $y=0$

\mybox{Ice}{
$\hypertarget{Eq:AK3bh0}{\hyperlink{Table:AK3bh0}{(\mathsf{AK3b},\alh_0)}}$
}
{
\begin{tikzcd}[row sep = small,column sep = small,ampersand replacement=\&]
\br{\tK}{\tP}= M\&\br{H}{\tP}=\tK
\end{tikzcd}
}
\item $y\neq0$
\mybox{Ice}{
$\hypertarget{Eq:AK3bh1}{\hyperlink{Table:AK3bh1}{(\mathsf{AK3b},\alh_1)}}\quad \text{centrally extended galilei}$
}
{
\begin{tikzcd}[row sep = small,column sep = small,ampersand replacement=\&]
\br{\tK}{\tP}= M\&\br{\tK}{H}=\tP
\end{tikzcd}
}
\ee
\item $\hyperlink{Table:AK4b}{\mathsf{AK4b}}$

There are three orbits:
\be
\item $y=0$

\mybox{Ice}{
$\hypertarget{Eq:AK4bh0}{\hyperlink{Table:AK4bh0}{(\mathsf{AK4b},\alh_0)}}$
}
{
\begin{tikzcd}[row sep = small,column sep = small,ampersand replacement=\&]
\br{\tK}{\tP}= M\&\br{H}{\tP}=\tP\&\br{H}{M}=M
\end{tikzcd}
}
\item $y\neq0$, $x=0$ 

\mybox{Ice}{
$\hypertarget{Eq:AK4bh1}{\hyperlink{Table:AK4bh1}{(\mathsf{AK4b},\alh_1)}}$
}
{
\begin{tikzcd}[row sep = small,column sep = small,ampersand replacement=\&]
\br{\tK}{\tP}= M\&\br{H}{\tK}=\tK\&\br{H}{M}=M
\end{tikzcd}
}
\item $y\neq0$, $x\neq0$
\mybox{Ice}{
$\hypertarget{Eq:AK4bh2}{\hyperlink{Table:AK4bh2}{(\mathsf{AK4b},\alh_2)}}$
}
{
\begin{tikzcd}[row sep = small,column sep = small,ampersand replacement=\&]
\br{\tK}{\tP}= M\&\br{\tK}{H}=\tP\&\br{H}{\tP}=\tP\&\br{H}{M}=M
\end{tikzcd}
}
\ee
\item $\hyperlink{Table:AK5b}{\mathsf{AK5b}}$

There is a unique orbit:
\bi
\item[]
\mybox{Ice}{
$\hypertarget{Eq:AK5bh0}{\hyperlink{Table:AK5bh0}{(\mathsf{AK5b},\alh_0)}}\quad \text{centrally extended euclidean newton}$
}
{
\begin{tikzcd}[row sep = small,column sep = small,ampersand replacement=\&]
\br{\tK}{\tP}= M\&\br{\tK}{H}= \tP\&\br{H}{\tP}= \tK
\end{tikzcd}
}
\ei
\item $\hyperlink{Table:AK6b}{\mathsf{AK6b}}$

There are two orbits:
\be
\item $x^2-y^2\neq0$
\mybox{Ice}{
$\hypertarget{Eq:AK6bh0}{\hyperlink{Table:AK6bh0}{(\mathsf{AK6b},\alh_0)}}\quad \text{centrally extended lorentzian newton}$
}
{
\begin{tikzcd}[row sep = small,column sep = small,ampersand replacement=\&]
\br{\tK}{\tP}= M\&\br{\tK}{H}= \tP\&\br{H}{\tP}= -\tK
\end{tikzcd}
}
\item $x^2-y^2=0$ 

\mybox{Ice}{
$\hypertarget{Eq:AK6bh1}{\hyperlink{Table:AK6bh1}{(\mathsf{AK6b},\alh_1)}}$
}
{
\begin{tikzcd}[row sep = small,column sep = small,ampersand replacement=\&]
\br{\tK}{\tP}= M\&\br{H}{\tK}= -\tK\&\br{H}{\tP}= \tP
\end{tikzcd}
}
\ee

\item $\hyperlink{Table:AK7b}{\mathsf{AK7b}_{\alpha_+}}$

Defining $\Delta_+:=x^2-\alpha_+xy+y^2$, the corresponding orbits are the following:
\be
\item $\Delta_+\neq0$
\mybox{Ice}{
$\hypertarget{Eq:AK7bh0}{\hyperlink{Table:AK7bh0}{(\mathsf{AK7b}_{\alpha_+},\alh_0)}}$
}
{
\begin{tikzcd}[row sep = small,column sep = small,ampersand replacement=\&]
\br{\tK}{\tP}= M\&\br{\tK}{H}=\tP\&\br{H}{\tP}=\tK+\alpha_+\tP\&\br{H}{M}=\alpha_+ M
\end{tikzcd}
}
\item $\Delta_+=0$, $\alpha_+=2$ 

\mybox{Ice}{
$\hypertarget{Eq:AK7bh1}{\hyperlink{Table:AK7bh1}{(\mathsf{AK7b}_{\alpha_+=2},\alh_1)}}$
}
{
\begin{tikzcd}[row sep = small,column sep = small,ampersand replacement=\&]
\br{\tK}{\tP}= M\&\br{H}{\tK}= \tK\&\br{H}{\tP}= \tK+\tP\&\br{H}{M}= 2\, M
\end{tikzcd}
}
\item $\Delta_+=0$, $\alpha_+>2$ 

\mybox{Ice}{
$\hypertarget{Eq:AK7bh2}{\hyperlink{Table:AK7bh2}{(\mathsf{AK7b}_{\alpha_+>2},\alh_{2_{\epsilon}})}}$
}
{
{\footnotesize
\begin{tikzcd}[row sep = small,column sep = small,ampersand replacement=\&]
\br{\tK}{\tP}= M\&\br{H}{\tK}=\frac{1}{2}\big(\alpha_+-\epsilon\sqrt{\alpha_+^2-4}\big)\tK\&\br{H}{\tP}=\frac{1}{2}\big(\alpha_++\epsilon\sqrt{\alpha_+^2-4}\big)\tP\&\br{H}{M}=\alpha_+ M
\end{tikzcd}
}
}
\ee

\item $\hyperlink{Table:AK8b}{\mathsf{AK8b}_{\alpha_-}}$

Defining $\Delta_-:=x^2+\alpha_-xy-y^2$, the corresponding orbits are the following:
\be
\item $\Delta_-\neq0$
\mybox{Ice}{
$\hypertarget{Eq:AK8bh0}{\hyperlink{Table:AK8bh0}{(\mathsf{AK8b}_{\alpha_-},\alh_0)}}$
}
{
\begin{tikzcd}[row sep = small,column sep = small,ampersand replacement=\&]
\br{\tK}{\tP}= M\&\br{\tK}{H}=\tP\&\br{H}{\tP}=-\tK+\alpha_-\tP\&\br{H}{M}=\alpha_- M
\end{tikzcd}
}
\item $\Delta_-=0$ 

\mybox{Ice}{
$\hypertarget{Eq:AK8bh1}{\hyperlink{Table:AK8bh1}{(\mathsf{AK8b}_{\alpha_-},\alh_{1_\epsilon})}}$
}
{
{\footnotesize
\begin{tikzcd}[row sep = small,column sep = small,ampersand replacement=\&]
\br{\tK}{\tP}= M\&\br{H}{\tK}=\frac{1}{2}\big(\alpha_--\epsilon\sqrt{\alpha_-^2+4}\big)\tK\&\br{H}{\tP}=\frac{1}{2}\big(\alpha_-+\epsilon\sqrt{\alpha_-^2+4}\big)\tP\&\br{H}{M}=\alpha_- M
\end{tikzcd}
}
}
\ee
\ei

We collect in Table \ref{Table:Effective ambient aristotelian Klein pairs with scalar ideal} all effective Klein pairs associated with the ambient aristotelian algebras classified in Table \ref{Table:Ambient aristotelian algebras admitting a scalar ideal}. Furthermore, all effective lifshitzian Klein pairs associated with non-effective ambient aristotelian Klein pairs are collected in Table \ref{Table:Effective lifshitzian Klein pairs associated with ambient aristotelian algebras}.

\begin{table}[ht]
\begin{adjustwidth}{-0.5cm}{}
\resizebox{18cm}{!}{
\begin{tabular}{l|l|l|l|ll|ll|ll|ll}
\multicolumn{2}{c|}{\textbf{Klein pair}}&\multicolumn{1}{c|}{\textbf{Comments}}&\multicolumn{1}{c|}{$\br{\alh}{\alh}\subset\alh$}&\multicolumn{2}{c|}{$\br{\alh}{\alp}\subset\alp$}&\multicolumn{2}{c|}{$\br{\alp}{\alp}\subset\alh$ (Curvature)}&\multicolumn{2}{c|}{$\br{\alp}{\alp}\subset\alp$ (Torsion)}&\multicolumn{2}{c}{$\br{\alh}{\alp}\subset\alh$ (Non-reductivity)}\\\hline
\hyperlink{Table:AK3}{$\mathsf{AK3}$}&\hypertarget{Table:AK3h1}{\hyperlink{Eq:AK3h1}{$\alh_1$}}&$\text{galilei}\oplus\mathbb R$&&&$\br{\tK}{H}= \tP$&&&&&\\
\rowcolor{Gray}
\hyperlink{Table:AK3a}{$\mathsf{AK3a}$}&\hypertarget{Table:AK3ah1}{\hyperlink{Eq:AK3ah1}{$\alh_1$}}&$\text{galilei}\inplus\mathbb R$&&&$\br{\tK}{H}= \tP$&&&&$\br{H}{M}=M$&&\\
\hyperlink{Table:AK4}{$\mathsf{AK4}$}&\hypertarget{Table:AK4h2}{\hyperlink{Eq:AK4h2}{$\alh_2$}}&&&&$\br{\tK}{H}= \tP$&&&$\br{H}{\tP}= \tP$&&&\\
\rowcolor{Gray}
\hyperlink{Table:AK4+}{$\mathsf{AK4}^+_\kappa$}&\hypertarget{Table:AK4h2+}{\hyperlink{Eq:AK4h2}{$\alh_2$}}&$\kappa>0$&&&$\br{\tK}{H}= \tP$&&&$\br{H}{\tP}= \tP$&$\br{H}{M}=\kappa M$&&\\
\hyperlink{Table:AK4-}{$\mathsf{AK4}^-_\kappa$}&\hypertarget{Table:AK4h2-}{\hyperlink{Eq:AK4h2}{$\alh_2$}}&$\kappa>0$&&&$\br{\tK}{H}= \tP$&&&$\br{H}{\tP}= \tP$&$\br{H}{M}=-\kappa M$&&\\
\rowcolor{Gray}
\hyperlink{Table:AK5}{$\mathsf{AK5}$}&\hypertarget{Table:AK5h0}{\hyperlink{Eq:AK5h0}{$\alh_0$}}&$\text{euclidean newton}\oplus\mathbb R$&&&$\br{\tK}{H}= \tP$&$\br{H}{\tP}= \tK$&&&&&\\
\hyperlink{Table:AK5+}{$\mathsf{AK5}^+_\kappa$}&\hypertarget{Table:AK5+h0}{\hyperlink{Eq:AK5+h0}{$\alh_0$}}&$\kappa>0$&&&$\br{\tK}{H}= \tP$&$\br{H}{\tP}= \tK$&&&$\br{H}{M}=\kappa M$&&\\
\rowcolor{Gray}
\hyperlink{Table:AK6}{$\mathsf{AK6}$}&\hypertarget{Table:AK6h0}{\hyperlink{Eq:AK6h0}{$\alh_0$}}&$\text{lorentzian newton}\oplus\mathbb R$&&&$\br{\tK}{H}= \tP$&$\br{H}{\tP}= -\tK$&&&&&\\
\hyperlink{Table:AK6+}{$\mathsf{AK6}^+_\kappa$}&\hypertarget{Table:AK6+h0}{\hyperlink{Eq:AK6+h0}{$\alh_0$}}&$\kappa>0$&&&$\br{\tK}{H}= \tP$&$\br{H}{\tP}= -\tK$&&&$\br{H}{M}=\kappa M$&&\\
\rowcolor{Gray}
\hyperlink{Table:AK7}{$\mathsf{AK7}_{\alpha_+}$}&\hypertarget{Table:AK7h0}{\hyperlink{Eq:AK7h0}{$\alh_0$}}&$\alpha_+>0$&&&$\br{\tK}{H}= \tP$&$\br{H}{\tP}= \tK$&&$\br{H}{\tP}= \alpha_+\tP$&&&\\
\hyperlink{Table:AK7+}{$\mathsf{AK7}^+_{\alpha_+,\kappa}$}&\hypertarget{Table:AK7+h0}{\hyperlink{Eq:AK7h0}{$\alh_0$}}&$\alpha_+>0\qq\kappa>0$&&&$\br{\tK}{H}= \tP$&$\br{H}{\tP}= \tK$&&$\br{H}{\tP}= \alpha_+\tP$&$\br{H}{M}=\kappa M$&&\\
\rowcolor{Gray}
\hyperlink{Table:AK7-}{$\mathsf{AK7}^-_{\alpha_+,\kappa}$}&\hypertarget{Table:AK7-h0}{\hyperlink{Eq:AK7h0}{$\alh_0$}}&$\alpha_+>0\qq\kappa>0$&&&$\br{\tK}{H}= \tP$&$\br{H}{\tP}= \tK$&&$\br{H}{\tP}= \alpha_+\tP$&$\br{H}{M}=-\kappa M$&&\\
\hyperlink{Table:AK8}{$\mathsf{AK8}_{\alpha_-}$}&\hypertarget{Table:AK8h0}{\hyperlink{Eq:AK8h0}{$\alh_0$}}&$\alpha_->0$&&&$\br{\tK}{H}= \tP$&$\br{H}{\tP}= -\tK$&&$\br{H}{\tP}= \alpha_-\tP$&&&\\
\rowcolor{Gray}
\hyperlink{Table:AK8+}{$\mathsf{AK8}^+_{\alpha_-,\kappa}$}&\hypertarget{Table:AK8+h0}{\hyperlink{Eq:AK8h0}{$\alh_0$}}&$\alpha_->0\qq\kappa>0$&&&$\br{\tK}{H}= \tP$&$\br{H}{\tP}= -\tK$&&$\br{H}{\tP}= \alpha_-\tP$&$\br{H}{M}=\kappa M$&&\\
\hyperlink{Table:AK8-}{$\mathsf{AK8}^-_{\alpha_-,\kappa}$}&\hypertarget{Table:AK8-h0}{\hyperlink{Eq:AK8h0}{$\alh_0$}}&$\alpha_->0\qq\kappa>0$&&&$\br{\tK}{H}= \tP$&$\br{H}{\tP}= -\tK$&&$\br{H}{\tP}= \alpha_-\tP$&$\br{H}{M}=-\kappa M$&&\\
\rowcolor{Gray}
\hline
\hyperlink{Table:AK9}{$\mathsf{AK9}$}&\hypertarget{Table:AK9h0}{\hyperlink{Eq:AK9h0}{$\alh_0$}}&$\text{carroll}\oplus\mathbb R$&&$\br{\tK}{\tP}= H$&&&&&&&\\
\hyperlink{Table:AK10}{$\mathsf{AK10}$}&\hypertarget{Table:AK10h0}{\hyperlink{Eq:AK10h0}{$\alh_0$}}&$\text{ds carroll}\oplus\mathbb R$&&$\br{\tK}{\tP}= H$&&$\br{H}{\tP}= -\tK$&$\br{\tP}{\tP}=-\tJ$&&&&\\
\rowcolor{Gray}
\hyperlink{Table:AK10}{$\mathsf{AK10}$}&\hypertarget{Table:AK10h1}{\hyperlink{Eq:AK10h1}{$\alh_1$}}&$\text{euclidean}\oplus\mathbb R$&$\br{\tK}{\tK}=-\tJ$&$\br{\tK}{\tP}= H$&$\br{\tK}{H}= -\tP$&&&&&&\\
\hyperlink{Table:AK11}{$\mathsf{AK11}$}&\hypertarget{Table:AK11h0}{\hyperlink{Eq:AK11h0}{$\alh_0$}}&$\text{ads carroll}\oplus\mathbb R$&&$\br{\tK}{\tP}= H$&&$\br{H}{\tP}= \tK$&$\br{\tP}{\tP}=\tJ$&&&&\\
\rowcolor{Gray}
\hyperlink{Table:AK11}{$\mathsf{AK11}$}&\hypertarget{Table:AK11h1}{\hyperlink{Eq:AK11h1}{$\alh_1$}}&$\text{minkowski}\oplus\mathbb R$&$\br{\tK}{\tK}=\tJ$&$\br{\tK}{\tP}= H$&$\br{\tK}{H}= \tP$&&&&&&\\
\hyperlink{Table:AK12}{$\mathsf{AK12}$}&\hypertarget{Table:AK12h0}{\hyperlink{Eq:AK12h0}{$\alh_0$}}&$\text{hyperbolic}\oplus\mathbb R$&$\br{\tK}{\tK}=-\tJ$&$\br{\tK}{\tP}= H$&$\br{\tK}{H}= -\tP$&$\br{H}{\tP}=\tK$&$\br{\tP}{\tP}=\tJ$&&&&\\
\rowcolor{Gray}
\hyperlink{Table:AK12}{$\mathsf{AK12}$}&\hypertarget{Table:AK12h1}{\hyperlink{Eq:AK12h1}{$\alh_1$}}&$\text{carroll light cone}\oplus\mathbb R$&&$\br{\tK}{\tP}= H$&&&&$\br{H}{\tP}= -\tP$&&$\br{H}{\tK}=\tK$&$\br{\tK}{\tP}= \tJ$\\
\hyperlink{Table:AK12}{$\mathsf{AK12}$}&\hypertarget{Table:AK12h2}{\hyperlink{Eq:AK12h2}{$\alh_2$}}&$\text{de sitter}\oplus\mathbb R$&$\br{\tK}{\tK}=\tJ$&$\br{\tK}{\tP}= H$&$\br{\tK}{H}= \tP$&$\br{H}{\tP}=-\tK$&$\br{\tP}{\tP}=-\tJ$&&&&\\
\rowcolor{Gray}
\hyperlink{Table:AK13}{$\mathsf{AK13}$}&\hypertarget{Table:AK13h0}{\hyperlink{Eq:AK13h0}{$\alh_0$}}&$\text{sphere}\oplus\mathbb R$&$\br{\tK}{\tK}=-\tJ$&$\br{\tK}{\tP}= H$&$\br{\tK}{H}= -\tP$&$\br{H}{\tP}=-\tK$&$\br{\tP}{\tP}=-\tJ$&&&&\\
\hyperlink{Table:AK14}{$\mathsf{AK14}$}&\hypertarget{Table:AK14h0}{\hyperlink{Eq:AK14h0}{$\alh_0$}}&$\text{anti de sitter}\oplus\mathbb R$&$\br{\tK}{\tK}=\tJ$&$\br{\tK}{\tP}= H$&$\br{\tK}{H}= \tP$&$\br{H}{\tP}=\tK$&$\br{\tP}{\tP}=\tJ$&&&&\\
\rowcolor{Gray}
\hline
\hyperlink{Table:AK1b}{$\mathsf{AK1b}$}&\hypertarget{Table:AK1bh0}{\hyperlink{Eq:AK1bh0}{$\alh_0$}}&centrally extended static&&$\br{\tK}{\tP}=M$&&&&&&&\\
\hyperlink{Table:AK2b}{$\mathsf{AK2b}$}&\hypertarget{Table:AK2bh0}{\hyperlink{Eq:AK2bh0}{$\alh_0$}}&&&$\br{\tK}{\tP}=M$&&&&$\br{H}{\tP}=\tP$&$\br{H}{M}=2M$&$\br{H}{\tK}=\tK$&\\
\rowcolor{Gray}
\hyperlink{Table:AK3b}{$\mathsf{AK3b}$}&\hypertarget{Table:AK3bh0}{\hyperlink{Eq:AK3bh0}{$\alh_0$}}&&&$\br{\tK}{\tP}=M$&&$\br{H}{\tP}=\tK$&&&&&\\
\hyperlink{Table:AK3b}{$\mathsf{AK3b}$}&\hypertarget{Table:AK3bh1}{\hyperlink{Eq:AK3bh1}{$\alh_1$}}&centrally extended galilei&&$\br{\tK}{\tP}=M$&$\br{\tK}{H}= \tP$&&&&&&\\
\rowcolor{Gray}
\hyperlink{Table:AK4b}{$\mathsf{AK4b}$}&\hypertarget{Table:AK4bh0}{\hyperlink{Eq:AK4bh0}{$\alh_0$}}&&&$\br{\tK}{\tP}=M$&&&&$\br{H}{\tP}= \tP$&$\br{H}{M}=M$&&\\
\hyperlink{Table:AK4b}{$\mathsf{AK4b}$}&\hypertarget{Table:AK4bh1}{\hyperlink{Eq:AK4bh1}{$\alh_1$}}&&&$\br{\tK}{\tP}=M$&&&&&$\br{H}{M}=M$&$\br{H}{\tK}=\tK$&\\
\rowcolor{Gray}
\hyperlink{Table:AK4b}{$\mathsf{AK4b}$}&\hypertarget{Table:AK4bh2}{\hyperlink{Eq:AK4bh2}{$\alh_2$}}&&&$\br{\tK}{\tP}=M$&$\br{\tK}{H}= \tP$&&&$\br{H}{\tP}= \tP$&$\br{H}{M}=M$&&\\
\hyperlink{Table:AK5b}{$\mathsf{AK5b}$}&\hypertarget{Table:AK5bh0}{\hyperlink{Eq:AK5bh0}{$\alh_0$}}&centrally extended euclidean newton&&$\br{\tK}{\tP}=M$&$\br{\tK}{H}= \tP$&$\br{H}{\tP}= \tK$&&&&&\\
\rowcolor{Gray}
\hyperlink{Table:AK6b}{$\mathsf{AK6b}$}&\hypertarget{Table:AK6bh0}{\hyperlink{Eq:AK6bh0}{$\alh_0$}}&centrally extended lorentzian newton&&$\br{\tK}{\tP}=M$&$\br{\tK}{H}= \tP$&$\br{H}{\tP}= -\tK$&&&&&\\
\hyperlink{Table:AK6b}{$\mathsf{AK6b}$}&\hypertarget{Table:AK6bh1}{\hyperlink{Eq:AK6bh1}{$\alh_1$}}&&&$\br{\tK}{\tP}=M$&&&&$\br{H}{\tP}=-\tP$&&$\br{H}{\tK}=\tK$&\\
\rowcolor{Gray}
\hyperlink{Table:AK7b}{$\mathsf{AK7b}_{\alpha_+}$}&\hypertarget{Table:AK7bh0}{\hyperlink{Eq:AK7bh0}{$\alh_0$}}&$\alpha_+>0$&&$\br{\tK}{\tP}=M$&$\br{\tK}{H}= \tP$&$\br{H}{\tP}= \tK$&&$\br{H}{\tP}= \alpha_+\tP$&$\br{H}{M}=\alpha_+ M$&&\\
\hyperlink{Table:AK7b}{$\mathsf{AK7b}_{\alpha_+=2}$}&\hypertarget{Table:AK7bh1}{\hyperlink{Eq:AK7bh1}{$\alh_1$}}&$\alpha_+=2$&&$\br{\tK}{\tP}=M$&&$\br{H}{\tP}= \tK$&&$\br{H}{\tP}=\tP$&$\br{H}{M}=2 M$&$\br{H}{\tK}=\tK$&\\
\rowcolor{Gray}
\hyperlink{Table:AK7b}{$\mathsf{AK7b}_{\alpha_+>2}$}&\hypertarget{Table:AK7bh2}{\hyperlink{Eq:AK7bh2}{$\alh_{2_\epsilon}$}}&$\alpha_+>2\qq\epsilon\in\pset{-1,1}$&&$\br{\tK}{\tP}=M$&&&&$\br{H}{\tP}=\frac{1}{2}\big(\alpha_++\epsilon\sqrt{\alpha_+^2-4}\big)\tP$&$\br{H}{M}=\alpha_+ M$&\multicolumn{2}{c}{$\br{H}{\tK}=\frac{1}{2}\big(\alpha_+-\epsilon\sqrt{\alpha_+^2-4}\big)\tK$}\\
\hyperlink{Table:AK8b}{$\mathsf{AK8b}_{\alpha_-}$}&\hypertarget{Table:AK8bh0}{\hyperlink{Eq:AK8bh0}{$\alh_0$}}&$\alpha_->0$&&$\br{\tK}{\tP}=M$&$\br{\tK}{H}= \tP$&$\br{H}{\tP}= -\tK$&&$\br{H}{\tP}= \alpha_-\tP$&$\br{H}{M}=\alpha_- M$&&\\
\rowcolor{Gray}
\hyperlink{Table:AK8b}{$\mathsf{AK8b}_{\alpha_-}$}&\hypertarget{Table:AK8bh1}{\hyperlink{Eq:AK8bh1}{$\alh_{1_\epsilon}$}}&$\alpha_->0\qq\epsilon\in\pset{-1,1}$&&$\br{\tK}{\tP}=M$&&&&$\br{H}{\tP}=\frac{1}{2}\big(\alpha_-+\epsilon\sqrt{\alpha_-^2+4}\big)\tP$&$\br{H}{M}=\alpha_- M$&\multicolumn{2}{c}{$\br{H}{\tK}=\frac{1}{2}\big(\alpha_--\epsilon\sqrt{\alpha_-^2+4}\big)\tK$}
\end{tabular}
  }
      \end{adjustwidth}
      \caption{Universal effective ambient aristotelian Klein pairs with scalar ideal}
      \label{Table:Effective ambient aristotelian Klein pairs with scalar ideal}
\end{table}

\begin{table}[ht]
\centering
\resizebox{8cm}{!}{
\begin{tabular}{l|l|lll}
\multicolumn{2}{c|}{\textbf{Klein pair}} &\multicolumn{3}{c}{\textbf{Lifshitzian Klein pair}}  \\ \hline
\hyperlink{Table:AK1}{$\mathsf{AK1}$}&\hypertarget{Table:AK1h0}{\hyperlink{Eq:AK1h0}{$\alh_0$}}& \hyperlink{Table:Li1}{$\mathsf{Li1}$} && \\ 
\rowcolor{Gray}
\hyperlink{Table:AK1a}{$\mathsf{AK1a}$}&\hypertarget{Table:AK1ah0}{\hyperlink{Eq:AK1ah0}{$\alh_0$}}& \hyperlink{Table:Li2}{$\mathsf{Li2}$} && \\ 
\hyperlink{Table:AK2}{$\mathsf{AK2}$}&\hypertarget{Table:AK2h0}{\hyperlink{Eq:AK2h0}{$\alh_0$}}& \hyperlink{Table:Li3}{$\mathsf{Li3}_z$} && $z=0$\\ 
\rowcolor{Gray}
\hyperlink{Table:AK2+}{$\mathsf{AK2}^+_\kappa$}&\hypertarget{Table:AK2+h0}{\hyperlink{Eq:AK2h0}{$\alh_0$}}& \hyperlink{Table:Li3}{$\mathsf{Li3}_z$} && $z=\kappa$\\ 
\hyperlink{Table:AK2-}{$\mathsf{AK2}^-_\kappa$}&\hypertarget{Table:AK2-h0}{\hyperlink{Eq:AK2h0}{$\alh_0$}}& \hyperlink{Table:Li3}{$\mathsf{Li3}_z$} && $z=-\kappa$\\ 
\rowcolor{Gray}
\hyperlink{Table:AK3}{$\mathsf{AK3}$}&\hypertarget{Table:AK3h0}{\hyperlink{Eq:AK3h0}{$\alh_0$}}& \hyperlink{Table:Li1}{$\mathsf{Li1}$} &&\\ 
\hyperlink{Table:AK3a}{$\mathsf{AK3a}$}&\hypertarget{Table:AK3ah0}{\hyperlink{Eq:AK3ah0}{$\alh_0$}}& \hyperlink{Table:Li2}{$\mathsf{Li2}$} &&\\ 
\rowcolor{Gray}
\hyperlink{Table:AK4}{$\mathsf{AK4}$}&\hypertarget{Table:AK4h0}{\hyperlink{Eq:AK4h0}{$\alh_0$}}& \hyperlink{Table:Li3}{$\mathsf{Li3}_z$} &&$z=0$\\ 
\hyperlink{Table:AK4+}{$\mathsf{AK4}^+_\kappa$}&\hypertarget{Table:AK4+h0}{\hyperlink{Eq:AK4h0}{$\alh_0$}}& \hyperlink{Table:Li3}{$\mathsf{Li3}_z$} &&$z=\kappa$\\ 
\rowcolor{Gray}
\hyperlink{Table:AK4-}{$\mathsf{AK4}^-_\kappa$}&\hypertarget{Table:AK4-h0}{\hyperlink{Eq:AK4h0}{$\alh_0$}}& \hyperlink{Table:Li3}{$\mathsf{Li3}_z$} &&$z=-\kappa$\\ 
\hyperlink{Table:AK4}{$\mathsf{AK4}$}&\hypertarget{Table:AK4h1}{\hyperlink{Eq:AK4h1}{$\alh_1$}}& \hyperlink{Table:Li3}{$\mathsf{Li1}$} &&\\ 
\rowcolor{Gray}
\hyperlink{Table:AK4+}{$\mathsf{AK4}^+_\kappa$}&\hypertarget{Table:AK4+h1}{\hyperlink{Eq:AK4h1}{$\alh_1$}}& \hyperlink{Table:Li2}{$\mathsf{Li2}$} &&\\ 
\hyperlink{Table:AK4-}{$\mathsf{AK4}^-_\kappa$}&\hypertarget{Table:AK4-h1}{\hyperlink{Eq:AK4h1}{$\alh_1$}}& \hyperlink{Table:Li2}{$\mathsf{Li2}$} &&\\ 
\rowcolor{Gray}
\hyperlink{Table:AK6}{$\mathsf{AK6}$}&\hypertarget{Table:AK6h1}{\hyperlink{Eq:AK6h1}{$\alh_1$}}& \hyperlink{Table:Li3}{$\mathsf{Li3}_z$} &&$z=0$\\ 
\hyperlink{Table:AK6+}{$\mathsf{AK6}^+_\kappa$}&\hypertarget{Table:AK6+h1}{\hyperlink{Eq:AK6+h1}{$\alh_{1_\epsilon}$}}& \hyperlink{Table:Li3}{$\mathsf{Li3}_z$} &&$z>0$ for $\epsilon=-1$ and $z<0$ for $\epsilon=1$\\ 
\rowcolor{Gray}
\hyperlink{Table:AK7}{$\mathsf{AK7}_{\alpha_+=2,\kappa}$}&\hypertarget{Table:AK7h1}{\hyperlink{Eq:AK7h1}{$\alh_1$}}& \hyperlink{Table:Li3}{$\mathsf{Li3}_z$} &&$z=\kappa$\\ 
\hyperlink{Table:AK7}{$\mathsf{AK7}_{\alpha_+>2,\kappa}$}&\hypertarget{Table:AK7h2}{\hyperlink{Eq:AK7h2}{$\alh_{2_\epsilon}$}}& \hyperlink{Table:Li3}{$\mathsf{Li3}_z$} &&$z=\frac{2\, \kappa}{\alpha_+-\epsilon\sqrt{\alpha_+^2-4}}$\\ 
\rowcolor{Gray}
\hyperlink{Table:AK8}{$\mathsf{AK8}_{\alpha_-,\kappa}$}&\hypertarget{Table:AK8h1}{\hyperlink{Eq:AK8h1}{$\alh_{1_\epsilon}$}}& \hyperlink{Table:Li3}{$\mathsf{Li3}_z$} &&$z=\frac{2\, \kappa}{\alpha_-+\epsilon\sqrt{\alpha_-^2+4}}$
\end{tabular}
}
\caption{Effective lifshitzian Klein pairs associated with non-effective ambient aristotelian Klein pairs}
\label{Table:Effective lifshitzian Klein pairs associated with ambient aristotelian algebras}
\end{table}

\subsection{Projectable ambient triplets}
\label{section:Projectable ambient triplets}

Having classified effective ambient aristotelian Klein pairs [with scalar ideal], we now proceed to classify projectable ambient triplets \UnskipRef{Definition:Projectable ambient triplets} by examining each effective Klein pair collected in Table \ref{Table:Effective ambient aristotelian Klein pairs with scalar ideal} and checking the effectiveness of the quotient Klein pair $(\alg/\ali,\alh)$. 
This step requires us to verify that the following obstructions [as described in footnote \ref{footnote:effectiveness}] are not simultaneously vanishing:
\[
 \begin{tikzcd}[row sep = small,column sep = small,ampersand replacement=\&]
\br{\tK}{\tK}\sim \tJ \& \br{\tK}{H}\sim \tP\&\br{\tK}{\tP}\sim H\&\br{\tK}{\tP}\sim \tJ\, .
\end{tikzcd}
\]
Note that all effective Klein pairs originating from the first and second branches successfully pass the effectiveness test at the quotient level. This can be traced back to the fact that their characteristic obstruction---\ie $\br{\tK}{H}=\tP$ and $\br{\tK}{\tP}=H$ respectively---pass well to the quotient. Contradistinctly, the characteristic obstruction of the third branch is given by $\br{\tK}{\tP}=M$ which does not survive the quotient along $\ali=\Span M$. Therefore, out of the 15 families of effective Klein pairs originating from the third branch, only 6 pass the effectiveness test at the quotient level, namely the ones with additional non-vanishing obstruction $\br{\tK}{H}=\tP$.\footnote{As we will see, these form the bargmannian family of projectable ambient triplets.} 
\medskip

The last remaining step consists in organising the obtained projectable ambient triplets according to the type of invariant structures living on the quotient space $\alg/\alh$. Letting $\alp=\Span M\oplus H\oplus \tP$ be a supplementary of $\alh$ [\ie $\alg=\alh\oplus\alp$ as a vector space\footnote{This decomposition is one of $\alh$-modules for all projectable ambient triplets, with the exception of the non-reductive trivial extension \hyperlink{Table:AK12h1}{$(\mathsf{AK12},\alh_1)$} of the carroll light cone \hyperlink{Table:E10}{$\mathsf{E10}$}.}], the set of commutators $\br{\alh}{\alp}\subset\alp$ determine what (possibly degenerate) metric structure on $\alp$ is $\ad_\alh$-invariant. Hence, to a different set of obstructions to non-effectiveness will correspond different invariant metric structures. 
\medskip

Let us now list the $\ad_\alh$-invariant structures living on the quotient space $\alp\simeq\alg/\alh$:
\bi
\item We start by noting that all Klein pairs listed in Table \ref{Table:Effective ambient aristotelian Klein pairs with scalar ideal} satisfy $\br{\alh}{M}=0$, hence they all admit $\bs{\xi}=M$ as an $\ad_\alh$-invariant.
\item The projectable Klein pairs originating from the first and third branches admit moreover the \emph{leibnizian} metric structure \UnskipRef{Definition:leibnizian metric structures on Klein pairs} defined in eq.\eqref{equation:leibniz metric structure on leibniz Klein pair}, supplementing the vector $\bs{\xi}$ with the $\ad_\alh$-invariant linear form $\bs{\psi}=-H^*$ and the $\ad_\alh$-invariant bilinear form $\bs{\gamma}=\tP^*\otimes \tP^*$ defined on $\Ker\bs{\psi}$.
\item Klein pairs originating from the first and second branches further admit the $\ad_\alh$-invariant linear form $\bs{{A}}=M^*$. 
\item Combining the above points, the leibnizian structure living on Klein pairs originating from the first branch can be upgraded to a \emph{$\mathsf G$-ambient} structure $(\bs{\xi},\bs{\psi},\bs{A},\bs{\gamma})$  \UnskipRef{Definition:Ambient aristotelian metric structures on Klein pairs}. 
\item Klein pairs originating from the second branch can be further partitioned into three categories, according to the type of $\ad_\alh$-invariant structures they admit, on top of the pair $(\bs{\xi},\bs{A})$\footnote{The notation `$\textnormal{kinematical structure}\oplus\mathbb R$' acts as a reminder that the corresponding ambient algebras are trivial extensions of the corresponding kinematical algebras. Let us emphasise that the corresponding metric structures are {\it not} kinematical \UnskipRef{Definition:Metric structures on Klein pairs} on the whole of the ambient space $\alp\simeq\gh$, but only on the quotient $\alp/\ali$ thereof. For example, the metric structure $\textnormal{riemann}\oplus\mathbb R$ is not riemannian on $\alp$ [since $\bs{g}$ admits a non-trivial radical spanned by $\bs{\xi}$], but does induce a riemannian structure on the associated quotient space $\alp/\ali$.}:
\begin{description}
\item[$\textnormal{carrollian}\oplus\mathbb R$]: $\bs{N}=H$ and $\bs{\gamma}=\tP^*\otimes \tP^*$
\item[$\textnormal{riemannian}\oplus\mathbb R$]: $\bs{g}=H^*\otimes H^*+\tP^*\otimes \tP^*$
\item[$\textnormal{lorentzian}\oplus\mathbb R$]: $\bs{g}=-H^*\otimes H^*+\tP^*\otimes \tP^*$
\end{description}

\item Projectable Klein pairs originating from the third branch are characterised by the extra non-vanishing obstruction $\br{\tK}{\tP}=M$ and as such fail to preserve the linear form $\bs{{A}}=M^*$. Rather, they are endowed with an invariant non-degenerate bilinear form $\bs{g}=-M^*\otimes H^*-H^*\otimes M^*+\tP^*\otimes \tP^*$ defined on the whole of $\alp$, thus upgrading the above leibnizian structure to a \emph{bargmannian} structure $(\bs{\xi},\bs{g})$  \UnskipRef{Definition:Ambient aristotelian metric structures on Klein pairs}. 
\ei
We collect the obtained classes of projectable ambient triplets in Table \ref{Table:Projectable ambient triplets}. 
\begin{Remark}
\label{Remark:Projectable ambient triplets table}
\hfill
\bi
\item Non-galilean effective kinematical Klein pairs \UnskipRef{Table:Effective kinematical Klein pairs} are seen to admit a unique (and trivial) ambient lift.
\item Contradistinctly, effective galilean Klein pairs are shown to admit ambient lifts into two distinct classes, namely the $\mathsf G$-ambient and bargmannian classes. 
\item The bargmannian projectable triplets classified in Table \ref{Table:Projectable ambient triplets} coincide with the ones exhibited and studied in \cite{FigueroaOFarrill2022e}.\footnote{Our choice of parameterisation for bargmannian Klein pairs differs from the one used in \cite[Table 2]{FigueroaOFarrill2022e}. The Behistun Inscription between the two characterisations can be found below: 
\medskip

\begin{center}
\begin{tabular}{l|l}
Table 2 of \cite{FigueroaOFarrill2022e}&Table \ref{Table:Projectable ambient triplets}\\\hline
$\mathsf{PG}$& \hyperlink{Table:AE1b}{$\mathsf{AE1b}$}\\
$\mathsf{PdSG} = \mathsf{PdSG}_{\gamma=-1}$\cellcolor{Gray}&\cellcolor{Gray}\hyperlink{Table:AE4b}{$\mathsf{AE4b}$}\\
$\mathsf{PdSG}_{\gamma\in(-1,0)}$&\hyperlink{Table:AE6b}{$\mathsf{AE6b}_{\alpha_-}$} with $\alpha_->0$ where $\alpha_-=\frac{\gamma+1}{\sqrt{-\gamma}}$\\
$\mathsf{PdSG}_{\gamma=0}$\cellcolor{Gray}&\cellcolor{Gray}\hyperlink{Table:AE2b}{$\mathsf{AE2b}$}\\
$\mathsf{PdSG}_{\gamma\in(0,1)}$&\hyperlink{Table:AE5b}{$\mathsf{AE5b}_{\alpha_+}$} with $\alpha_+>2$ where $\alpha_+=\frac{\gamma+1}{\sqrt{\gamma}}$\\
$\mathsf{PdSG}_{\gamma=1}$\cellcolor{Gray}&\cellcolor{Gray}\hyperlink{Table:AE5b}{$\mathsf{AE5b}_{\alpha_+}$} with $\alpha_+=2$\\
$\mathsf{PAdSG} = \mathsf{PAdSG}_{\chi=0}$&\hyperlink{Table:AE3b}{$\mathsf{AE3b}$}\\
$\mathsf{PAdSG}_{\chi>0}$\cellcolor{Gray}&\cellcolor{Gray}\hyperlink{Table:AE5b}{$\mathsf{AE5b}_{\alpha_+}$} with $0<\alpha_+<2$ where $\alpha_+=\frac{2\chi}{\sqrt{1+\chi^2}}$
  \end{tabular}
  \end{center}
} This class is distinguished as the only ambient class admitting an invariant non-degenerate bilinear form on the whole ambient space $\alp\simeq\gh$. 
\item While every effective galilean Klein pair admits a \emph{unique} bargmannian lift, embedding galilean homogeneous spaces into $\mathsf{G}$-ambient Klein pairs leaves greater leeway, (generically) allowing a one-parameter family of ambient lifts\footnote{The sole exception to this generic trend is the galilei Klein pair \hyperlink{Table:E1}{$\mathsf{E1}$}, for which the rescaling freedom of the generator $H$ can be used to kill the one-parameter freedom in the ambient torsion [\cf \hyperlink{Eq:AK3a}{$\mathsf{AK3a}$}]. } whose arbitrariness is encoded into the ambient torsion. 
\item Possible lifts of effective galilean Klein pairs are summarised in Table \ref{Table:Possible lifts of effective galilean Klein pairs}.
\item For future uses, let us abstract away the bargmannian Klein pairs in Table \ref{Table:Projectable ambient triplets} as  the 2-parameter family $\mathsf{AEb_{\epsilon,\omega}}$ with non-trivial commutation relations:
 \bea
  \label{equation:algebra bargmann}
 \begin{tikzcd}[row sep = small,column sep = small,ampersand replacement=\&]
\mathsf{AEb_{\epsilon,\omega}}\&\br{\tK}{\tP}= M\&\br{\tK}{H}=\tP\&\br{H}{\tP}=\epsilon\, \tK+\omega\, \tP\&\br{H}{M}=\omega\, M
\end{tikzcd}
 \eea
where the two free parameters $\epsilon\in\mathbb R$ and $\omega\in\mathbb R$ encode curvature and torsion, respectively. 

Proceeding similarly for $\mathsf{G}$-ambient Klein pairs leads us to introduce the 3-parameter family $\mathsf{AE_{\epsilon,\omega,\lambda}}$ with non-trivial commutation relations:
 \bea
 \label{equation:algebra G-ambient}
 \hspace{-2.65cm}
 \begin{tikzcd}[row sep = small,column sep = small,ampersand replacement=\&]
\mathsf{AE_{\epsilon,\omega,\lambda}}\&\br{\tK}{H}=\tP\&\br{H}{\tP}=\epsilon\, \tK+\omega\, \tP\&\br{H}{M}=\lambda\, M
\end{tikzcd}
 \eea
where the extra parameter $\lambda\in\mathbb R$ encodes the mass part of the torsion.
\ei
\end{Remark}

\begin{table}[ht]
\begin{adjustwidth}{-0.5cm}{}
\resizebox{18cm}{!}{
\begin{tabular}{c|l|l|l|l|ll|ll|ll|ll}
\textbf{Metric structure}&\multicolumn{1}{c|}{\textbf{Label}}&\multicolumn{1}{c|}{\textbf{Algebra}}&\multicolumn{1}{c|}{\textbf{Comments}}&\multicolumn{1}{c|}{$\br{\alh}{\alh}\subset\alh$}&\multicolumn{2}{c|}{$\br{\alh}{\alp}\subset\alp$}&\multicolumn{2}{c|}{$\br{\alp}{\alp}\subset\alh$ (Curvature)}&\multicolumn{2}{c|}{$\br{\alp}{\alp}\subset\alp$ (Torsion)}&\multicolumn{2}{c}{$\br{\alh}{\alp}\subset\alh$ (Non-reductivity)}\\\hline
\multirow{15}{*}{$\mathsf G$-\textnormal{ambient}}&\hypertarget{Table:AE1}{$\mathsf{AE1}$}&\hyperlink{Table:AK3}{$\mathsf{AK3}$}&$\text{galilei}\oplus\mathbb R$&&&$\br{\tK}{H}= \tP$&&&&&\\
&\cellcolor{Gray}\hypertarget{Table:AE1a}{$\mathsf{AE1a}$}&\cellcolor{Gray}\hyperlink{Table:AK3a}{$\mathsf{AK3a}$}&$\text{galilei}\inplus\mathbb R$\cellcolor{Gray}&\cellcolor{Gray}&\cellcolor{Gray}&\cellcolor{Gray}$\br{\tK}{H}= \tP$&\cellcolor{Gray}&\cellcolor{Gray}&\cellcolor{Gray}&\cellcolor{Gray}$\br{H}{M}=M$&\cellcolor{Gray}&\cellcolor{Gray}\\
&\hypertarget{Table:AE2}{$\mathsf{AE2}$}&\hyperlink{Table:AK4}{$\mathsf{AK4}$}&&&&$\br{\tK}{H}= \tP$&&&$\br{H}{\tP}= \tP$&&&\\
&\cellcolor{Gray}\hypertarget{Table:AE2+}{$\mathsf{AE2}^+_\kappa$}&\cellcolor{Gray}\hyperlink{Table:AK4+}{$\mathsf{AK4}^+_\kappa$}&\cellcolor{Gray}$\kappa>0$&\cellcolor{Gray}&\cellcolor{Gray}&\cellcolor{Gray}$\br{\tK}{H}= \tP$&\cellcolor{Gray}&\cellcolor{Gray}&\cellcolor{Gray}$\br{H}{\tP}= \tP$&\cellcolor{Gray}$\br{H}{M}=\kappa M$&\cellcolor{Gray}&\cellcolor{Gray}\\
&\hypertarget{Table:AE2-}{$\mathsf{AE2}^-_\kappa$}&\hyperlink{Table:AK4+}{$\mathsf{AK4}^-_\kappa$}&$\kappa>0$&&&$\br{\tK}{H}= \tP$&&&$\br{H}{\tP}= \tP$&$\br{H}{M}=-\kappa M$&&\\
&\cellcolor{Gray}\hypertarget{Table:AE3}{$\mathsf{AE3}$}&\cellcolor{Gray}\hyperlink{Table:AK5}{$\mathsf{AK5}$}&\cellcolor{Gray}$\text{euclidean newton}\oplus\mathbb R$&\cellcolor{Gray}&\cellcolor{Gray}&\cellcolor{Gray}$\br{\tK}{H}= \tP$&\cellcolor{Gray}$\br{H}{\tP}= \tK$&\cellcolor{Gray}&\cellcolor{Gray}&\cellcolor{Gray}&\cellcolor{Gray}&\cellcolor{Gray}\\
&\hypertarget{Table:AE3+}{$\mathsf{AE3}^+_\kappa$}&\hyperlink{Table:AK5+}{$\mathsf{AK5}^+_\kappa$}&$\kappa>0$&&&$\br{\tK}{H}= \tP$&$\br{H}{\tP}= \tK$&&&$\br{H}{M}=\kappa M$&&\\
&\cellcolor{Gray}\hypertarget{Table:AE4}{$\mathsf{AE4}$}&\cellcolor{Gray}\hyperlink{Table:AK6}{$\mathsf{AK6}$}&\cellcolor{Gray}$\text{lorentzian newton}\oplus\mathbb R$&\cellcolor{Gray}&\cellcolor{Gray}&\cellcolor{Gray}$\br{\tK}{H}= \tP$&\cellcolor{Gray}$\br{H}{\tP}= -\tK$&\cellcolor{Gray}&\cellcolor{Gray}&\cellcolor{Gray}&\cellcolor{Gray}&\cellcolor{Gray}\\
&\hypertarget{Table:AE4+}{$\mathsf{AE4}^+_\kappa$}&\hyperlink{Table:AK6+}{$\mathsf{AK6}^+_\kappa$}&$\kappa>0$&&&$\br{\tK}{H}= \tP$&$\br{H}{\tP}= -\tK$&&&$\br{H}{M}=\kappa M$&&\\
&\cellcolor{Gray}\hypertarget{Table:AE5}{$\mathsf{AE5}_{\alpha_+}$}&\cellcolor{Gray}\hyperlink{Table:AK7}{$\mathsf{AK7}_{\alpha_+}$}&\cellcolor{Gray}$\alpha_+>0$&\cellcolor{Gray}&\cellcolor{Gray}&\cellcolor{Gray}$\br{\tK}{H}= \tP$&\cellcolor{Gray}$\br{H}{\tP}= \tK$&\cellcolor{Gray}&\cellcolor{Gray}$\br{H}{\tP}= \alpha_+\tP$&\cellcolor{Gray}&\cellcolor{Gray}&\cellcolor{Gray}\\
&\hypertarget{Table:AE5+}{$\mathsf{AE5}^+_{\alpha_+,\kappa}$}&\hyperlink{Table:AK7+}{$\mathsf{AK7}^+_{\alpha_+,\kappa}$}&$\alpha_+>0\qq\kappa>0$&&&$\br{\tK}{H}= \tP$&$\br{H}{\tP}= \tK$&&$\br{H}{\tP}= \alpha_+\tP$&$\br{H}{M}=\kappa M$&&\\
&\cellcolor{Gray}\hypertarget{Table:AE5-}{$\mathsf{AE5}^-_{\alpha_+,\kappa}$}&\cellcolor{Gray}\hyperlink{Table:AK7-}{$\mathsf{AK7}^-_{\alpha_+,\kappa}$}&\cellcolor{Gray}$\alpha_+>0\qq\kappa>0$&\cellcolor{Gray}&\cellcolor{Gray}&\cellcolor{Gray}$\br{\tK}{H}= \tP$&\cellcolor{Gray}$\br{H}{\tP}= \tK$&\cellcolor{Gray}&\cellcolor{Gray}$\br{H}{\tP}= \alpha_+\tP$&\cellcolor{Gray}$\br{H}{M}=-\kappa M$&\cellcolor{Gray}&\cellcolor{Gray}\\
&\hypertarget{Table:AE6}{$\mathsf{AE6}_{\alpha_-}$}&\hyperlink{Table:AK8}{$\mathsf{AK8}_{\alpha_-}$}&$\alpha_->0$&&&$\br{\tK}{H}= \tP$&$\br{H}{\tP}= -\tK$&&$\br{H}{\tP}= \alpha_-\tP$&&&\\
&\cellcolor{Gray}\hypertarget{Table:AE6+}{$\mathsf{AE6}^+_{\alpha_-,\kappa}$}&\cellcolor{Gray}\hyperlink{Table:AK8+}{$\mathsf{AK8}^+_{\alpha_-,\kappa}$}&\cellcolor{Gray}$\alpha_->0\qq\kappa>0$&\cellcolor{Gray}&\cellcolor{Gray}&\cellcolor{Gray}$\br{\tK}{H}= \tP$&\cellcolor{Gray}$\br{H}{\tP}= -\tK$&\cellcolor{Gray}&\cellcolor{Gray}$\br{H}{\tP}= \alpha_-\tP$&\cellcolor{Gray}$\br{H}{M}=\kappa M$&\cellcolor{Gray}&\cellcolor{Gray}\\
&\hypertarget{Table:AE6-}{$\mathsf{AE6}^-_{\alpha_-,\kappa}$}&\hyperlink{Table:AK8+}{$\mathsf{AK8}^-_{\alpha_-,\kappa}$}&$\alpha_->0\qq\kappa>0$&&&$\br{\tK}{H}= \tP$&$\br{H}{\tP}= -\tK$&&$\br{H}{\tP}= \alpha_-\tP$&$\br{H}{M}=-\kappa M$&&\\
\hline
\multirow{6}{*}{\textnormal{bargmannian}}&\cellcolor{Gray}\hypertarget{Table:AE1b}{$\mathsf{AE1b}$}&\cellcolor{Gray}\hyperlink{Table:AK3b}{$\mathsf{AK3b}$}&\cellcolor{Gray}centrally extended galilei&\cellcolor{Gray}&\cellcolor{Gray}$\br{\tK}{\tP}=M$&\cellcolor{Gray}$\br{\tK}{H}= \tP$&\cellcolor{Gray}&\cellcolor{Gray}&\cellcolor{Gray}&\cellcolor{Gray}&\cellcolor{Gray}&\cellcolor{Gray}\\
&\hypertarget{Table:AE2b}{$\mathsf{AE2b}$}&\hyperlink{Table:AK4b}{$\mathsf{AK4b}$}&&&$\br{\tK}{\tP}=M$&$\br{\tK}{H}= \tP$&&&$\br{H}{\tP}= \tP$&$\br{H}{M}=M$&&\\
&\cellcolor{Gray}\hypertarget{Table:AE3b}{$\mathsf{AE3b}$}&\cellcolor{Gray}\hyperlink{Table:AK5b}{$\mathsf{AK5b}$}&\cellcolor{Gray}centrally extended euclidean newton&\cellcolor{Gray}&\cellcolor{Gray}$\br{\tK}{\tP}=M$&\cellcolor{Gray}$\br{\tK}{H}= \tP$&\cellcolor{Gray}$\br{H}{\tP}= \tK$&\cellcolor{Gray}&\cellcolor{Gray}&\cellcolor{Gray}&\cellcolor{Gray}&\cellcolor{Gray}\\
&\hypertarget{Table:AE4b}{$\mathsf{AE4b}$}&\hyperlink{Table:AK6b}{$\mathsf{AK6b}$}&centrally extended lorentzian newton&&$\br{\tK}{\tP}=M$&$\br{\tK}{H}= \tP$&$\br{H}{\tP}= -\tK$&&&&&\\
&\cellcolor{Gray}\hypertarget{Table:AE5b}{$\mathsf{AE5b}_{\alpha_+}$}&\cellcolor{Gray}\hyperlink{Table:AK7b}{$\mathsf{AK7b}_{\alpha_+}$}&\cellcolor{Gray}$\alpha_+>0$&\cellcolor{Gray}&\cellcolor{Gray}$\br{\tK}{\tP}=M$&\cellcolor{Gray}$\br{\tK}{H}= \tP$&\cellcolor{Gray}$\br{H}{\tP}= \tK$&\cellcolor{Gray}&\cellcolor{Gray}$\br{H}{\tP}= \alpha_+\tP$&\cellcolor{Gray}$\br{H}{M}=\alpha_+ M$&\cellcolor{Gray}&\cellcolor{Gray}\\
&\hypertarget{Table:AE6b}{$\mathsf{AE6b}_{\alpha_-}$}&\hyperlink{Table:AK8b}{$\mathsf{AK8b}_{\alpha_-}$}&$\alpha_->0$&&$\br{\tK}{\tP}=M$&$\br{\tK}{H}= \tP$&$\br{H}{\tP}= -\tK$&&$\br{H}{\tP}= \alpha_-\tP$&$\br{H}{M}=\alpha_- M$&&\\
\hline
\multirow{4}{*}{$\textnormal{carrollian}\oplus\mathbb R$}&\cellcolor{Gray}\hypertarget{Table:AE7}{$\mathsf{AE7}$}&\cellcolor{Gray}\hyperlink{Table:AK9}{$\mathsf{AK9}$}&\cellcolor{Gray}$\text{carroll}\oplus\mathbb R$&\cellcolor{Gray}&\cellcolor{Gray}$\br{\tK}{\tP}= H$&\cellcolor{Gray}&\cellcolor{Gray}&\cellcolor{Gray}&\cellcolor{Gray}&\cellcolor{Gray}&\cellcolor{Gray}&\cellcolor{Gray}\\
&\hypertarget{Table:AE8}{$\mathsf{AE8}$}&\hyperlink{Table:AK10}{$\mathsf{AK10}$}&$\text{ds carroll}\oplus\mathbb R$&&$\br{\tK}{\tP}= H$&&$\br{H}{\tP}= -\tK$&$\br{\tP}{\tP}=-\tJ$&&&&\\
&\cellcolor{Gray}\hypertarget{Table:AE9}{$\mathsf{AE9}$}&\cellcolor{Gray}\hyperlink{Table:AK11}{$\mathsf{AK11}$}&\cellcolor{Gray}$\text{ads carroll}\oplus\mathbb R$&\cellcolor{Gray}&\cellcolor{Gray}$\br{\tK}{\tP}= H$&\cellcolor{Gray}&\cellcolor{Gray}$\br{H}{\tP}= \tK$&\cellcolor{Gray}$\br{\tP}{\tP}=\tJ$&\cellcolor{Gray}&\cellcolor{Gray}&\cellcolor{Gray}&\cellcolor{Gray}\\
&\hypertarget{Table:AE10}{$\mathsf{AE10}$}&\hyperlink{Table:AK12}{$\mathsf{AK12}$}&$\text{carroll light cone}\oplus\mathbb R$&&$\br{\tK}{\tP}= H$&&&&$\br{H}{\tP}= -\tP$&&$\br{\tK}{\tP}= \tJ$&$\br{\tK}{H}= -\tK$\\
\hline
&\cellcolor{Gray}\hypertarget{Table:AE11}{$\mathsf{AE11}$}&\cellcolor{Gray}\hyperlink{Table:AK10}{$\mathsf{AK10}$}&\cellcolor{Gray}$\text{euclidean}\oplus\mathbb R$&\cellcolor{Gray}$\br{\tK}{\tK}=-\tJ$&\cellcolor{Gray}$\br{\tK}{\tP}= H$&\cellcolor{Gray}$\br{\tK}{H}= -\tP$&\cellcolor{Gray}&\cellcolor{Gray}&\cellcolor{Gray}&\cellcolor{Gray}&\cellcolor{Gray}&\cellcolor{Gray}\\
$\textnormal{riemannian}\oplus\mathbb R$&\hypertarget{Table:AE12}{$\mathsf{AE12}$}&\hyperlink{Table:AK13}{$\mathsf{AK13}$}&$\text{sphere}\oplus\mathbb R$&$\br{\tK}{\tK}=-\tJ$&$\br{\tK}{\tP}= H$&$\br{\tK}{H}= -\tP$&$\br{H}{\tP}=-\tK$&$\br{\tP}{\tP}=-\tJ$&&&&\\
&\cellcolor{Gray}\hypertarget{Table:AE13}{$\mathsf{AE13}$}&\cellcolor{Gray}\hyperlink{Table:AK12}{$\mathsf{AK12}$}&\cellcolor{Gray}$\text{hyperbolic}\oplus\mathbb R$&\cellcolor{Gray}$\br{\tK}{\tK}=-\tJ$&\cellcolor{Gray}$\br{\tK}{\tP}= H$&\cellcolor{Gray}$\br{\tK}{H}= -\tP$&\cellcolor{Gray}$\br{H}{\tP}=\tK$&\cellcolor{Gray}$\br{\tP}{\tP}=\tJ$&\cellcolor{Gray}&\cellcolor{Gray}&\cellcolor{Gray}&\cellcolor{Gray}\\
\hline
&\hypertarget{Table:AE14}{$\mathsf{AE14}$}&\hyperlink{Table:AK11}{$\mathsf{AK11}$}&$\text{minkowski}\oplus\mathbb R$&$\br{\tK}{\tK}=\tJ$&$\br{\tK}{\tP}= H$&$\br{\tK}{H}= \tP$&&&&&&\\
$\textnormal{lorentzian}\oplus\mathbb R$&\cellcolor{Gray}\hypertarget{Table:AE15}{$\mathsf{AE15}$}&\cellcolor{Gray}\hyperlink{Table:AK12}{$\mathsf{AK12}$}&\cellcolor{Gray}$\text{de sitter}\oplus\mathbb R$&\cellcolor{Gray}$\br{\tK}{\tK}=\tJ$&\cellcolor{Gray}$\br{\tK}{\tP}= H$&\cellcolor{Gray}$\br{\tK}{H}= \tP$&\cellcolor{Gray}$\br{H}{\tP}=-\tK$&\cellcolor{Gray}$\br{\tP}{\tP}=-\tJ$&\cellcolor{Gray}&\cellcolor{Gray}&\cellcolor{Gray}&\cellcolor{Gray}\\
&\hypertarget{Table:AE16}{$\mathsf{AE16}$}&\hyperlink{Table:AK14}{$\mathsf{AK14}$}&$\text{anti de sitter}\oplus\mathbb R$&$\br{\tK}{\tK}=\tJ$&$\br{\tK}{\tP}= H$&$\br{\tK}{H}= \tP$&$\br{H}{\tP}=\tK$&$\br{\tP}{\tP}=\tJ$&&&&
\end{tabular}
  }
      \end{adjustwidth}
      \caption{Universal projectable ambient triplets}
      \label{Table:Projectable ambient triplets}
\end{table}

\begin{table}[ht]
\centering
\resizebox{10cm}{!}{
\begin{tabular}{l|l|lll|c}
\centering
\textbf{Galilean Klein pair}&\textbf{Ambient Klein pair}&\multicolumn{3}{c|}{\textbf{Additional commutators}}&\textbf{Extra parameter}\\\hline
\multirow{3}{*}{\hyperlink{Table:E1}{$\mathsf{E1}$}}&\hyperlink{Table:AE1}{$\mathsf{AE1}$}&&&&\\
&\cellcolor{Gray}\hyperlink{Table:AE1a}{$\mathsf{AE1a}$}&\cellcolor{Gray}&\cellcolor{Gray}&\cellcolor{Gray}$\br{H}{M}=M$&\cellcolor{Gray}\\
&\hyperlink{Table:AE1b}{$\mathsf{AE1b}$}&$\br{\tK}{\tP}=M$&&&\\
\hline
\multirow{4}{*}{\hyperlink{Table:E2}{$\mathsf{E2}$}}&\cellcolor{Gray}\hyperlink{Table:AE2}{$\mathsf{AE2}$}&\cellcolor{Gray}&\cellcolor{Gray}&\cellcolor{Gray}&\cellcolor{Gray}\\
&\hyperlink{Table:AE2+}{$\mathsf{AE2}^+_\kappa$}&&&$\br{H}{M}=\kappa M$&$\kappa>0$\\
&\cellcolor{Gray}\hyperlink{Table:AE2-}{$\mathsf{AE2}^-_\kappa$}&\cellcolor{Gray}&\cellcolor{Gray}&\cellcolor{Gray}$\br{H}{M}=-\kappa M$&\cellcolor{Gray}$\kappa>0$\\
&\hyperlink{Table:AE2b}{$\mathsf{AE2b}$}&$\br{\tK}{\tP}=M$&&$\br{H}{M}=M$&\\
\hline
\multirow{3}{*}{\hyperlink{Table:E3}{$\mathsf{E3}$}}&\cellcolor{Gray}\hyperlink{Table:AE3}{$\mathsf{AE3}$}&\cellcolor{Gray}&\cellcolor{Gray}&\cellcolor{Gray}&\cellcolor{Gray}\\
&\hyperlink{Table:AE3+}{$\mathsf{AE3}^+_\kappa$}&&&$\br{H}{M}=\kappa M$&$\kappa>0$\\
&\cellcolor{Gray}\hyperlink{Table:AE3b}{$\mathsf{AE3b}$}&\cellcolor{Gray}$\br{\tK}{\tP}=M$&\cellcolor{Gray}&\cellcolor{Gray}&\cellcolor{Gray}\\
\hline
\multirow{3}{*}{\hyperlink{Table:E4}{$\mathsf{E4}$}}&\hyperlink{Table:AE4}{$\mathsf{AE4}$}&&&&\\
&\cellcolor{Gray}\hyperlink{Table:AE4+}{$\mathsf{AE4}^+_\kappa$}&\cellcolor{Gray}&\cellcolor{Gray}&\cellcolor{Gray}$\br{H}{M}=\kappa M$&\cellcolor{Gray}$\kappa>0$\\
&\hyperlink{Table:AE4b}{$\mathsf{AE4b}$}&$\br{\tK}{\tP}=M$&&&\\
\hline
\multirow{4}{*}{\hyperlink{Table:E5}{$\mathsf{E5}_{\alpha_+}$}}&\cellcolor{Gray}\hyperlink{Table:AE5}{$\mathsf{AE5}_{\alpha_+}$}&\cellcolor{Gray}&\cellcolor{Gray}&\cellcolor{Gray}&\cellcolor{Gray}\\
&\hyperlink{Table:AE5+}{$\mathsf{AE5}^+_{\alpha_+,\kappa}$}&&&$\br{H}{M}=\kappa M$&$\kappa>0$\\
&\cellcolor{Gray}\hyperlink{Table:AE5-}{$\mathsf{AE5}^-_{\alpha_+,\kappa}$}&\cellcolor{Gray}&\cellcolor{Gray}&\cellcolor{Gray}$\br{H}{M}=-\kappa M$&\cellcolor{Gray}$\kappa>0$\\
&\hyperlink{Table:AE5b}{$\mathsf{AE5b}_{\alpha_+}$}&$\br{\tK}{\tP}=M$&&$\br{H}{M}=\alpha_+M$&\\
\hline
\multirow{4}{*}{\hyperlink{Table:E6}{$\mathsf{E6}_{\alpha_-}$}}&\cellcolor{Gray}\hyperlink{Table:AE6}{$\mathsf{AE6}_{\alpha_-}$}&\cellcolor{Gray}&\cellcolor{Gray}&\cellcolor{Gray}&\cellcolor{Gray}\\
&\hyperlink{Table:AE6+}{$\mathsf{AE6}^+_{\alpha_-,\kappa}$}&&&$\br{H}{M}=\kappa M$&$\kappa>0$\\
&\cellcolor{Gray}\hyperlink{Table:AE6-}{$\mathsf{AE6}^-_{\alpha_-,\kappa}$}&\cellcolor{Gray}&\cellcolor{Gray}&\cellcolor{Gray}$\br{H}{M}=\kappa M$&\cellcolor{Gray}$\kappa>0$\\
&\hyperlink{Table:AE6b}{$\mathsf{AE6b}_{\alpha_-}$}&$\br{\tK}{\tP}=M$&&$\br{H}{M}=\alpha_-M$&
  \end{tabular}
  }
      \caption{Possible lifts of effective galilean Klein pairs}
      \label{Table:Possible lifts of effective galilean Klein pairs}
\end{table}


\section{Drink me: \IW contractions of Klein pairs}
\label{section:Drink me}
Section \ref{section:Climbing up one leg of the table: an ambient perspective on Klein pairs} led us to a crossroad, with two paths diverging before us. On our right lied the---admittedly simpler and well-trodden---path of ambient aristotelian algebras, which we explored in Section \ref{section:Bargmann and his modern rivals}. This path led us to the classification of projectable ambient triplets, as seen in Table \ref{Table:Projectable ambient triplets}. On our left, we saw the---harder and less traveled---path of ambient kinematical algebras. Exploring this direction via the main road would amount in classifying Klein pairs lying at the intersection of projectable triplets \UnskipRef{Definition:Projectable triplet} and ambient kinematical Klein pairs \UnskipRef{Definition:Ambient kinematical Klein pairs}. Alas, such a pursuit is too daunting a task to be taken on now and we will have to save it for future exploration. Rather, we will prefer to take a shortcut at present, at the intersection of the two following directions: dimensional reduction via branching rules, as discussed in \SectionRef{section:A Tangled Tale: (s,v)-Lie algebras}, on the one hand, and \IW contractions, on the other, that we shall address now. 

\subsection{Contraction of Klein pairs}
\label{section:Contraction of Klein pairs}

Letting $(\alg,\IWun)$ be a Klein pair and $\IWdeux$ be a supplementary of $\IWun$ in $\alg$ so that $\alg=\IWun\oplus\IWdeux$ as a vector space, the \emph{\IW contraction} \cite{Inonu1953} of $\alg$ along the splitting $\alg=\IWun\oplus \IWdeux$ is defined as the Lie algebra whose underlying vector space is the one underlying $\alg$ and whose Lie structure is obtained from the one of $\alg$ by sending the commutators:
\[
 \begin{tikzcd}[row sep = small,column sep = small,ampersand replacement=\&]
 \br{\IWun}{\IWdeux}\subset\IWun\& \br{\IWdeux}{\IWdeux}\subset\IWun\& \br{\IWdeux}{\IWdeux}\subset\IWdeux \quad\text{to zero}.
\end{tikzcd}
\]
 The only non-vanishing commutators resulting from the contraction are therefore:
\[
 \begin{tikzcd}[row sep = small,column sep = small,ampersand replacement=\&]
 \br{\IWun}{\IWun}\subset\IWun\&\br{\IWun}{\IWdeux}\subset\IWdeux
\end{tikzcd}
\]
  so that $\IWdeux$ becomes an abelian ideal with respect to the contracted Lie structure.
 \begin{Example}
Applying the above \IW contraction to the kinematical Klein pairs of Table \ref{Table:Effective kinematical Klein pairs} along the splitting $\alg=\IWun\oplus\IWdeux$ with $\IWun=\tK\oplus\tJ$ and $\IWdeux=H\oplus\tP$ yields, for each of the four families, the pair appearing in the upper row of the family (\eg the \IW contraction of any lorentzian kinematical Klein pair yields the minkowski Klein pair, \etc) \ie the contraction sets the curvature, torsion and non-reductivity to zero.
\end{Example}

\paragraph{Reducing and contracting}

As noted in Section \ref{section:A Tangled Tale: (s,v)-Lie algebras}, the use of branching rules allows to map Klein pairs with $(d+1)$-dimensional spatial isotropy to Klein pairs with $d$-dimensional spatial isotropy containing an augmented set of generators. Specifically, the restriction of Corollary \ref{corollary:dimensional reduction of Klein pairs} to the case $k=1$ asserts that any kinematical Klein pair with $(d+1)$-dimensional spatial isotropy is isomorphic as a Klein pair to an ambient kinematical Klein pair \UnskipRef{Definition:Ambient kinematical Klein pairs} with $d$-dimensional spatial isotropy. A promising starting point to generate ambient kinematical lifts of galilean Klein pairs is then to start from the corresponding kinematical Klein pair with $(d+1)$-dimensional spatial isotropy and to perform a dimensional reduction.   
Let us illustrate the above on the simplest example, namely the galilei Klein pair \hyperlink{Table:E1}{$\mathsf{E1}$}:
\begin{Example}[Galilei Klein pair]
\label{Example:Galilei Klein pair}
Starting from the galilei Klein pair \hyperlink{Table:E1}{$\mathsf{E1}$} with $(d+1)$-dimensional spatial isotropy:
\[
 \begin{tikzcd}[row sep = small,column sep = small,ampersand replacement=\&]
\br{\tK}{H}=\tP\&\br{\tJ}{\tP}=\tP\&\br{\tJ}{\tJ}=\tJ
\end{tikzcd}
\]
and applying the branching rules:
\[
 \begin{tikzcd}[row sep = small,column sep = small,ampersand replacement=\&]
 H\mapsto H\&  P_0\mapsto M\&  \tP\mapsto \tP\&  K_0\mapsto C\&  \tK\mapsto \tK\&  J_{0i}\mapsto \tD\&  \tJ\mapsto \tJ
 \end{tikzcd}
\]
yields the following ambient kinematical Klein pair with $d$-dimensional spatial isotropy:
\bea
\label{equation:galilei in ambient form}
 \begin{tikzcd}[row sep = small,column sep = small,ampersand replacement=\&]
\br{\tK}{H}=\tP\&\br{C}{H}= M\&\br{\tD}{M}= -\tP\& \br{\tD}{\tP}= M\&\br{\tD}{C}= -\tK\\
\br{\tD}{\tK}= C\&\br{\tD}{\tD}= -\tJ\&\br{\tJ}{\tD}=\tD\&\br{\tJ}{\tK}=\tK\&\br{\tJ}{\tP}=\tP\&\br{\tJ}{\tJ}=\tJ\, .
 \end{tikzcd}
\eea
\end{Example}
The obtained ambient kinematical Klein pair admits as $\ad_\alh$-invariant structure on $\gh$ the $(d+1)$-dimensional galilean structure \UnskipRef{Definition:Metric structures on Klein pairs}:
\[
 \begin{tikzcd}[row sep = small,column sep = small,ampersand replacement=\&]
\bs{\psi}=-H^*\&\bs{\gamma}=M^*\otimes M^*+\tP^*\otimes \tP^*\, .
\end{tikzcd}
\]
As such, the obtained ambient Klein pair is \emph{not} leibnizian \UnskipRef{Definition:leibnizian metric structures on Klein pairs}, as can be readily seen from the fact that the obstruction $\br{\tD}{M}\sim\tP$ does not vanish [\cf footnote \ref{footnote:leibnizian criteria}]. This can however be cured by performing a suitable \IW contraction that will send the corresponding obstruction to zero. As reviewed above, such contraction presupposes to find a subalgebra $\alk$ of our ambient algebra $\alg$. Note that, by construction, the above ambient kinematical algebra possesses a canonical subalgebra \UnskipRef{Remark:subalgebra} $\alk:=H\oplus\tP\oplus\tK\oplus\tJ$ whose commutators take the same form as those from the original kinematical algebra (the galilei algebra in the above example \UnskipRef{Example:Galilei Klein pair}) albeit with $d$-dimensional spatial isotropy. Performing an \IW contraction of the Klein pair $(\alg,\alk)$ along the decomposition $\alg=\alk\oplus\ali$, where $\ali:=M\oplus C\oplus \tD$, yields a new ambient kinematical Klein pair for which $\ali$ is a canonical abelian ideal.\footnote{This ensures in particular that the obstruction $\br{\tD}{M}\sim\tP$ vanishes.} The output of the procedure is thus a projectable triplet $(\alg,\alh,\ali)$ \UnskipRef{Definition:Projectable triplet} canonically associated with any kinematical algebra. Applying such canonical contraction to the above example yields:
 \begin{Example}[From galilei to leibniz]
\label{Example:From galilei to leibniz}
Starting from the ambient kinematical algebra \eqref{equation:galilei in ambient form} with $d$-dimensional spatial isotropy---obtained by dimensional reduction of the galilei Klein pair with $(d+1)$-dimensional spatial isotropy---and performing an \IW contraction along the canonical splitting $\alg=\alk\oplus\ali$ yields the leibniz projectable triplet \UnskipRef{Example:Leibniz projectable triplet}. The associated quotient Klein pair is obviously isomorphic to the original galilei Klein pair \hyperlink{Table:E1}{$\mathsf{E1}$} with $d$-dimensional spatial isotropy.
\end{Example}

\subsection{Leibnizian lifts of galilean Klein pairs}
Democratising the procedure illustrated in Examples \ref{Example:Galilei Klein pair}-\ref{Example:From galilei to leibniz} to each of the galilean Klein pairs \hyperlink{Table:E1}{$\mathsf{E1}$}-\hyperlink{Table:E6}{$\mathsf{E6}_{\alpha_-}$} listed in Table \ref{Table:Effective kinematical Klein pairs} allows to define curved/torsional avatars of the leibniz Klein pair \UnskipRef{Example:Leibniz Klein pair}. For all the induced Klein pairs\footnote{Where $\alg=M\oplus H\oplus C\oplus \tP\oplus\tD\oplus\tK\oplus\tJ$ and $\alh=C\oplus\tD\oplus\tK\oplus\tJ$. } $(\alg,\alh)$, the subspace $\ali=M\oplus C\oplus \tD$ is by construction an abelian ideal, making the triplet $(\alg,\alh,\ali)$ into a projectable triplet \UnskipRef{Definition:Projectable triplet} lifting the initial galilean Klein pair.

\begin{Proposition}[Leibnizian lifts of galilean Klein pairs]
\label{Proposition:Leibnizian lifts of galilean Klein pairs}
The projectable triplets obtained from dimensional reduction and \IW contraction of the galilean Klein pairs \hyperlink{Table:E1}{$\mathsf{E1}$}-\hyperlink{Table:E6}{$\mathsf{E6}_{\alpha_-}$} are listed in \textnormal{Table \ref{Table:Reductive leibnizian lifts of galilean algebras}}.
\end{Proposition}

\begin{table}[ht]
\begin{adjustwidth}{-0.5cm}{}
\resizebox{18cm}{!}{
\begin{tabular}{l|l|l|lll|ll|ll|l}
\multicolumn{1}{c|}{\textbf{Label}}&\multicolumn{1}{c|}{\textbf{Comments}}&\multicolumn{1}{c|}{$\br{\alh}{\alh}\subset\alh$}&\multicolumn{3}{c|}{$\br{\alh}{\alp}\subset\alp$}&\multicolumn{2}{c|}{$\br{\alp}{\alp}\subset\alh$ (Curvature)}&\multicolumn{2}{c|}{$\br{\alp}{\alp}\subset\alp$ (Torsion)}&\multicolumn{1}{c}{$\br{\alh}{\alp}\subset\alh$ (Non-reductivity)}\\\hline
\hypertarget{Table:L1}{$\mathsf{L1}$}&\text{leibniz}&$\br{\tD}{\tK}= C$&$\br{\tK}{H}= \tP$&$\br{\tD}{\tP}= M$&$\br{C}{H}= M$&&&&&\\
\rowcolor{Gray}
\hypertarget{Table:L2}{$\mathsf{L2}$}&&$\br{\tD}{\tK}= C$&$\br{\tK}{H}= \tP$&$\br{\tD}{\tP}= M$&$\br{C}{H}= M$&&&$\br{H}{M}=M$&$\br{H}{\tP}= \tP$&\\
\hypertarget{Table:L3}{$\mathsf{L3}$}&\text{euclidean leibniz}&$\br{\tD}{\tK}= C$&$\br{\tK}{H}= \tP$& $\br{\tD}{\tP}= M$&$\br{C}{H}= M$&$\br{H}{M}= C$& $\br{H}{\tP}= \tK$&&&\\
\rowcolor{Gray}
\hypertarget{Table:L4}{$\mathsf{L4}$}&\text{lorentzian leibniz}&$\br{\tD}{\tK}= C$&$\br{\tK}{H}= \tP$& $\br{\tD}{\tP}= M$&$\br{C}{H}= M$&$\br{H}{M}=-C$&$\br{H}{\tP}=-\tK$&&&\\
\hypertarget{Table:L5}{$\mathsf{L5}_{\alpha_+}$}&$\alpha_+>0$&$\br{\tD}{\tK}= C$&$\br{\tK}{H}= \tP$& $\br{\tD}{\tP}= M$&$\br{C}{H}= M$&$\br{H}{M}= C$& $\br{H}{\tP}= \tK$&$\br{H}{M}= \alpha_+M$& $\br{H}{\tP}= \alpha_+\tP$&\\
\rowcolor{Gray}
\hypertarget{Table:L6}{$\mathsf{L6}_{\alpha_-}$}&$\alpha_->0$&$\br{\tD}{\tK}= C$&$\br{\tK}{H}= \tP$& $\br{\tD}{\tP}= M$&$\br{C}{H}= M$&$\br{H}{M}=-C$&$\br{H}{\tP}=-\tK$&$\br{H}{M}=\alpha_-M$& $\br{H}{\tP}= \alpha_-\tP$&
\end{tabular}
  }
      \end{adjustwidth}
      \caption{Reductive leibnizian lifts of galilean Klein pairs}
      \label{Table:Reductive leibnizian lifts of galilean algebras}
\end{table}
\begin{Remark}
Note that, for each of the Klein pairs in Table \ref{Table:Reductive leibnizian lifts of galilean algebras}, the subalgebra $\alh=C\oplus\tD\oplus\tK\oplus\tJ$ is isomorphic to the carroll algebra \hyperlink{Table:K9}{$\mathsf{K9}$}. This implies in particular that the boost subalgebra spanned by $C\oplus\tD\oplus\tK$ is isomorphic to the Heisenberg algebra in $d$ dimensions, the geometric counterpart thereof being a unification of the notions of \emph{galilean boosts} (also called  \emph{Milne boosts} \cite{Carter1994,Duval:1993pe,Bekaert:2014bwa}) and \emph{carrollian boosts} \cite{Bekaert2015b,Hartong:2015xda} into the non-abelian \emph{leibnizian boosts} mediated by the local Heisenberg group, \cf \cite{Bekaert2015b,Morand2023} for details.
\end{Remark}
The Klein pairs collected in Table \ref{Table:Reductive leibnizian lifts of galilean algebras} are reductive and can readily be seen to be effective thanks to the obstructions of type $\br{\alh}{\alp}\subset\alp$ [where $\alp=M\oplus H\oplus \tP$]. Furthermore, the quotient Klein pairs $(\alg/\ali,\alh/\alj)$ [where $\alj$ denotes the ideal of $\alh$ defined as $\alj:=\alh\cap \ali$] are also effective, thanks to the fact that the obstruction $\br{\tK}{H}= \tP$ survives the quotient along the canonical ideal $\ali=M\oplus C\oplus \tD$. The following sequence of Lie algebras:
\[
\begin{tikzcd}
0 \arrow[r] &\ali  \arrow[r, hook, "i"] & \alg \arrow[r, two heads,"\pi"]        & \alg_0\arrow[r]  & 0
\end{tikzcd}
\]
where $\alg=\mathsf{L}_i$ and $\alg_0=\mathsf{E}_i$,  is exact for all $i\in\pset{1,\ldots,6}$ hence the above Klein pairs provide ambient kinematical lifts for all galilean Klein pairs.

\begin{Proposition}
\label{Proposition:Leibnizian structure on Klein pairs}
The Klein pairs collected in \textnormal{Table \ref{Table:Reductive leibnizian lifts of galilean algebras}} are leibnizian \UnskipRef{Definition:leibnizian metric structures on Klein pairs} for the canonical leibnizian metric structure \eqref{equation:leibniz metric structure on leibniz Klein pair}.
\end{Proposition}
\begin{proof}
The proposition follows straightforwardly from the necessary and sufficient criterion given in footnote \ref{footnote:leibnizian criteria}.
\end{proof}
Overall, we defined three possible classes of lifts for galilean Klein pairs, with respective metric structure being $\mathsf{G}$-ambient, bargmannian and leibnizian. Since bargmannian and $\mathsf{G}$-ambient Klein pairs are in particular leibnizian \UnskipRef{Definition:Ambient aristotelian metric structures on Klein pairs}, a natural question one may ask is the following:
\begin{center}
\textit{
Are the $\mathsf{G}$-ambient and bargmannian Klein pairs listed in \textnormal{Table \ref{Table:Projectable ambient triplets}} \\
sub-Klein pairs of the leibnizian Klein pairs of \textnormal{Table \ref{Table:Reductive leibnizian lifts of galilean algebras}}?
}
\end{center}
As noted previously \UnskipRef{Example:Ambient aristotelian Klein pairs}, the question can be answered positively in the case of the leibniz Klein pair \hyperlink{Table:L1}{$\mathsf{L1}$} which admits both the trivial extension of the galilei Klein pair \hyperlink{Table:AE1}{$\mathsf{AE1}$} and the bargmann Klein pair \hyperlink{Table:AE1b}{$\mathsf{AE1b}$} as sub-Klein pairs. In order to address the general case, we will recast the leibnizian algebras displayed in Table \ref{Table:Reductive leibnizian lifts of galilean algebras}  in the spirit of Remark \ref{Remark:Projectable ambient triplets table}. The latter can collectively be abstracted as the Lie algebras $\mathsf{Li}_{\epsilon,\om}$ with non-trivial commutators:
 \[
 \begin{tikzcd}[row sep = small,column sep = small,ampersand replacement=\&]
\mathsf{Li}_{\epsilon,\om}\&\br{\tD}{\tK}= C\&\hspace{-11.3mm}\br{\tD}{\tP}= M\&\hspace{-12.7mm}\br{C}{H}=M\\
\&\br{\tK}{H}=\tP\&\br{H}{\tP}=\epsilon\, \tK+\omega\, \tP\&\br{H}{M}=\epsilon\, C+\omega\, M
\end{tikzcd}
 \]
where $\epsilon\in\mathbb R$ encodes the curvature and $\om\in\mathbb R$ the torsion. It can readily be seen from the above commutation relations that $M\oplus H\oplus \tP\oplus \tK\oplus \tJ$ spans a subalgebra of $\mathsf{Li}_{\epsilon,\om}$ if and only if $\epsilon=0$. In that case, the sub-Klein pair of $\mathsf{Li}_{0,\om}$ is isomorphic to the $\mathsf{G}$-ambient Klein pair $\mathsf{AE_{0,\om,\om}}$ [\cf \eqref{equation:algebra G-ambient}]. Alternatively, performing a redefinition $\tK\mapsto \tilde \tK:=\tK+\tD$, the sub-Klein pair of $\mathsf{Li}_{0,\om}$ spanned by $M\oplus H\oplus \tP\oplus \tilde\tK\oplus \tJ$ is isomorphic to the bargmannian Klein pair $\mathsf{AEb_{0,\om}}$ [\cf \eqref{equation:algebra bargmann}].

Coming back to the classification of Table \ref{Table:Projectable ambient triplets}, we conclude that  \hyperlink{Table:AE1}{$\mathsf{AE1}$} and \hyperlink{Table:AE1b}{$\mathsf{AE1b}$} (resp.  \hyperlink{Table:AE2+}{$\mathsf{AE2}^+_1$} and \hyperlink{Table:AE2b}{$\mathsf{AE2b}$}) admit an embedding into \hyperlink{Table:L1}{$\mathsf{L1}$} (resp. \hyperlink{Table:L1}{$\mathsf{L2}$}). In fact, in can be shown that these are the only $\mathsf G$-ambient or bargmannian Klein pairs admitting an embedding into the leibnizian Klein pairs of Table \ref{Table:Reductive leibnizian lifts of galilean algebras}. This apparently surprising fact can be justified and (partially) amended by generalising Table \ref{Table:Reductive leibnizian lifts of galilean algebras} to include non-reductivity. Indeed, a noticeable feature of the possible ($\mathsf G$-ambient or bargmannian) lifts of galilean Klein pairs displayed in Table \ref{Table:Possible lifts of effective galilean Klein pairs} is that they were all reductive. The exhaustive classification performed in Section \ref{section:Bargmann and his modern rivals} thus provides a \emph{no-go} result regarding non-reductive lifts of galilean Klein pairs into ambient aristotelian Klein pairs. As exemplified in the next section, this no-go result does not longer hold when lifting into ambient kinematical Klein pairs. 

\subsection{Non-reductive lifts of galilean Klein pairs}
\label{section:Non-reductive lifts of galilean Klein pairs}
As a last variation on our leitmotiv, let us indulge---without further motivation than to catch a glimpse of the richness of leibnizian projective Klein pairs---in presenting a generalisation of Table \ref{Table:Reductive leibnizian lifts of galilean algebras} including \emph{non-reductive} lifts of galilean Klein pairs. Such non-reductive deformations are collected in Table \ref{Table:Leibnizian lifts of galilean algebras}. 
\begin{table}[ht]
\begin{adjustwidth}{-0.5cm}{}
\resizebox{18cm}{!}{
\begin{tabular}{l|l|l|lll|ll|ll|ll}
\multicolumn{1}{c|}{\textbf{Label}}&\multicolumn{1}{c|}{\textbf{Comments}}&\multicolumn{1}{c|}{$\br{\alh}{\alh}\subset\alh$}&\multicolumn{3}{c|}{$\br{\alh}{\alp}\subset\alp$}&\multicolumn{2}{c|}{$\br{\alp}{\alp}\subset\alh$ (Curvature)}&\multicolumn{2}{c|}{$\br{\alp}{\alp}\subset\alp$ (Torsion)}&\multicolumn{2}{c}{$\br{\alh}{\alp}\subset\alh$ (Non-reductivity)}\\\hline
$\mathsf{L1}$&\text{leibniz}&$\br{\tD}{\tK}= C$&$\br{\tK}{H}= \tP$&$\br{\tD}{\tP}= M$&$\br{C}{H}= M$&&&&&&\\
\rowcolor{Gray}
$\mathsf{L1a}$&&$\br{\tD}{\tK}= C$&$\br{\tK}{H}= \tP$&$\br{\tD}{\tP}= M$&$\br{C}{H}= M$&&&$\br{H}{M}=M$&&$\br{H}{C}= C$&$\br{H}{\tD}= \tD$\\
$\mathsf{L2}$&&$\br{\tD}{\tK}= C$&$\br{\tK}{H}= \tP$&$\br{\tD}{\tP}= M$&$\br{C}{H}= M$&&&$\br{H}{M}=M$&$\br{H}{\tP}= \tP$&&\\
\rowcolor{Gray}
$\mathsf{L2}_r$&$r\neq0$&$\br{\tD}{\tK}= C$&$\br{\tK}{H}= \tP$&$\br{\tD}{\tP}= M$&$\br{C}{H}= M$&&&$\br{H}{M}=(1+r)\, M$&$\br{H}{\tP}= \tP$&$\br{H}{C}= r\, C$&$\br{H}{\tD}= r\tD$\\
$\mathsf{L3}$&\text{euclidean leibniz}&$\br{\tD}{\tK}= C$&$\br{\tK}{H}= \tP$& $\br{\tD}{\tP}= M$&$\br{C}{H}= M$&$\br{H}{M}= C$& $\br{H}{\tP}= \tK$&&&&\\
\rowcolor{Gray}
$\mathsf{L3}_r$&$r\neq0$&$\br{\tD}{\tK}= C$&$\br{\tK}{H}= \tP$& $\br{\tD}{\tP}= M$&$\br{C}{H}= M$&$\br{H}{M}= C$& $\br{H}{\tP}= \tK$&$\br{H}{M}=r\, M$&&$\br{H}{C}= r\, C$&$\br{H}{\tD}= r\tD$\\
$\mathsf{L4}$&\text{lorentzian leibniz}&$\br{\tD}{\tK}= C$&$\br{\tK}{H}= \tP$& $\br{\tD}{\tP}= M$&$\br{C}{H}= M$&$\br{H}{M}=-C$&$\br{H}{\tP}=-\tK$&$\br{H}{M}=r\, M$&&$\br{H}{C}= r\, C$&$\br{H}{\tD}= r\tD$\\
\rowcolor{Gray}
$\mathsf{L4}_r$&$r\neq0$&$\br{\tD}{\tK}= C$&$\br{\tK}{H}= \tP$& $\br{\tD}{\tP}= M$&$\br{C}{H}= M$&$\br{H}{M}=-C$&$\br{H}{\tP}=-\tK$&$\br{H}{M}=r\, M$&&$\br{H}{C}= r\, C$&$\br{H}{\tD}= r\tD$\\
$\mathsf{L5}_{\alpha_+}$&$\alpha_+>0$&$\br{\tD}{\tK}= C$&$\br{\tK}{H}= \tP$& $\br{\tD}{\tP}= M$&$\br{C}{H}= M$&$\br{H}{M}= C$& $\br{H}{\tP}= \tK$&$\br{H}{M}= \alpha_+M$& $\br{H}{\tP}= \alpha_+\tP$&&\\
\rowcolor{Gray}
$\mathsf{L5}_{\alpha_+,r}$&$\alpha_+>0\qq r\neq0$&$\br{\tD}{\tK}= C$&$\br{\tK}{H}= \tP$& $\br{\tD}{\tP}= M$&$\br{C}{H}= M$&$\br{H}{M}= C$& $\br{H}{\tP}= \tK$&$\br{H}{M}= (\alpha_++r)\, M$& $\br{H}{\tP}= \alpha_+\tP$&$\br{H}{C}= r\, C$&$\br{H}{\tD}= r\tD$\\
$\mathsf{L6}_{\alpha_-}$&$\alpha_->0$&$\br{\tD}{\tK}= C$&$\br{\tK}{H}= \tP$& $\br{\tD}{\tP}= M$&$\br{C}{H}= M$&$\br{H}{M}=-C$&$\br{H}{\tP}=-\tK$&$\br{H}{M}=\alpha_-M$& $\br{H}{\tP}= \alpha_-\tP$&&\\
\rowcolor{Gray}
$\mathsf{L6}_{\alpha_-,r}$&$\alpha_->0\qq r\neq0$&$\br{\tD}{\tK}= C$&$\br{\tK}{H}= \tP$& $\br{\tD}{\tP}= M$&$\br{C}{H}= M$&$\br{H}{M}=-C$&$\br{H}{\tP}=-\tK$&$\br{H}{M}=(\alpha_-+r)\, M$& $\br{H}{\tP}=\alpha_-\tP$&$\br{H}{C}= r\, C$&$\br{H}{\tD}= r\tD$
\end{tabular}
  }
      \end{adjustwidth}
      \caption{Leibnizian lifts of galilean Klein pairs}
      \label{Table:Leibnizian lifts of galilean algebras}
\end{table}

The above Klein pairs share the same $\ad_\alh$-invariants as the ones displayed in Table \ref{Table:Reductive leibnizian lifts of galilean algebras} and are thus leibnizian \UnskipRef{Definition:leibnizian metric structures on Klein pairs} for the canonical leibnizian metric structure \eqref{equation:leibniz metric structure on leibniz Klein pair}. 

As we would like to now show, including non-reductivity allows to embed (almost all) ambient aristotelian lifts of galilean Klein pairs---as displayed in Table \ref{Table:Projectable ambient triplets}---as sub-Klein pairs of leibnizian Klein pairs. 

Let us start by collectively denote the above leibnizian Klein pairs as:
 \[
 \begin{tikzcd}[row sep = small,column sep = small,ampersand replacement=\&]
\mathsf{Li}_{\epsilon,\om,r}\&\br{\tD}{\tK}= C\&\hspace{-11.5mm}\br{\tD}{\tP}= M\&\hspace{-12mm}\br{C}{H}= M-r\, C\&\br{H}{\tD}=r\, \tD\\
\&\br{\tK}{H}=\tP\&\br{H}{\tP}=\epsilon\, \tK+\omega\, \tP\&\br{H}{M}=\epsilon\, C+(\omega+r)\, M\&
\end{tikzcd}
 \]
where the extra parameter $r\in\mathbb R$ encodes non-reductivity.
\begin{description}
 \item[$\mathsf{G}$-ambient] Let us perform the change of variables $M\mapsto\tilde M:=M+x\, C$ yielding:
  \[
 \begin{tikzcd}[row sep = small,column sep = small,ampersand replacement=\&]
\mathsf{Li}_{\epsilon,\om,r}\&\br{\tD}{\tK}= C\&\hspace{-0.2cm}\br{\tD}{\tP}= \tilde M-x\, C\&
\hspace{-2.9cm}
\br{C}{H}= \tilde M-(r+x)\, C\&
\hspace{-2.5cm}
\br{H}{\tD}=r\, \tD\\
\&\br{\tK}{H}=\tP\&\br{H}{\tP}=\epsilon\, \tK+\omega\, \tP\&\br{H}{\tilde M}=(x^2-\omega\, x+\epsilon)\, C+(\omega+r-x)\, \tilde M\&
\end{tikzcd}
 \]
  In these new variables, the vector space spanned by $\tilde M\oplus H\oplus\tP\oplus \tK\oplus\tJ$ is a subalgebra of $\mathsf{Li}_{\epsilon,\om,r}$ provided $x^2-\omega\, x+\epsilon=0$. The latter admits real solutions provided $\om^2-4\, \epsilon\geq0$. Although this condition excludes the cases  \hyperlink{Table:AE3}{$\mathsf{AE3}_\kappa$} and \hyperlink{Table:AE5}{$\mathsf{AE5}_{\alpha_+<2,\kappa}$}, setting $r=\lambda-\omega+x$ allows to recover all other $\mathsf{G}$-ambient Klein pairs as sub-Klein pairs.
   \item[bargmann]
 
  Let us perform the change of variables: 
  \[
 \begin{tikzcd}[row sep = small,column sep = small,ampersand replacement=\&]
M\mapsto\tilde M:=M+x\, C\&\tP\mapsto\tilde\tP:=\tP-x\, \tD\&\tK\mapsto \tilde \tK:=\tK+\tD
\end{tikzcd}
 \]
 
 yielding:
  \[
     \hspace{-2.2cm}
\begin{tikzcd}[row sep = small,column sep = small,ampersand replacement=\&]
\mathsf{Li}_{\epsilon,\om,r}\&\br{\tD}{\tilde\tK}= C\&\br{\tD}{\tilde\tP}= \tilde M-x\, C\&\br{C}{H}= \tilde M-(r+x)\, C\&\br{H}{\tD}=r\, \tD\&\br{\tK}{\tilde\tP}= \tilde M
\end{tikzcd}
\]
\[
    \begin{tikzcd}[row sep = small,column sep = small,ampersand replacement=\&]
\br{\tilde\tK}{H}=\tilde\tP+(x-r)\tD\&\hspace{-0.3cm}\br{H}{\tilde\tP}=\epsilon\, \tilde\tK+\omega\, \tilde\tP+\big(x\, (\om-r)-\epsilon\big)\tilde\tD\&\hspace{-0.3cm}\br{H}{\tilde M}=(x^2-\omega\, x+\epsilon)\, C+(\omega+r-x)\, \tilde M
\end{tikzcd}
 \]
 In these new variables, the vector space $\tilde M\oplus H\oplus\tilde\tP\oplus\tilde \tK\oplus\tJ$ defines a subalgebra of of $\mathsf{Li}_{\epsilon,\om,r}$ provided
 \[
  \begin{tikzcd}[row sep = small,column sep = small,ampersand replacement=\&]
 r=x\&x^2-\omega\, x+\epsilon=0\, .
 \end{tikzcd}
 \]
 The second equation admits real solutions provided $\om^2-4\, \epsilon\geq0$. This allows to realise all bargmannian Klein pairs as sub-Klein pairs, with the exception of \hyperlink{Table:AE3b}{$\mathsf{AE3b}$} and \hyperlink{Table:AE5b}{$\mathsf{AE5b}_{\alpha_+<2}$} for which no real solution exists.

\end{description}

\setcounter{secnumdepth}{0}
\section{Conclusion}
\setcounter{secnumdepth}{3}

The present work aimed at providing a description \emph{more algebraico} of the ambient approach, as put forward---in its geometric form---by Eisenhart \cite{Eisenhart1928} and Duval \etal \cite{Duval1985}. To this end, we proposed that algebraic realisations of the ambient approach can be collectively subsumed under the notion of \emph{projectable triplets} \UnskipRef{Definition:Projectable triplet}, defined as effective Klein pairs supplemented with a proper ideal such that the quotient Klein pair is also effective hence susceptible of geometric realisations. The effective Klein pair obtained by taking the quotient along the canonical ideal is then said to admit a \emph{lift} into the original projectable triplet. The two paradigmatic examples of such projectable triplets are given by the bargmann \UnskipRef{Example:Bargmann Klein pair} and leibniz \UnskipRef{Example:Leibniz Klein pair} Klein pairs, both of which provides a lift of the galilei kinematical Klein pair \hyperlink{Table:E1}{$\mathsf{E1}$}. These two examples provide two natural avenues for generalisation. 

\medskip
We followed the \emph{bargmannian} route in Section \ref{section:Bargmann and his modern rivals} by classifying \emph{projectable ambient triplets} \UnskipRef{Definition:Projectable ambient triplets}, defined as projectable triplets admitting the `bargmann-like' $\so(d)$-decomposition \eqref{equation:ambient aristotelian algebra}. The classification---as reported in Table \ref{Table:Projectable ambient triplets}---admits a five-fold partition according to the type of invariant metric structure living on the corresponding quotient space. Among these five categories feature two distinct classes of lifts for each of the kinematical galilean Klein pairs, including the torsional pairs unveiled in \cite{Figueroa-OFarrill2018} [see Table \ref{Table:Possible lifts of effective galilean Klein pairs}]. The first class includes the bargmann Klein pair as well as its curved and torsional generalisations, as displayed in \cite{FigueroaOFarrill2022} and studied in a similar context in \cite{FigueroaOFarrill2022e}. A first result of the present work thus consists in integrating such bargmannian pairs into an exhaustive classification allowing to characterise them as the unique ambient Klein pairs admitting a non-degenerate invariant structure, namely a leibnizian structure supplemented with a compatible lorentzian metric \UnskipRef{Definition:Ambient aristotelian metric structures on Klein pairs}. The second class is novel and characterised by a canonical invariant linear form upgrading the underlying leibnizian structure into what we called a \emph{$\mathsf G$-ambient} structure \UnskipRef{Definition:Ambient aristotelian metric structures on Klein pairs} [see Figure \ref{Figure:Hierarchy of ambient geometries}]. While each of the kinematical galilean Klein pairs admits a unique bargmannian lift, our classification reveals for each galilean Klein pair the existence of a one-parameter family of lifts into $\mathsf G$-ambient projectable triplets, the arbitrariness thereof being encoded into the ambient torsion. In a forthcoming companion paper \cite{Morand2023}, we will provide geometric realisations of the two above classes of ambient Klein pairs as isometry algebras of homogeneous ambient manifolds. 

\medskip
Lastly, we explored the \emph{leibnizian} route (albeit via a shortcut) in Section \ref{section:Drink me} by resorting to dimensional reduction and \IW contraction so that to generate lifts of galilean Klein pairs into the maximal class of leibnizian Klein pairs [see Table \ref{Table:Reductive leibnizian lifts of galilean algebras}]. The corresponding projectable triplets admit the $\so(d)$-decomposition \eqref{equation:ambient kinematical algebra} and are of maximal dimension, generalising the leibniz Klein pair \cite{Bekaert2015b} by featuring non-trivial curvature/torsion. Contradistinctly to their bargmannian and $\mathsf G$-ambient counterparts, those leibnizian lifts admit non-reductive generalisations, as exhibited in Table \ref{Table:Leibnizian lifts of galilean algebras}. We commented on the relationships between those and the previously described bargmannian and $\mathsf G$-ambient pairs. This glimpse into the richness of the landscape of leibnizian Klein pairs prompts to integrate the ones displayed in the present work into a comprehensive classification of maximal projectable triplets. Another natural follow-up to the present study would consist in addressing the dual case by classifying Klein pairs allowing to embed kinematical Klein pairs, thus providing an algebraic counterpart of the embedding of Carrollian manifolds into ambient bargmannian \cite{Duval1991} and leibnizian \cite{Bekaert:2014bwa} manifolds. Among such embedding Klein pairs would feature the previously mentioned bargmannian Klein pairs as well as the dual counterpart of $\mathsf G$-ambient geometries---dubbed \emph{$\mathsf C$-ambient}---and introduced \emph{en passant} in Section \ref{section:Climbing up one leg of the table: an ambient perspective on Klein pairs}. We hope to be able to investigate these aforementioned directions in future works.

\section*{Acknowledgements}
\noindent 
The author is grateful to T.~Basile, X.~Bekaert, Y.~Herfray and J.~Figueroa--O'Farrill for stimulating and enlightening discussions. We also express our sincere gratitude to T.~Basile and J.~Figueroa--O'Farrill for their invaluable input and thoughtful comments on an early version of this manuscript. This work was supported by Brain Pool Program through the National Research Foundation of Korea (NRF) funded by the Ministry of Science and ICT \texttt{2018H1D3A1A01030137} and by Basic Science Research Program through the National Research Foundation of Korea (NRF) funded by the Ministry of Education \texttt{NRF-2020R1A6A1A03047877} and \texttt{NRF-2022R1I1A1A01071497}.

\bibliographystyle{authoryear}

\end{document}